\documentclass[final,12pt,authoryear]{elsarticle}

\usepackage{amssymb}
\usepackage{amsmath}
\usepackage{subcaption}
\usepackage{optidef}
\usepackage{booktabs}
\usepackage{multirow}

\allowdisplaybreaks

\usepackage{amsthm} 

\newtheorem{theorem}{Theorem}[section]
\newtheorem{definition}[theorem]{Definition}
\newtheorem{remark}[theorem]{Remark}
\newtheorem{proposition}[theorem]{Proposition}
\newtheorem{corollary}[theorem]{Corollary}

\usepackage[ruled,vlined]{algorithm2e}

\usepackage[colorlinks=true]{hyperref}
\usepackage[margin=1in]{geometry}
\makeatletter
\def\ps@pprintTitle{%
 \let\@oddfoot\@empty
 \def\@oddfoot{\hfil \thepage\hfil}  
}
\makeatother

\begin{document}

\begin{frontmatter}

\title{\textbf{Minimum Covariance Determinant Estimator and\\ Outlier Detection for Interval-valued Data}}

\author[inst1]{Catarina P. Loureiro\corref{cor1}}
\ead{catarinapadrela@tecnico.ulisboa.pt}
\author[inst1]{M. Ros\'ario Oliveira}
\ead{rosario.oliveira@tecnico.ulisboa.pt}
\author[inst2]{Paula Brito}
\ead{mpbrito@fep.up.pt}
\author[inst3]{Lina Oliveira}
\ead{lina.oliveira@tecnico.ulisboa.pt}

\cortext[cor1]{Corresponding author.}

\affiliation[inst1]{organization={CEMAT, Departamento de Matemática, Instituto Superior Técnico, Universidade de Lisboa},
            addressline={\\Av. Rovisco Pais 1}, 
            postcode={1049-001},
            city={Lisbon},
            country={Portugal}}

\affiliation[inst2]{organization={Faculdade de Economia e Gestão, Universidade do Porto, and LIAAD-INESC TEC},
            addressline={\\Rua Dr. Roberto Frias}, 
            postcode={4200-464},
            city={Porto},
            country={Portugal}}

\affiliation[inst3]{organization={CAMGSD, Departamento de Matemática, Instituto Superior Técnico, Universidade de Lisboa},
            addressline={\\Av. Rovisco Pais 1}, 
            postcode={1049-001},
            city={Lisbon},
            country={Portugal}}

\begin{abstract}
Interval-valued data are one of the most common symbolic data types, which enables the preservation of the underlying variability of the data. The interval mean and covariance matrix can be estimated using the barycenter approach based on the Mallows distance. However, as for conventional data, classical estimates can be significantly affected by anomalous data points, frequently present in real-life datasets. To address this problem, we develop a robust alternative which estimates location and scale by extending the Minimum Covariance Determinant estimator to interval-valued data. The algorithm yields a robust Interval-Mahalanobis distance, which can be used to detect anomalous observations based on adaptive cutoff values. Through extensive simulation studies across various contamination levels, we demonstrate that the interval-valued robust estimator consistently outperforms classical methods in covariance matrix estimation and achieves superior outlier detection accuracy. Finally, the applicability and effectiveness of the proposed method are illustrated through real-world datasets.
\end{abstract}


\end{frontmatter}

\section{Introduction}
\label{sec:intro}
In the era of big data, traditional statistical methods often struggle to handle increasingly complex data structures efficiently. Symbolic Data Analysis (SDA) \citep{Diday} has emerged as a powerful framework for dealing with such complexity by representing data in more sophisticated forms. Among these representations, histogram-valued and interval-valued data stand out as particularly useful, as they naturally capture inherent variability in numerical data. They usually emerge from the aggregation (\textit{macrodata}) of individual observations (\textit{microdata}). For instance, daily temperature ranges, price fluctuations, or sensor reading variations can be more meaningfully represented as intervals rather than single values. Other scenarios where symbolic data can be useful include privacy preservation and large sample size. For a comprehensive review of SDA, the reader is referred to \cite{billard_diday_2006} and \cite{SDAoverview}.

In this work, we adopt the ontic view of an interval, where a closed interval represents precise information associated with an objective entity and reflects its intrinsic variability. Under this view, an interval is treated as a realization of a set-valued random variable. By contrast, in fuzzy data analysis, the epistemic view is adopted, where an interval represents uncertainty, when the information of an objective entity takes on a single value in an interval with some level of uncertainty \citep{COUSO2014,Grzegorzewski2017}.

Several methodological approaches have been developed to analyze interval-valued data effectively, like clustering \citep{DESOUZA2004,RODRIGUEZ2022,DESA2025}, discriminant analysis \citep{Queiroz2018,Silva2015}, regression analysis \citep{DIAS2017,Irpino2015regression,LIMANETO2018}, time series \citep{LIN2016,MAIA2008}, principal component analysis \citep{RGSerrao2023,LeRademacher2012}, modelling \citep{LERADEMACHER2011,Brito2012,Beranger2023}, outlier detection \citep{Li2006,Viattchenin2012,DuarteSilva2018}, among others.

Each of these methods relies on different frameworks to represent and analyze interval-valued data. For this work, we adopt the representation and framework proposed by \cite{Oliveira2022,oliveira2024}, which models interval-valued data as tuples consisting of the interval's center, range, and a distribution function characterizing the microdata within the interval. This approach allows for a more nuanced analysis that considers both the macrodata (intervals) and microdata (distribution within intervals). In this framework, \cite{oliveira2024} also derived descriptive statistics such as the barycenter and symbolic covariance matrix using the Mallows distance. This distance, also known as the $L_2$ Wasserstein distance, has been widely adopted in SDA for measuring dissimilarities between histogram-valued and interval-valued observations \citep{Irpino2015}. It takes into account the centers, ranges, and distributions of the microdata, providing a comprehensive measure of proximity.

The increased expressiveness of interval-valued data brings new challenges to statistical analysis, particularly in the presence of anomalous observations. Classical statistical methods for interval-valued data, while effective under ideal conditions, can be severely compromised by outliers. These outlying observations, which deviate significantly from the majority of the data, can substantially distort estimates of location and scale parameters, leading to unreliable analyses and potentially misleading conclusions. The field of robust statistics offers tools to address this vulnerability, with the Minimum Covariance Determinant (MCD) estimator, in particular the FastMCD algorithm \citep{fastMCD}, being one of the most successful approaches for multivariate data. Originally introduced by \cite{Rousseeuw1984,Rousseeuw1985}, the MCD estimator provides robust estimates of location and scatter by identifying a subset of observations with the most compact scatter structure. In addition, the Mahalanobis distance computed using the MCD estimates serves as a robust measure of outlyingness, enabling effective outlier detection. This approach has proven particularly effective in conventional multivariate settings \citep{hubert2018}. 

Outlier detection in interval-valued data is further complicated by the need to consider deviations in both the centers and ranges of the intervals, as well as in their interactions. \cite{DuarteSilva2018} approached outlier detection by collecting the centers and logarithms of the ranges of the intervals into a single multivariate vector. Under a multivariate normal or skew-normal distribution with specific covariance configurations, conventional multivariate outlier detection methods were then applied. While this method leverages existing techniques, it does not fully exploit the unique structure of interval-valued data. 

More recently, \cite{Tian2024} proposed an MCD estimator for interval-valued data based on a different framework, where each interval is represented by its lower and upper bounds. The symbolic covariance estimator \citep{Billard2008} in use assumes a uniform distribution of the microdata, regardless of the dataset in use, thereby not incorporating information about the underlying microdata distribution. In addition, the associated location estimator is defined as the mean of the interval centers, and a weighted Euclidean distance is likewise computed solely from the centers, without accounting for the ranges or the interval-valued structure. Such simplifications may limit the method's ability to detect outliers, especially when shifts occur in the ranges or in the interaction between centers and ranges, potentially also being misleading when the microdata do not follow a uniform distribution.

In this paper, we present a comprehensive framework for robust analysis of interval-valued data that addresses these challenges. We introduce a novel extension of the MCD estimator to interval-valued data, providing robust estimates of location and scatter while accounting for the center, range, and microdata components. Moreover, a robust Interval-Mahalanobis distance derived from the Interval-valued Minimum Covariance Determinant (IMCD) estimator is proposed, along with appropriate cutoff values for outlier detection. Finally, we validate our methodology through simulation studies and real-world applications.

The remainder of this paper is organized as follows. Section \ref{sec:intervaldata} introduces the fundamental concepts of interval-valued data and establishes notation. Section \ref{sec:MCD} presents our extension of the MCD estimator to interval-valued data, while Section \ref{sec:outlier} develops the outlier detection methodology. Section \ref{sec:simulation} presents and discusses simulation results, and Section \ref{sec:applications} demonstrates the methodology on real-world datasets. Finally, Section \ref{sec:conclusion} summarizes our main contributions. The proofs of the theoretical results from Section \ref{sec:MCD}, the IMCD algorithm pseudocode, as well as additional simulation results are provided in \ref{sec:appendix_mcd}, \ref{sec:appendix_mcd_algorithm}, and \ref{sec:appendix_simulations}, respectively.

\section{Interval-valued Data}
\label{sec:intervaldata}
Interval-valued data are one of the most common types of symbolic data, where each observation is typically represented by an interval rather than a single value. We call a \textit{symbolic interval} the object composed of the real-valued interval, defined by its lower and upper bounds (\textit{macrodata}), together with the set of individual points between those bounds (\textit{microdata}). We follow the notation and model proposed by \cite{Oliveira2022,oliveira2024} and represent a symbolic interval by a tuple $\tilde{x}=([a,b],F)$, where $[a,b]\subset\mathbb{R}$ is the real-valued interval corresponding to the macrodata and $F$ is an absolutely continuous distribution function characterizing the microdata.

The set of all real closed and bounded intervals is defined as $\mathbb{IR} = \{[a,b]:\; a,b \in \mathbb{R},\; a \leq b \}$ and the set of $p$-dimensional hyperrectangles, with $p$ $\in \mathbb{N}$, as the Cartesian product, $\mathbb{IR}^p = \{([a_1,b_1],\dots,[a_p,b_p])^\top:\; a_j,b_j \in \mathbb{R},\; a_j \leq b_j, j=1,\dots,p \}$. Alternatively, a real-valued interval $[a,b]\in\mathbb{IR}$ can also be represented by its center $c=(a+b)/2\in\mathbb{R}$ and range $r=b-a\in\mathbb{R}_0^+$, becoming $[c-r/2,c+r/2]$. Additionally, the microdata are normalized to the interval $[-1,1]$ through an affine transformation. Therefore, the symbolic interval can, equivalently, be represented by a new tuple $\tilde{x}=(x,F_U)$, where $x=(c,r)^\top$ corresponds to the macrodata and $F_U$ characterizes the normalized microdata. $F_U$ is the distribution function of an absolutely continuous real-valued latent random variable $U$ with support $[-1,1]$. This distribution function can be directly estimated from the microdata when available. Otherwise, an assumption has to be made, with the continuous uniform distribution being the most common assumption in the literature (see \cite{oliveira2024} for details).

Let $\boldsymbol{U} = (U_1,\ldots,U_p)^\top$ be a real-valued random vector of independent and absolutely continuous random variables with support $\left[-1,1\right]^p$ and distribution function $F_{\boldsymbol{U}}$. An interval-valued random vector can be defined as $\boldsymbol{\mathcal{\tilde{X}}} = (X_1,\ldots,X_p)^\top$ and described by $\boldsymbol{\mathcal{\tilde{X}}}=(\boldsymbol{\mathcal{X}},F_{\boldsymbol{U}})$, where $\boldsymbol{\mathcal{X}}=(\boldsymbol{\mathcal{C}}^\top,\boldsymbol{\mathcal{R}}^\top)^\top$ with $\boldsymbol{\mathcal{C}} = (C_1,\ldots,C_p)^\top$ and $\boldsymbol{\mathcal{R}} = (R_1,\ldots,R_p)^\top$ are the real-valued random vectors of centers and ranges, respectively, such that $\mathbb{P}(\boldsymbol{\mathcal{R}} \geq \boldsymbol{0}) = 1$. Note that a conventional random vector, $\boldsymbol{\mathcal{X}}=\boldsymbol{\mathcal{C}}$, can be seen as a particular case when $\mathbb{P}(\boldsymbol{\mathcal{R}} = \boldsymbol{0}) = 1$ and $\mathbb{P}(\boldsymbol{U} = \boldsymbol{0}) = 1$. In this context, \cite{Oliveira2022} proposed a model linking the macrodata to the microdata. Let $\boldsymbol{V} = (V_1,\ldots,V_p)^\top\in\mathbb{R}^p$ be a real-valued random vector characterizing the microdata within the macrodata of the interval-valued random vector $\boldsymbol{\mathcal{\tilde{X}}}$. The proposed model links the microdata and macrodata by defining $V_j=C_j+U_jR_j/2$, if $\mathbb{P}(R_j=0)=0$, and $V_j=C_j$, if $\mathbb{P}(R_j=0)=1$ and $\mathbb{P}(U_j=0)=1$.

Another crucial concept is the distance between two intervals. In interval-valued data analysis, the Mallows distance is the common choice for this task, being particularly useful as it accounts for the variability and distribution of the microdata within the intervals. Interestingly, the Mallows distance can be expressed in terms of the centers and ranges of the intervals, leading to a somewhat interpretable representation. Under some assumptions on the latent variables, \cite{oliveira2024} proved that the square of the Mallows distance can be expressed as summarized in \autoref{thm:Mallows}. We introduce the notation $\mathrm{diag}(\boldsymbol{v})$ to represent the diagonal matrix whose main diagonal is the vector $\boldsymbol{v}$.

\begin{theorem}[Mallows Distance \protect{\citep[Thm.~3.5]{oliveira2024}}]
    \label{thm:Mallows}
    For $i=1,2$, let $\boldsymbol{\tilde{x}}_i=(\boldsymbol{x}_i,F_{\boldsymbol{U}})$ denote the $i$-th realization of an interval-valued random vector $\boldsymbol{\mathcal{\tilde{X}}}$, where $\boldsymbol{x}_i=(\boldsymbol{c}_i^\top,\boldsymbol{r}_i^\top)^\top$ are the macrodata, $\boldsymbol{c}_i=(c_{i1},\dots,c_{ip})^\top\in\mathbb{R}^p$ are the centers, $\boldsymbol{r}_i=(r_{i1},\dots,r_{ip})^\top\in(\mathbb{R}^{+}_0)^p$ are the ranges, and $\boldsymbol{U}=(U_1,\dots,U_p)^\top$ is a latent random vector of independent random variables supported on $\left[-1,1\right]^p$, with distribution function $F_{\boldsymbol{U}}$, and finite second moments. Then, the squared Mallows distance between these intervals is
    \begin{equation}
        \label{eq:mallows}
        d_\mathrm{M}^2(\boldsymbol{x}_1, \boldsymbol{x}_2) = (\boldsymbol{c}_1-\boldsymbol{c}_2)^\top(\boldsymbol{c}_1-\boldsymbol{c}_2) + (\boldsymbol{r}_1-\boldsymbol{r}_2)^\top\boldsymbol{\Delta}(\boldsymbol{r}_1-\boldsymbol{r}_2) + (\boldsymbol{c}_1-\boldsymbol{c}_2)^\top\boldsymbol{\Psi}(\boldsymbol{r}_1-\boldsymbol{r}_2),
    \end{equation}
    where $\boldsymbol{\Delta}=\mathrm{diag}(\delta_1,\dots,\delta_p)$, with $\delta_j=\mathbb{E}(U_j^2)/4$ and $\boldsymbol{\Psi}=\mathrm{diag}(\mathbb{E}(U_1),\dots,\mathbb{E}(U_p))$.
\end{theorem}
\begin{remark}
    The parameter $\delta_j=\mathbb{E}(U_j^2)/4$, actually, takes values in the interval $[0,1/4]$. Moreover, if the latent random variables, $U_j$, are symmetric and identically distributed, with $\delta=\mathbb{E}(U_j^2)/4\in [0,1/4]$, then the squared Mallows distance can be simplified to $d_\mathrm{M}^2(\boldsymbol{x}_1, \boldsymbol{x}_2) = (\boldsymbol{c}_1-\boldsymbol{c}_2)^\top(\boldsymbol{c}_1-\boldsymbol{c}_2) + \delta(\boldsymbol{r}_1-\boldsymbol{r}_2)^\top(\boldsymbol{r}_1-\boldsymbol{r}_2)$.
\end{remark}

This distance can be used to derive descriptive statistics for interval-valued data. In particular, \cite{Irpino2015} combine the Mallows distance with the barycenter (or Fréchet mean) approach to define the sample barycenter of a set of intervals. The barycenter plays a central role in SDA, providing a representative for the data. It is defined as the interval that minimizes the weighted sum of squared distances to a set of intervals. In the case of uniform latent random variables, \cite{Irpino2006} showed that the sample barycenter is the interval whose center is the mean of the centers and whose range is the mean of the ranges of the set of intervals being considered. Based on the model linking microdata and macrodata, \cite{oliveira2024} generalized this result to any distribution of the latent random variables and derived the corresponding population barycenter. This leads to the following theorem.

\begin{theorem}[Barycenter \protect{\citep[Thm.~4.1]{oliveira2024}}]
    \label{thm:barycenter}
    Let $\boldsymbol{\mathcal{\tilde{X}}}=(\boldsymbol{\mathcal{X}},F_{\boldsymbol{U}})$ be an interval-valued random vector, with $\boldsymbol{\mathcal{X}}=(\boldsymbol{\mathcal{C}}^\top,\boldsymbol{\mathcal{R}}^\top)^\top$, where the centers and ranges, $\boldsymbol{\mathcal{C}}$ and $\boldsymbol{\mathcal{R}}$, are assumed to have finite expected values, $\boldsymbol{\mu}_C$ and $\boldsymbol{\mu}_R$, respectively. Additionally, let $\boldsymbol{U}=(U_1,\dots,U_p)^\top$ be a real-valued latent random vector of independent random variables, independent of $\boldsymbol{\mathcal{X}}$, with distribution $F_{\boldsymbol{U}}$. Then the population barycenter of $\boldsymbol{\mathcal{\tilde{X}}}$ is given by $\boldsymbol{\mu}_B=(\boldsymbol{\mu}_C,\boldsymbol{\mu}_R,F_{\boldsymbol{U}})$.
\end{theorem}

The symbolic covariance matrix is also essential in SDA, summarizing the variability of multivariate interval-valued data. In \cite{oliveira2024}, the authors have shown that the symbolic covariance matrix has a closed expression based on the covariance matrices of the centers, ranges, and cross centers-ranges of the intervals, as well as the latent variables' first and second order moments. The symbolic covariance matrix based on the Mallows distance and the barycenter is given in Corollary~\ref{cor:symbolic_covariance_matrix}, where the notation $[\boldsymbol{A}]_{j\ell}$ represents the entry $(j,\ell)$ of the matrix $\boldsymbol{A}$.

\begin{corollary}[Symbolic Covariance Matrix \protect{\citep[Cor.~4.3]{oliveira2024}}]
    \label{cor:symbolic_covariance_matrix}
    Let $\boldsymbol{\mathcal{\tilde{X}}}=(\boldsymbol{\mathcal{X}},F_{\boldsymbol{U}})$ be an interval-valued random vector, with $\boldsymbol{\mathcal{X}}=(\boldsymbol{\mathcal{C}}^\top,\boldsymbol{\mathcal{R}}^\top)^\top$, where the centers and ranges, $\boldsymbol{\mathcal{C}}$ and $\boldsymbol{\mathcal{R}}$, have covariance matrices, $\boldsymbol{\Sigma}_{CC}$ and $\boldsymbol{\Sigma}_{RR}$, respectively, and $\boldsymbol{\Sigma}_{CR}$ is the covariance matrix between $\boldsymbol{\mathcal{C}}$ and $\boldsymbol{\mathcal{R}}$. Furthermore, let $\boldsymbol{U}=(U_1,\dots,U_p)^\top$ be a real-valued latent random vector of independent random variables, independent of $\boldsymbol{\mathcal{X}}$, with distribution $F_{\boldsymbol{U}}$. Then, the symbolic covariance matrix of $\boldsymbol{\mathcal{\tilde{X}}}$ is
    \begin{equation}
        \label{eq:cov_matrix}
        \boldsymbol{\Sigma}_B=\boldsymbol{\Sigma}_{CC}+\dfrac{1}{4}\boldsymbol{\mathfrak{E}}_{UU}\bullet\boldsymbol{\Sigma}_{RR}+\dfrac{1}{2}\boldsymbol{\Sigma}_{CR}\boldsymbol{\Psi}+\dfrac{1}{2}\boldsymbol{\Psi}\boldsymbol{\Sigma}_{RC},
    \end{equation}
    where $[\boldsymbol{\mathfrak{E}}_{UU}]_{j\ell}=\mathcal{E}(U_j,U_\ell)$, $j\neq \ell$, with $\mathcal{E}(U_j,U_\ell)=\mathbb{E}(F_{U_j}^{-1}(T) F_{U_\ell}^{-1}(T))=\int_0^1 F_{U_j}^{-1}(t) F_{U_\ell}^{-1}(t) \, dt$, $T\sim\mathrm{Unif}(0,1)$, $[\boldsymbol{\mathfrak{E}}_{UU}]_{jj}=\mathbb{E}(U_j^2)$, $\boldsymbol{\Psi}=\mathrm{diag}(\mathbb{E}(U_1),\dots,\mathbb{E}(U_p))$, and $\bullet$ denotes the Schur or entrywise product of matrices.
\end{corollary}
\begin{remark}
    If the latent random variables, $U_j$, are identically distributed and symmetric, then the symbolic covariance matrix can be simplified to $\boldsymbol{\Sigma}_B=\boldsymbol{\Sigma}_{CC}+\delta\boldsymbol{\Sigma}_{RR}$, where $\delta=\mathbb{E}(U_j^2)/4=\mathrm{Var}(U_j)/4$.
\end{remark}

The symbolic covariance matrix is pivotal in the development of multivariate statistical methods, such as Principal Component Analysis and Linear Discriminant Analysis, in the context of interval-valued data. However, the sample estimator of this matrix, obtained by replacing the population parameters with their sample counterparts, is sensitive to the presence of outliers in the data. This sensitivity can lead to misleading results and interpretations. To address this issue, we propose a robust estimator of the symbolic covariance matrix, which is introduced in the next section.

\section{IMCD Estimator}
\label{sec:MCD}
When thinking about a robust estimator for the covariance matrix, the Minimum Covariance Determinant (MCD) estimator is a natural choice. The MCD estimator was introduced by \cite{Rousseeuw1984,Rousseeuw1985} as a robust estimator of location and scatter for multivariate data, but it was only with the introduction of the computationally efficient FastMCD algorithm in 1999 \citep{fastMCD} that it became widely used. The idea behind the MCD estimator is to find a subset of the data that has the smallest determinant of the covariance matrix, under the assumption that this subset will be less likely to include outliers. This subset is then used to compute robust estimates of location and scatter. The MCD estimator is affine equivariant and has a high breakdown point. These properties make it a suitable choice for robust statistics, having been widely used in various fields, including finance, medicine, and chemistry. For more details, see \cite{hubert2018}.

\subsection{Derivation of the \textnormal{IMCD} Estimator}
Taking inspiration from the conventional case, we propose an extension of the MCD estimator to interval-valued data. In order to derive the interval-valued MCD estimator, we first rewrite the symbolic covariance matrix as follows.

\begin{proposition}
    \label{prop:cov_matrix_rewritten}
    Under the conditions of Corollary~\ref{cor:symbolic_covariance_matrix}, the symbolic covariance matrix \eqref{eq:cov_matrix} can be rewritten as
    \begin{equation}
        \label{eq:cov_matrix_rewritten}
        \boldsymbol{\Sigma}_B=\boldsymbol{\Lambda}\boldsymbol{\Sigma}_{2p}\boldsymbol{\Lambda}^\top
        + \frac{1}{4}\boldsymbol{\Xi}\bullet\boldsymbol{\Sigma}_{RR},
    \end{equation}
    where $\boldsymbol{\Lambda}=\begin{bmatrix}
                \boldsymbol{I}_p & \boldsymbol{\Psi}/2
            \end{bmatrix}\in\mathbb{R}^{p\times2p}$, $\boldsymbol{\Sigma}_{2p}=\begin{bmatrix}
                \boldsymbol{\Sigma}_{CC} & \boldsymbol{\Sigma}_{CR} \\
                \boldsymbol{\Sigma}_{RC} & \boldsymbol{\Sigma}_{RR}
        \end{bmatrix}$ is the $2p\times 2p$ covariance matrix of the centers and ranges, $\boldsymbol{I}_p$ is the identity matrix of dimension $p$, $\boldsymbol{\Psi}=\mathrm{diag}(\mathbb{E}(U_1),\dots,\mathbb{E}(U_p))$, and $\boldsymbol{\Xi}$ is a $p\times p$ matrix whose elements are defined by
        \begin{equation}
        \label{eq:cov_xi}
            [\boldsymbol{\Xi}]_{j\ell}=\widetilde{\mathrm{Cov}}(U_j,U_\ell)=\begin{cases}
                \mathrm{Var}(U_j), &\mathrm{if} \quad j=\ell,\\
                \mathrm{Cov}(F_{U_j}^{-1}(T), F_{U_\ell}^{-1}(T)), &\mathrm{if} \quad j\neq \ell,
            \end{cases}
        \end{equation}
        with $T\sim\mathrm{Unif}(0,1)$.
\end{proposition}

\begin{proof}
    The proof of this proposition is given in \ref{sec:proof_cov_matrix_rewritten}.
\end{proof}

In the interval-valued data setting, we consider that data are stored in a matrix $\boldsymbol{X} = (\boldsymbol{x}_1,\dots,\boldsymbol{x}_n)^\top$, where $\boldsymbol{x}_i\in\mathbb{IR}^p$, for $i=1,\dots,n$, are $n$ multivariate interval-valued observations. Suppose that the centers and ranges of the intervals are stored in the data matrices $\boldsymbol{C} = (\boldsymbol{c}_1,\dots,\boldsymbol{c}_n)^\top\in\mathbb{R}^{n\times p}$ and $\boldsymbol{R} = (\boldsymbol{r}_1,\dots,\boldsymbol{r}_n)^\top\in(\mathbb{R}_0^+)^{n\times p}$, respectively. Then, we have $\boldsymbol{x}_i=(\boldsymbol{c}_i^\top,\boldsymbol{r}_i^\top)^\top$ and $\boldsymbol{X} = [\boldsymbol{C}, \boldsymbol{R}]\in\mathbb{R}^{n\times2p}$. Throughout this work, we also assume that the parameters of the latent variables' distributions are known. In practice, these parameters can be estimated from the microdata when available or assumed based on prior knowledge.

Following the idea of the conventional MCD estimator, we want to find a subset of $m$ observations with the smallest global variability, which is measured by the determinant of the covariance matrix. Let $\boldsymbol{z}=(z_1,\dots,z_n)^\top\in[0,1]^n$ be a weight vector indicating which observations are ``active'', where $z_i=1$ corresponds to inclusion of observation $i$ in the subset and $z_i=0$ to exclusion. Then, assuming $\sum_{i=1}^nz_i = m$, the sample barycenter can be written as a function of $\boldsymbol{z}$ resulting in
\begin{equation}
    \label{eq:barycenter_z}
    \overline{\boldsymbol{x}}_B(\boldsymbol{z})=\frac{1}{m}\boldsymbol{X}^{\top}\boldsymbol{z}=\left(\left(\frac{1}{m}\boldsymbol{C}^{\top}\boldsymbol{z}\right)^\top,\left(\frac{1}{m}\boldsymbol{R}^{\top}\boldsymbol{z}\right)^\top\right)^\top:=\left(\overline{\boldsymbol{c}}(\boldsymbol{z})^\top,\overline{\boldsymbol{r}}(\boldsymbol{z})^\top\right)^\top.
\end{equation}
Likewise, the sample symbolic covariance matrix of the ``active'' observations can be written as a function of $\boldsymbol{z}$:
\begin{equation}
    \label{eq:cov_matrix_z}
    \boldsymbol{S}_B(\boldsymbol{z})=\boldsymbol{\Lambda}\left(\frac{1}{m}\sum_{i=1}^nz_i\boldsymbol{x}_i\boldsymbol{x}_i^\top-\overline{\boldsymbol{x}}_B(\boldsymbol{z})\overline{\boldsymbol{x}}_B(\boldsymbol{z})^\top\right)\boldsymbol{\Lambda}^\top+\frac{1}{4}\boldsymbol{\Xi}\bullet\left(\frac{1}{m}\sum_{i=1}^nz_i\boldsymbol{r}_i\boldsymbol{r}_i^\top-\overline{\boldsymbol{r}}(\boldsymbol{z})\overline{\boldsymbol{r}}(\boldsymbol{z})^\top\right).
\end{equation}

We are now able to define the optimization problem that yields the IMCD estimates: 
\begin{mini}|l|[2]
    {\boldsymbol{z}}{\log\det\left(\boldsymbol{S}_B(\boldsymbol{z})\right)}{}{\label{eq:opt_prob}}
    \addConstraint{\sum_{i=1}^n z_i = m,\;\; \boldsymbol{z}\in\{0,1\}^n.}{}{}
\end{mini}
In order to solve this optimization problem, we take a closer look at the objective function, which is defined on the domain of subsets such that $\boldsymbol{S}_B(\boldsymbol{z})$ is positive definite. As explained in \cite{convex_book}, an important characteristic of the objective function in the optimization context is whether it is convex or concave. In our case, the objective function is concave, as shown in the following proposition.

\begin{proposition}
    \label{prop:concavity}
    The objective function $g(\boldsymbol{z})=\log\det\left(\boldsymbol{S}_B(\boldsymbol{z})\right)$ is concave in $\boldsymbol{z}$.
\end{proposition}
\begin{proof}
    The proof of this proposition is given in \ref{sec:proof_matrix_concave}.
\end{proof}

Another desirable property of the objective function is the existence of a well-defined gradient, which is crucial for the optimization algorithm used to solve the optimization problem \eqref{eq:opt_prob}. The gradient of the objective function is given in the following proposition.
\begin{proposition}
    \label{prop:gradient}
    Let $g(\boldsymbol{z})=\log\det\left(\boldsymbol{S}_B(\boldsymbol{z})\right)$ be the objective function of the \textnormal{IMCD} estimator defined in \eqref{eq:opt_prob}. Then, for $i=1,\dots,n$,  the $i$-th component of the gradient, $\nabla g(\boldsymbol{z})=\left(\frac{\partial g(\boldsymbol{z})}{\partial z_1},\dots,\frac{\partial g(\boldsymbol{z})}{\partial z_n}\right)^\top$, is given by
    \begin{equation}
        \begin{aligned}
        \label{eq:gradient}
        \frac{\partial g(\boldsymbol{z})}{\partial z_i}&=\frac{1}{m}\Bigg[(\boldsymbol{x}_i-\overline{\boldsymbol{x}}_B(\boldsymbol{z}))^\top\boldsymbol{\Lambda}^\top\boldsymbol{S}_B(\boldsymbol{z})^{-1}\boldsymbol{\Lambda}(\boldsymbol{x}_i-\overline{\boldsymbol{x}}_B(\boldsymbol{z}))-\overline{\boldsymbol{x}}_B(\boldsymbol{z})^\top\boldsymbol{\Lambda}^\top\boldsymbol{S}_B(\boldsymbol{z})^{-1}\boldsymbol{\Lambda}\overline{\boldsymbol{x}}_B(\boldsymbol{z})\\
        &+\frac{1}{4}(\boldsymbol{r}_i-\overline{\boldsymbol{r}}(\boldsymbol{z}))^\top\left(\boldsymbol{\Xi}\bullet\boldsymbol{S}_B(\boldsymbol{z})^{-1}\right)(\boldsymbol{r}_i-\overline{\boldsymbol{r}}(\boldsymbol{z}))-\frac{1}{4}\overline{\boldsymbol{r}}(\boldsymbol{z})^\top\left(\boldsymbol{\Xi}\bullet\boldsymbol{S}_B(\boldsymbol{z})^{-1}\right)\overline{\boldsymbol{r}}(\boldsymbol{z})\Bigg].
        \end{aligned}
    \end{equation}
\end{proposition}
\begin{proof}
    The proof of this proposition is given in \ref{sec:proof_gradient}.
\end{proof}

As shown above, the objective function is concave, which means that its minimization is generally an NP-hard problem. This is where the Majorization-Minorization (MM) algorithm comes in. The MM algorithm \citep{MMalgorithm} is an iterative optimization method that converts a hard optimization problem into a series of easier ones. It does this by finding a surrogate function that majorizes the objective function and is easier to minimize. The idea is then to optimize the surrogate function instead of directly minimizing the objective function. In this way, the MM algorithm iteratively refines the solution estimate.

The Majorization step consists, at each iteration $t$, in constructing the surrogate function, $G(\boldsymbol{z}|\boldsymbol{z}^{(t)})$, which upper bounds the objective function, $g(\boldsymbol{z})=\log\det\boldsymbol{S}_B(\boldsymbol{z})$, and coincides with it at the current estimate $\boldsymbol{z}^{(t)}$. Therefore, for all $\boldsymbol{z}\in\{0,1\}^n$, the surrogate function needs to satisfy $G(\boldsymbol{z}|\boldsymbol{z}^{(t)})\geq g(\boldsymbol{z})$ and $G(\boldsymbol{z}^{(t)}|\boldsymbol{z}^{(t)})=g(\boldsymbol{z}^{(t)})$. In our case, since $g(\boldsymbol{z})=\log\det\boldsymbol{S}_B(\boldsymbol{z})$ is concave and continuously differentiable, it can be upper bounded using its first-order Taylor expansion around the current estimate $\boldsymbol{z}^{(t)}$. The surrogate function is, then, constructed as $G(\boldsymbol{z}|\boldsymbol{z}^{(t)})=g(\boldsymbol{z}^{(t)})+\nabla g(\boldsymbol{z}^{(t)})^\top(\boldsymbol{z}-\boldsymbol{z}^{(t)})$, where $\nabla g(\boldsymbol{z}^{(t)})$ is the gradient of the objective function evaluated at $\boldsymbol{z}^{(t)}$.

The Minorization step consists of minimizing the surrogate function, $G(\boldsymbol{z}|\boldsymbol{z}^{(t)})$. Using the gradient of the objective function \eqref{eq:gradient} obtained in Proposition~\ref{prop:gradient}, at each iteration step $t$, we can compute the new estimate, $\boldsymbol{z}^{(t+1)}$, as
\begin{alignat}{3}
    \boldsymbol{z}^{(t+1)}&=&&\ \underset{\boldsymbol{z}\in\{0,1\}^n,\ \sum_iz_i=m}{\mathrm{arg\,min}}&&\quad g(\boldsymbol{z}^{(t)})+\nabla g(\boldsymbol{z}^{(t)})^\top(\boldsymbol{z}-\boldsymbol{z}^{(t)})\notag
    = \underset{\boldsymbol{z}\in\{0,1\}^n,\ \sum_iz_i=m}{\mathrm{arg\,min}}\quad \nabla g(\boldsymbol{z}^{(t)})^\top\boldsymbol{z}\notag\\
    &=&&\ \underset{\boldsymbol{z}\in\{0,1\}^n,\ \sum_iz_i=m}{\mathrm{arg\,min}}&&\quad\sum_{i=1}^n d_i^2(\boldsymbol{z}^{(t)}) z_i,\label{eq:MM_min}
\end{alignat}
where
\begin{equation}
    \begin{aligned}
    \label{eq:terms}
    d_i^2(\boldsymbol{z}^{(t)})&=\left(\boldsymbol{x}_i-\overline{\boldsymbol{x}}_B(\boldsymbol{z}^{(t)})\right)^\top\boldsymbol{\Lambda}^\top\boldsymbol{S}_B(\boldsymbol{z}^{(t)})^{-1}\boldsymbol{\Lambda}\left(\boldsymbol{x}_i-\overline{\boldsymbol{x}}_B(\boldsymbol{z}^{(t)})\right)\\
    &\quad+\frac{1}{4}\left(\boldsymbol{r}_i-\overline{\boldsymbol{r}}(\boldsymbol{z}^{(t)})\right)^\top\left(\boldsymbol{\Xi}\bullet\boldsymbol{S}_B(\boldsymbol{z}^{(t)})^{-1}\right)\left(\boldsymbol{r}_i-\overline{\boldsymbol{r}}(\boldsymbol{z}^{(t)})\right).
    \end{aligned}
\end{equation}

To find the solution to the minimization problem \eqref{eq:MM_min}, we need to keep in mind that $\boldsymbol{z}$ is restricted to $\{0,1\}^n$ and that the objective function is linear in $\boldsymbol{z}$. Therefore, the solution can be found by sorting the terms defined in \eqref{eq:terms} in ascending order, i.e., $d_{(1)}^2(\boldsymbol{z}^{(t)})\leqslant \dots\leqslant d_{(m)}^2(\boldsymbol{z}^{(t)})\leqslant d_{(m+1)}^2(\boldsymbol{z}^{(t)})\leqslant\dots\leqslant d_{(n)}^2(\boldsymbol{z}^{(t)})$ and selecting the $m$ smallest ones by setting $z_i^{(t+1)}=1$, if $d_i^2(\boldsymbol{z}^{(t)})\leqslant d_{(m)}^2(\boldsymbol{z}^{(t)})$, and $z_i^{(t+1)}=0$, otherwise.

Analogously to the conventional MCD estimation, it can be shown that $d_i(\boldsymbol{z}^{(t)})$ is a distance between observation $\boldsymbol{x}_i$ and the current estimate of the barycenter $\overline{\boldsymbol{x}}_B(\boldsymbol{z}^{(t)})$. Accordingly, the update of $\boldsymbol{z}$ can be interpreted as selecting the $m$ observations closest to the current estimate of the barycenter, which is consistent with the conventional MCD approach. In fact, this Minorization step is equivalent to what is called the concentration step, or, more commonly, the C-step in the FastMCD algorithm \citep{fastMCD}.

\subsection{\textnormal{IMCD} Algorithm Implementation}
The MM algorithm guarantees convergence to a local minimum of the objective function $g(\boldsymbol{z})=\log\det\boldsymbol{S}_B(\boldsymbol{z})$, iterating between the Majorization and Minorization steps for a fixed number of iterations or until the change in the objective function is below a certain threshold. Consequently, the initialization of $\boldsymbol{z}$ can have a significant impact on the final solution. For this reason, we have chosen to follow the same initialization procedure and selective iteration incorporated in the FastMCD algorithm \citep{fastMCD}, given its proven efficiency and effectiveness. This procedure ensures that the randomly drawn subsets have nonsingular covariance matrices by augmenting them until $\det(\boldsymbol{S}_B) > 0$, which, for symmetric covariance matrices, is equivalent to positive definiteness when $m>p$. Starting from such an initialization, subsequent MM updates operate on subsets of size $m>p$ and rely on distances computed from a positive definite covariance matrix, which in practice maintains positive definiteness throughout the iterations. In the rare event that a candidate covariance matrix becomes nearly singular, numerical stability is ensured by resorting to the Moore-Penrose generalized inverse. A detailed pseudocode of the IMCD algorithm is included in \ref{sec:appendix_mcd_algorithm}. 

Suppose the algorithm has converged to a solution $\boldsymbol{z}_\mathrm{final}$. The algorithm first returns what, similarly to the MCD, can be called the raw IMCD estimators of location and scatter, $\overline{\boldsymbol{x}}_{\mathrm{rawIMCD}}=\overline{\boldsymbol{x}}_B(\boldsymbol{z}_\mathrm{final})$ and $\boldsymbol{S}_{\mathrm{rawIMCD}}=\boldsymbol{S}_B(\boldsymbol{z}_\mathrm{final})$, respectively. Following the conventional FastMCD, a one-step reweighting is performed in order to enhance efficiency while preserving robustness \citep{breakdown_point_1991,robust_statistics_book}. Specifically, we define a weight vector $\boldsymbol{w}$ such that $w_i=1$, if $d_i^2(\boldsymbol{z}_\mathrm{final})\leqslant\lambda$, and $w_i=0$, otherwise, and compute the final reweighted IMCD estimators as $\overline{\boldsymbol{x}}_{\mathrm{IMCD}}=\overline{\boldsymbol{x}}_B(\boldsymbol{w})$ and $\boldsymbol{S}_{\mathrm{IMCD}}=\boldsymbol{S}_B(\boldsymbol{w})$. Here, the cutoff value $\lambda$ is obtained based on adjusted boxplots \citep{adjbox} or the farness concept \citep{classmap,Raymaekers2024}. These are further explored in the next section. We have chosen to focus on these two non-parametric methods, in detriment of the conventional chi-squared or F-distribution quantiles, since our squared distances $d_i^2$ are not necessarily chi-squared or F-distributed. Thus, our entire approach is non-parametric and does not rely on any distributional assumptions.

The IMCD estimator inherits certain properties from the conventional MCD estimator, such as the high-dimensional limitation, which restricts its applicability to cases where the number of observations is greater than the number of variables. It also inherits certain limitations from the interval-valued data structure itself, such as the lack of affine equivariance, as discussed in \cite{RGSerrao2023}. This is not specific to the proposed estimator but arises from the constraint $\mathbb{P}(\boldsymbol{\mathcal{R}} \geq \boldsymbol{0}) = 1$, which is inherent to interval-valued random variables. As a consequence, the IMCD estimates may depend on the scaling of the variables. However, this dependence cannot be fully removed by rescaling of centers and ranges, since it originates from the structure of the symbolic covariance matrix itself. This highlights that the notion of affine equivariance requires reconsideration in the context of interval-valued data.

\section{Outlier Detection based on Robust Distance}
\label{sec:outlier}
As a result of the IMCD algorithm, we obtain a distance measure from expression \eqref{eq:terms}, which leads us to propose the following new definition for the interval-valued equivalent of the Mahalanobis distance.

\begin{definition}[Interval-Mahalanobis Distance]
    \label{def:int_mah_dist}
    Let $\boldsymbol{x}=(\boldsymbol{c},\boldsymbol{r},F_{\boldsymbol{U}})$ with $\boldsymbol{c}=(c_1,\dots,c_p)^\top\in\mathbb{R}^p$, $\boldsymbol{r}=(r_1,\dots,r_p)^\top\in(\mathbb{R}^{+}_0)^p$, and $\boldsymbol{U}=(U_1,\dots,U_p)^\top$ a latent random vector of independent random variables supported on $\left[-1,1\right]^p$ with distribution function $F_{\boldsymbol{U}}$. The squared Interval-Mahalanobis distance between $\boldsymbol{x}$ and the barycenter $\boldsymbol{\mu}_B$, as in \autoref{thm:barycenter}, of a population with symbolic covariance matrix $\boldsymbol{\Sigma}_B$, as in Corollary~\ref{cor:symbolic_covariance_matrix}, is defined as
    \begin{equation}
    \label{eq:d2_mah}
    \begin{aligned}
        d_\mathrm{IMah}^2(\boldsymbol{x})&=(\boldsymbol{c}-\boldsymbol{\mu}_C)^\top\boldsymbol{\Sigma}_B^{-1}(\boldsymbol{c}-\boldsymbol{\mu}_C)+\frac{1}{4}(\boldsymbol{r}-\boldsymbol{\mu}_R)^\top\left(\boldsymbol{\mathfrak{E}}_{UU}\bullet\boldsymbol{\Sigma}_B^{-1}\right)(\boldsymbol{r}-\boldsymbol{\mu}_R)\\
        &\quad+(\boldsymbol{c}-\boldsymbol{\mu}_C)^\top\boldsymbol{\Sigma}_B^{-1}\boldsymbol{\Psi}(\boldsymbol{r}-\boldsymbol{\mu}_R).
    \end{aligned}
    \end{equation}
    Alternatively, we have $d_\mathrm{IMah}^2(\boldsymbol{x})=(\boldsymbol{y}-\boldsymbol{\eta})^\top\boldsymbol{H}(\boldsymbol{y}-\boldsymbol{\eta})$, where $\boldsymbol{y}=(\boldsymbol{c}^\top,\boldsymbol{r}^\top)^\top$ represents the macrodata of $\boldsymbol{x}$, $\boldsymbol{\eta}=(\boldsymbol{\mu}_c^\top,\boldsymbol{\mu}_r^\top)^\top$ the macrodata of $\boldsymbol{\mu}_B$, and $\boldsymbol{H}=\begin{bmatrix}
        \boldsymbol{\Sigma}_B^{-1} & \dfrac{1}{2}\boldsymbol{\Sigma}_B^{-1}\boldsymbol{\Psi}\\
        \dfrac{1}{2}\boldsymbol{\Psi}\boldsymbol{\Sigma}_B^{-1} & \dfrac{1}{4}\boldsymbol{\mathfrak{E}}_{UU}\bullet\boldsymbol{\Sigma}_B^{-1}
    \end{bmatrix}$.
\end{definition}

\begin{remark}
    If the latent random variables, $U_j$, are symmetric and identically distributed, with $\delta=\mathrm{Var}(U_j)/4$, then the squared Interval-Mahalanobis distance \eqref{eq:d2_mah} can be simplified to $d_\mathrm{IMah}^2(\boldsymbol{x})=(\boldsymbol{c}-\boldsymbol{\mu}_C)^{\top}\boldsymbol{\Sigma}_B^{-1}(\boldsymbol{c}-\boldsymbol{\mu}_C)+\delta(\boldsymbol{r}-\boldsymbol{\mu}_R)^{\top}\boldsymbol{\Sigma}_B^{-1}(\boldsymbol{r}-\boldsymbol{\mu}_R)$.
\end{remark}

\begin{remark}
    The Mallows distance \eqref{eq:mallows} is a special case of the Interval-Mahalanobis distance \eqref{eq:d2_mah} when $\boldsymbol{\Sigma}_B=\boldsymbol{I}_p$, where $\boldsymbol{I}_p$ is the identity matrix of dimension $p$. Moreover, if the latent random variables $U_j$, $j=1,\dots,p$, are non-degenerate, then $\boldsymbol{H}$ is positive definite, and $d_\mathrm{IMah}(\boldsymbol{x})$ defines a metric, particularly, a weighted Euclidean distance on $\mathbb{R}^p\times(\mathbb{R}_0^+)^p$.
\end{remark}

If the barycenter and the interval covariance matrix are robustly estimated using the IMCD algorithm, then the Interval-Mahalanobis distance is also robust. This robust distance can be used to detect outliers in interval-valued data, similar to how the robust Mahalanobis distance is used for conventional multivariate data. In this setting, the robust distance $d_\mathrm{IMah}(\boldsymbol{x})$ measures how far an observation $\boldsymbol{x}$ is from the IMCD estimated barycenter $\overline{\boldsymbol{x}}_\mathrm{IMCD}$. Therefore, given a cutoff value $\tau$, an observation $\boldsymbol{x}$ is flagged as an outlier if $d_\mathrm{IMah}(\boldsymbol{x}) > \tau$. Following the same strategy as in the reweighting step of the IMCD algorithm mentioned in the previous section, the cutoff value $\tau$ can be chosen based on adjusted boxplots \citep{adjbox} or on the farness concept \citep{classmap,Raymaekers2024}. Once again, we opt for these non-parametric methods, as they do not rely on any distributional assumptions. The choice of the cutoff value is user defined and reflects the tolerated false alarm rate.

The adjusted boxplot extends the traditional boxplot by accounting for skewness in the data. Instead of symmetric whiskers, these are adjusted using a robust skewness measure, the medcouple \citep{medcouple}, denoted by $\mathrm{MC} \in [-1,1]$. Additionally, let $Q_1$ and $Q_3$ denote the first and third quartiles, respectively. The whiskers are defined as the upper and lower bounds of the interval $[Q_1-ke^{-4\mathrm{MC}}(Q_3-Q_1),Q_3+ke^{3\mathrm{MC}}(Q_3-Q_1)]$, if $\mathrm{MC}\geqslant0$, and $[Q_1-ke^{-3\mathrm{MC}}(Q_3-Q_1),Q_3+ke^{4\mathrm{MC}}(Q_3-Q_1)]$, if $\mathrm{MC}<0$, where $k$ is a coefficient that controls the whisker length (typically $k=1.5$). Hence, observations outside this interval, or, in this case, distances higher than these cutoff values are flagged as outliers. This method is suitable for any distribution, including skewed or heavy tailed ones, making it particularly useful in real-world applications where data may not follow a normal distribution. More specifically, we use the \texttt{adjbox} function from the \texttt{robustbase} R-package \citep{robustbase}. Further details can be found in \cite{adjbox}.

The farness concept, on the other hand, is based on the idea of fitting a cumulative distribution function (CDF) to a robust transformation of the observations' distances from the barycenter. This is an adaptation of the farness concept introduced by \cite{classmap} for visualizing classification results, in the sense that we no longer have several classes but rather a single class which corresponds to the regular observations. The idea is to identify the observations that are farthest from the robust location of that class, which are flagged as outliers. Therefore, for our purpose, we fit a CDF to the squared distances $d_\mathrm{IMah}^2(\boldsymbol{x})$ of the observations while accounting for skewness. First, we standardize them using the median and the median absolute deviation (MAD). Next, we apply the robust Yeo-Johnson transformation \citep{yeojohnson2000,Raymaekers2024} from the \texttt{transfo} function in the \texttt{cellWise} R-package \citep{cellWise}, whose goal is to make the distribution of the standardized distances more symmetric and closer to a normal distribution. More precisely, the transformation parameter is estimated via a robust reweighted maximum likelihood estimator \citep{Raymaekers2024}. Another standardization is then carried out using the median and MAD of the transformed distances, resulting in an approximately standard normal distribution. Finally, we obtain the standard normal CDF of these standardized transformed distances, resulting in the final probabilities. These can be seen as outlier scores, with higher values indicating a higher likelihood of being an outlier. Thus, the cutoff value can be chosen based on a desired threshold, such as $0.95$, meaning that observations with a score above $0.95$ will be flagged as outliers.

\section{Simulation Study}
\label{sec:simulation}
A simulation study is used to assess the performance of the proposed IMCD estimator and outlier detection method. 

\subsection{Experimental Setup}
Our experimental design consists of the following factors: number of variables $P\in\{5,20\}$; number of observations $N\in\{500,1000\}$; contamination levels $\epsilon\in\{0,0.05,0.1,0.2\}$; contamination scheme: shift in the means of the first variable only in the centers, $\mu_{C_1}+2$, only in the ranges, $\mu_{R_1}+5$, or in both the centers and ranges simultaneously, $\mu_{C_1}+2$ and $\mu_{R_1}+5$; latent variables: symmetric and identically distributed (i.d.) or generated under a more general asymmetric setting (general case).

For each of the $96$ combinations of simulation parameters, we generate $500$ independent samples. For the covariance matrix estimator, we compare the classic raw estimator (raw.Classic) with three versions of the IMCD estimator: raw (raw.IMCD), reweighted using adjusted boxplots (adjbox.IMCD), and reweighted using farness (farness.IMCD). In addition, we consider four outlier detection methods: the classic Interval-Mahalanobis distance with adjusted boxplot cutoff (adjbox.Classic), the Mallows distance with adjusted boxplot cutoff (adjbox.Classic\_Mallows), the robust Interval-Mahalanobis distance with reweighting and cutoff based on adjusted boxplots with coefficient $k=1.5$ (adjbox.IMCD), and the robust Interval-Mahalanobis distance with reweighting and cutoff based on farness with cutoff value $0.95$ (farness.IMCD). The subset size parameter of the IMCD algorithm is set to $M=\lfloor 0.75N\rfloor$, following the conventional usual choice.

The macrodata are generated from a multivariate normal distribution where, for each interval variable $j=1,\dots,p$, we set $C_j\sim\mathcal{N}\left(\mu_{C_j}=0, \sigma_{C_j}=3j/(4p)\right)$ and $R_j\sim\mathcal{N}\left(\mu_{R_j}=3\right.$, $\left.\sigma_{R_j}=3j/(4p)\right)$. We note that, even in the worst case, the probability of generating a negative range is on the order of $10^{-5}$ and, if it does happen, we take its absolute value. All cross-variable correlations are set to zero, while $\mathrm{Corr}(C_j,R_j)=0.1$ for odd $j$ and $-0.1$ for even $j$. In the case of symmetric and i.d. microdata, we assume that the latent variables, $U_j$, $j=1,\dots,p$, follow a uniform distribution, $U\sim\mathrm{Unif}(-1,1)$. As for the general case, the latent variables, $U_j$, $j=1,\dots,p$, are obtained from the same distributional family, following an asymmetric triangular distribution, $U_j\sim\mathrm{Triang}(-1,1,mo_j)$, with a mode, $mo_j$, generated from a uniform distribution, i.e., $mo_j\sim\mathrm{Unif}(-0.5,-0.2)$.

This setting yields a symbolic covariance matrix that is diagonal for both latent variable scenarios. For instance, in the case of the symmetric and i.d. latent variables with $P=5$, the interval covariance matrix is $\boldsymbol{\Sigma}_B=\mathrm{diag}(0.02,0.10,0.22,0.39,0.61)$, rounded to two decimal places. However, since the IMCD estimator is not affine equivariant, we also considered a scenario in which the first and second variables were correlated. Specifically, $\mathrm{Corr}(C_1,C_2)=\mathrm{Corr}(R_1,R_2)=0.8$, and $\mathrm{Corr}(C_1,R_2)=\mathrm{Corr}(C_2,R_1)=0.1$. The results obtained in this correlated scenario were very similar to those of the uncorrelated case and are therefore omitted from this work.

\begin{figure}[ht]
    \centering
    \begin{subfigure}[t]{0.32\textwidth}
        \centering
        \includegraphics[width=\textwidth]{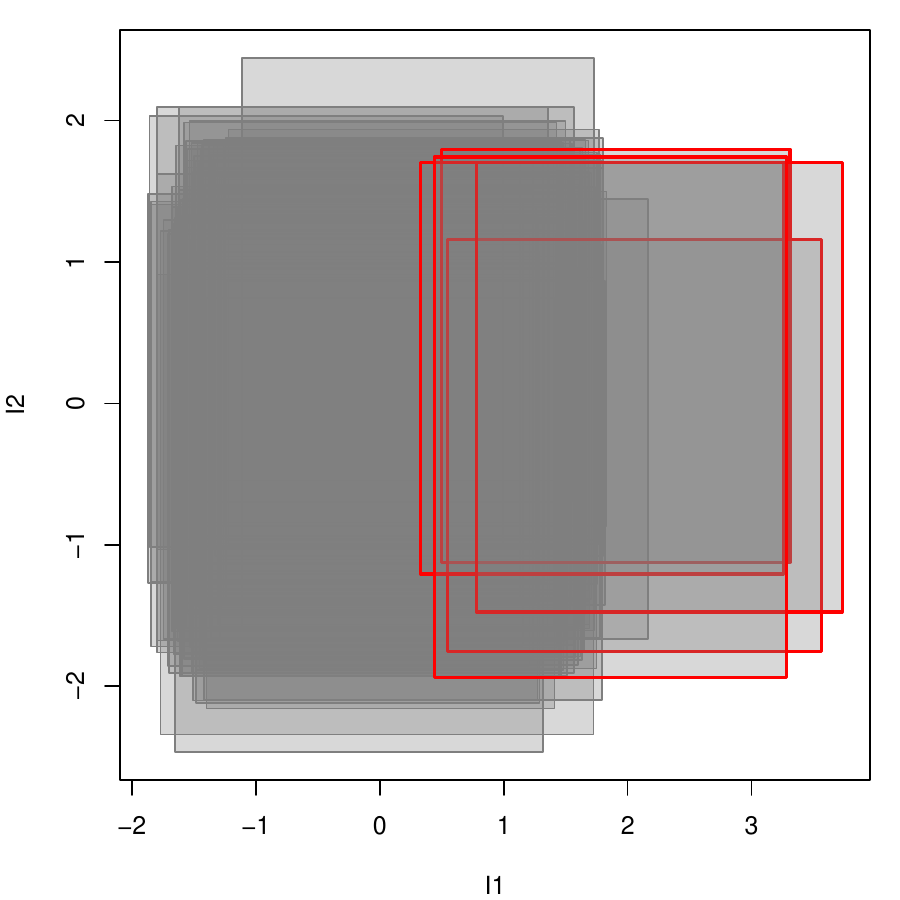}
        \caption{Scenario 1.}
        \label{fig:sim_sample_shift_C}
    \end{subfigure}
    \hfill
    \begin{subfigure}[t]{0.32\textwidth}
        \centering
        \includegraphics[width=\textwidth]{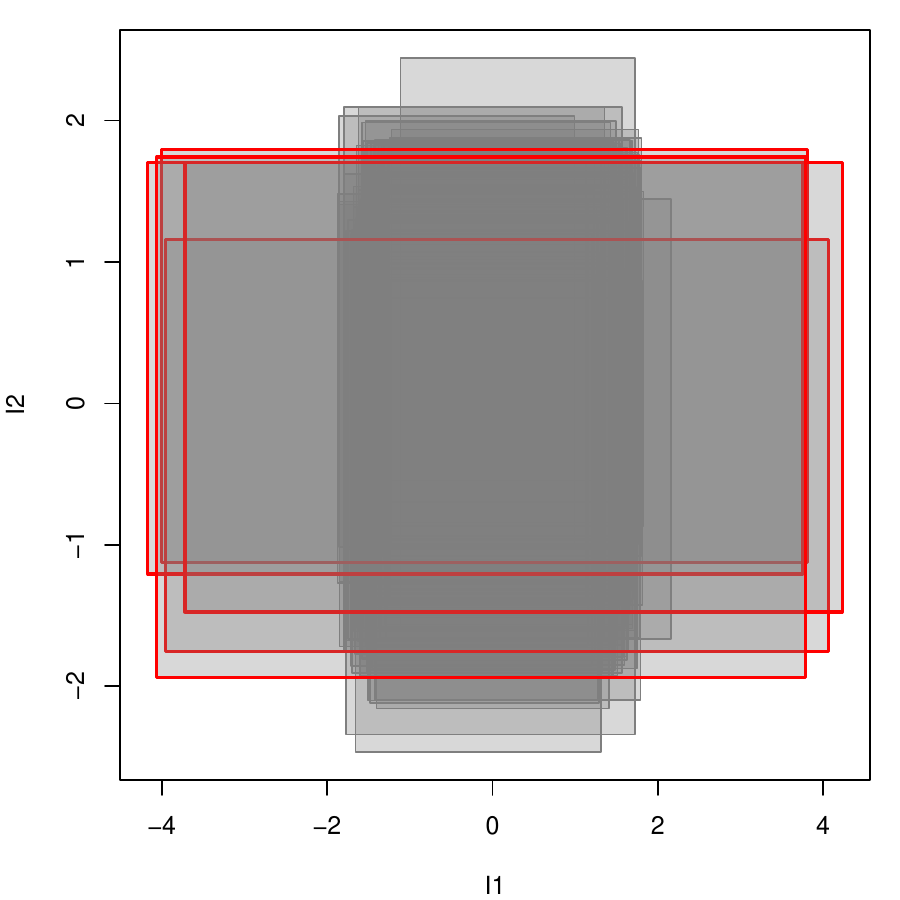}
        \caption{Scenario 2.}
        \label{fig:sim_sample_shift_R}
    \end{subfigure}
    \hfill
    \begin{subfigure}[t]{0.32\textwidth}
        \centering
        \includegraphics[width=\textwidth]{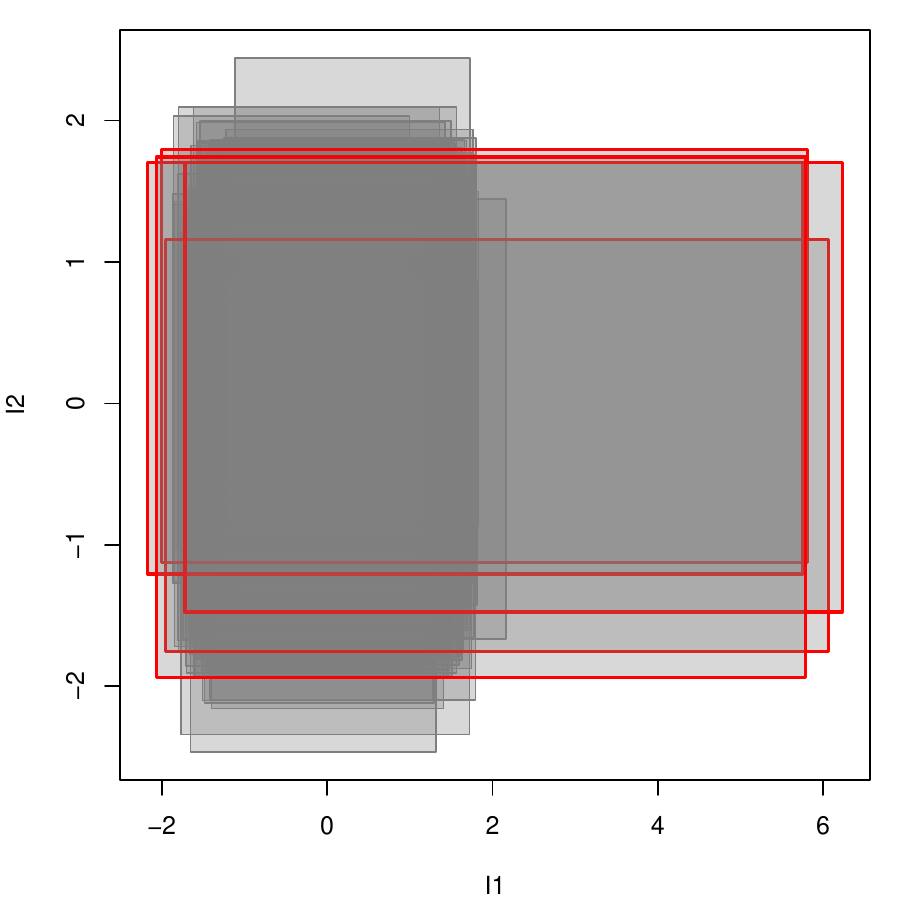}
        \caption{Scenario 3.}
        \label{fig:sim_sample_shift_CR}
    \end{subfigure}
    \caption{Scatter plots of the first two variables from three generated dataset samples with symmetric latent variables and contamination in the mean of: (a) centers, (b) ranges, and (c) both centers and ranges. Outliers are shown in red.}
    \label{fig:sim_sample}
\end{figure}

For simplicity, we define six scenarios by combining the two types of latent variables and the three contamination schemes. Scenarios 1, 2, and 3 correspond to the symmetric and i.d. latent variables case, while scenarios 4, 5, and 6 correspond to the general case. Within each group, contamination affects the centers only, the ranges only, and both the centers and ranges, respectively. \autoref{fig:sim_sample} displays representative scatter plots for samples generated under scenario 1 (\autoref{fig:sim_sample_shift_C}), scenario 2 (\autoref{fig:sim_sample_shift_R}), and scenario 3 (\autoref{fig:sim_sample_shift_CR}). The results for scenarios 3 and 6 were similar to those for scenarios 1 and 4, respectively, since shifts in the centers have a substantially larger impact on the robust distance than shifts in the ranges. For this reason, the results for scenarios 3 and 6 are reported in \ref{sec:appendix_simulations}.

\subsection{Evaluation Metrics}
The performance of the IMCD estimator and the proposed outlier detection method is assessed using several metrics defined as follows. For completeness, additional metrics such as the angle error, Kullback-Leibler (KL) divergence, accuracy, F$_1$-score of the outlier class, and G-mean are reported in \ref{sec:appendix_simulations}.

\paragraph*{Covariance matrix estimation} 
Evaluation is done using the relative Frobenius norm error, a general-purpose measure of reconstruction quality. Considering $\hat{\boldsymbol{\Sigma}}$ and $\boldsymbol{\Sigma}$ are the estimated and ground-truth covariance matrices, respectively, the relative Frobenius norm error is $\|\hat{\boldsymbol{\Sigma}} - \boldsymbol{\Sigma}\|_F/\|\boldsymbol{\Sigma}\|_F = 
    (\sum_{j=1}^{p}\sum_{\ell=1}^{p}([\hat{\boldsymbol{\Sigma}}]_{j\ell}-[\boldsymbol{\Sigma}]_{j\ell})^2)^{1/2}
    (\sum_{j=1}^{p}\sum_{\ell=1}^{p}[\boldsymbol{\Sigma}]_{j\ell}^2)^{-1/2}$.

\paragraph*{Outlier detection} 
Performance measures are functions of: TP (true positives) and TN (true negatives), which denote correctly identified outliers and regular observations, respectively; FP (false positives) and FN (false negatives), which denote incorrect classifications of outliers and regular observations, respectively. We report the estimated precision of the outlier class, which measures the proportion of detected outliers that are correct, $Pr(1) = \mathrm{TP}/(\mathrm{TP}+\mathrm{FP})$; and the estimated recall of the outlier class, measuring the proportion of actual outliers that are detected, $Re(1) = \mathrm{TP}/(\mathrm{TP}+\mathrm{FN})$.

\subsection{Simulation Results}
\autoref{fig:frobenius.error} shows the boxplots of the relative Frobenius error obtained based on the $500$ independent samples for scenarios 1, 2, 4, and 5. The aim is to achieve the smallest relative Frobenius error possible. We can see that for scenarios 1 and 4 the results are very similar, indicating that, for the contamination in the centers, the symmetry of the latent variables does not have a strong impact on the estimation of the covariance matrix. For scenarios 2 and 5, however, we see that the results are slightly worse for the general latent variables case. In scenario 5, the microdata follow triangular distributions instead of uniform distributions, this implies that the microdata are more concentrated around the center of the interval, and thus the ranges play a smaller part in the distance than in scenario 2. Therefore, outliers in the ranges will be harder to identify. In fact, scenarios 2 and 5 present worse results than scenarios 1 and 4, in general. Once again, this can be explained by the robust distance expression, since the $\delta\in[0,\frac{1}{4}]$ parameter ensures that an outlier in the ranges will never affect the distance as much as an outlier in the centers.

\begin{figure}[ht]
    \centering
    \begin{subfigure}[b]{0.49\textwidth}
        \centering
        \includegraphics[width=\textwidth]{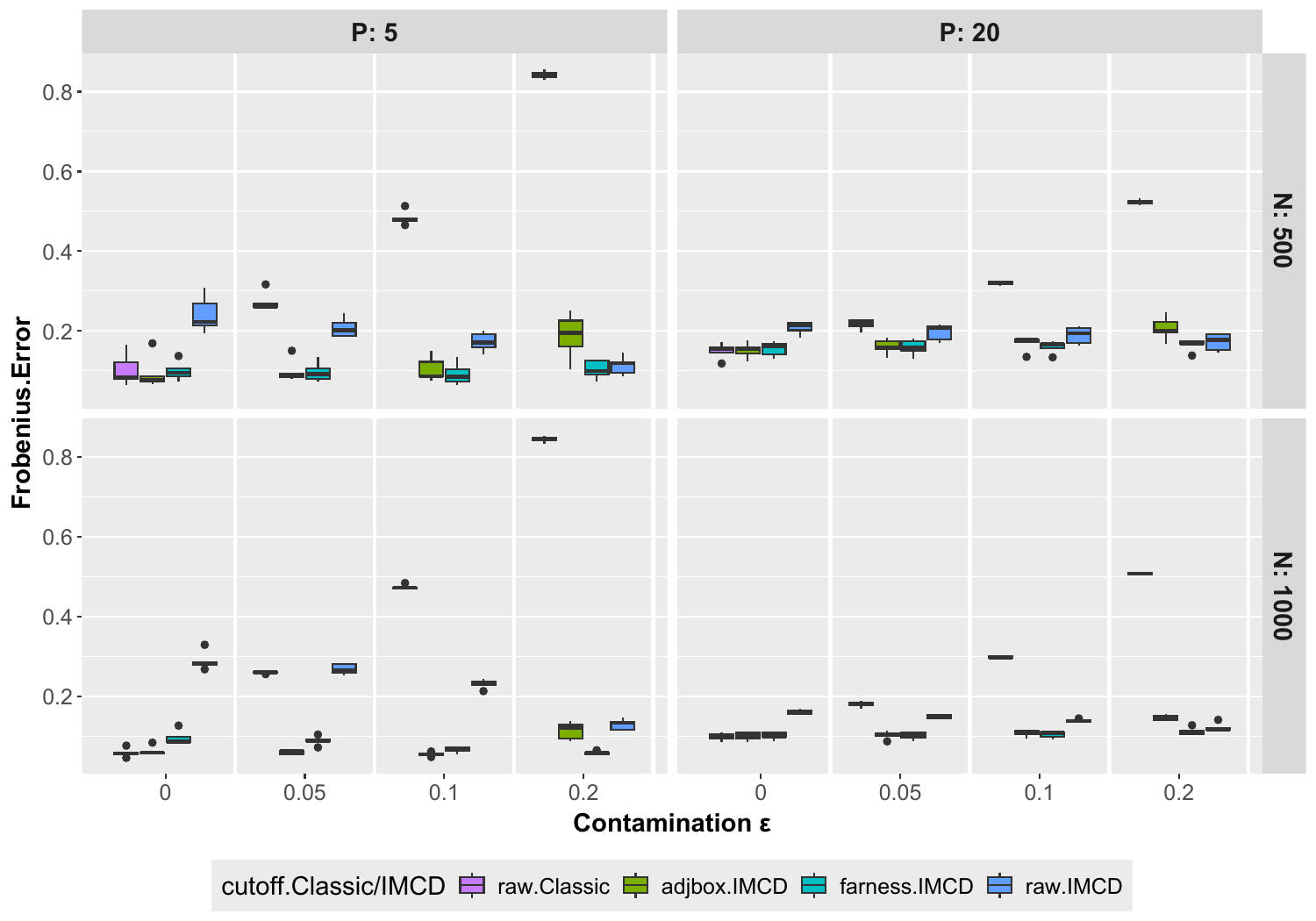}
        \caption{Scenario 1.} 
    \end{subfigure}
    \hfill
    \begin{subfigure}[b]{0.49\textwidth}
        \centering
        \includegraphics[width=\textwidth]{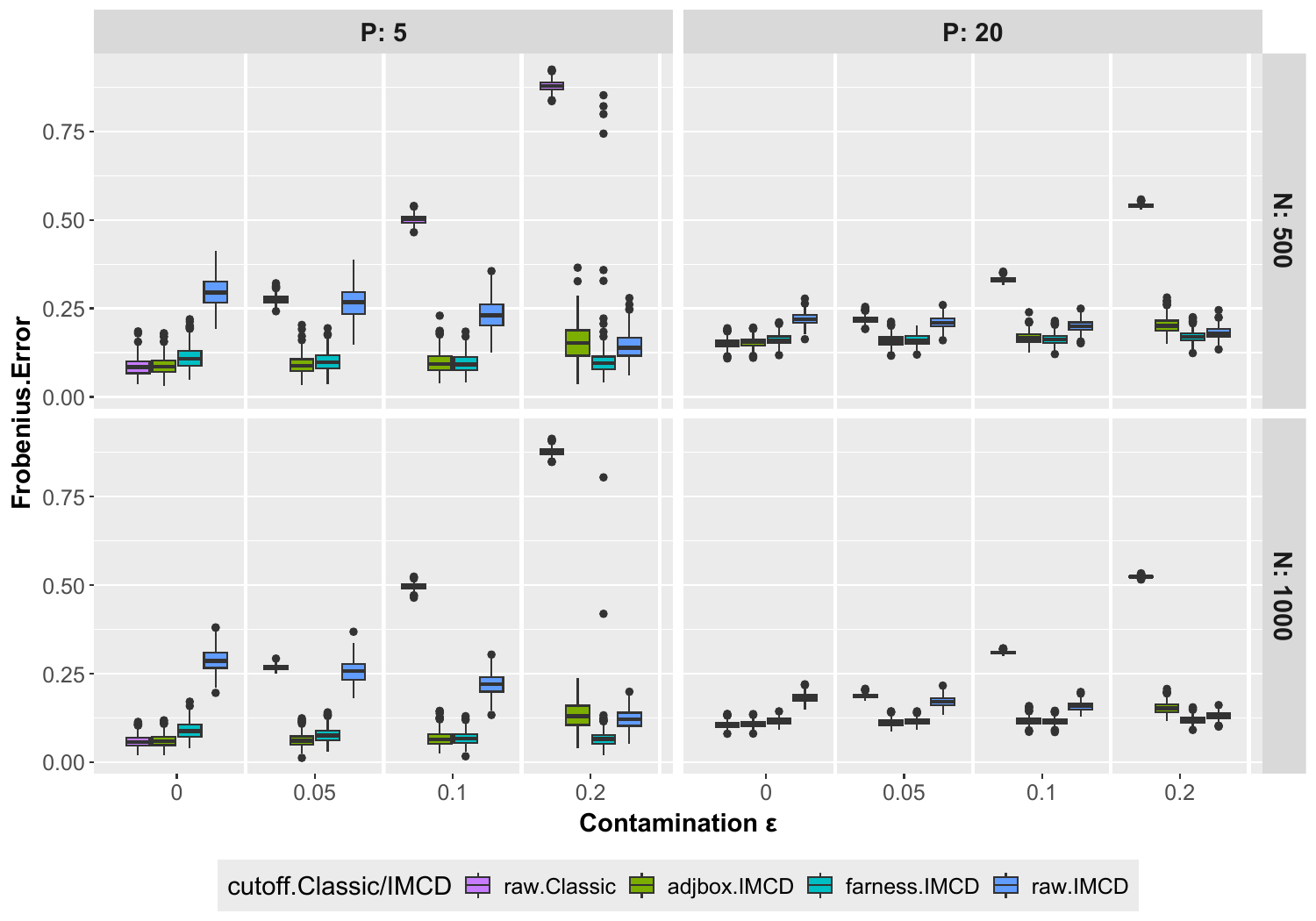}
        \caption{Scenario 4.} 
    \end{subfigure}
    \begin{subfigure}[b]{0.49\textwidth}
        \centering
        \includegraphics[width=\textwidth]{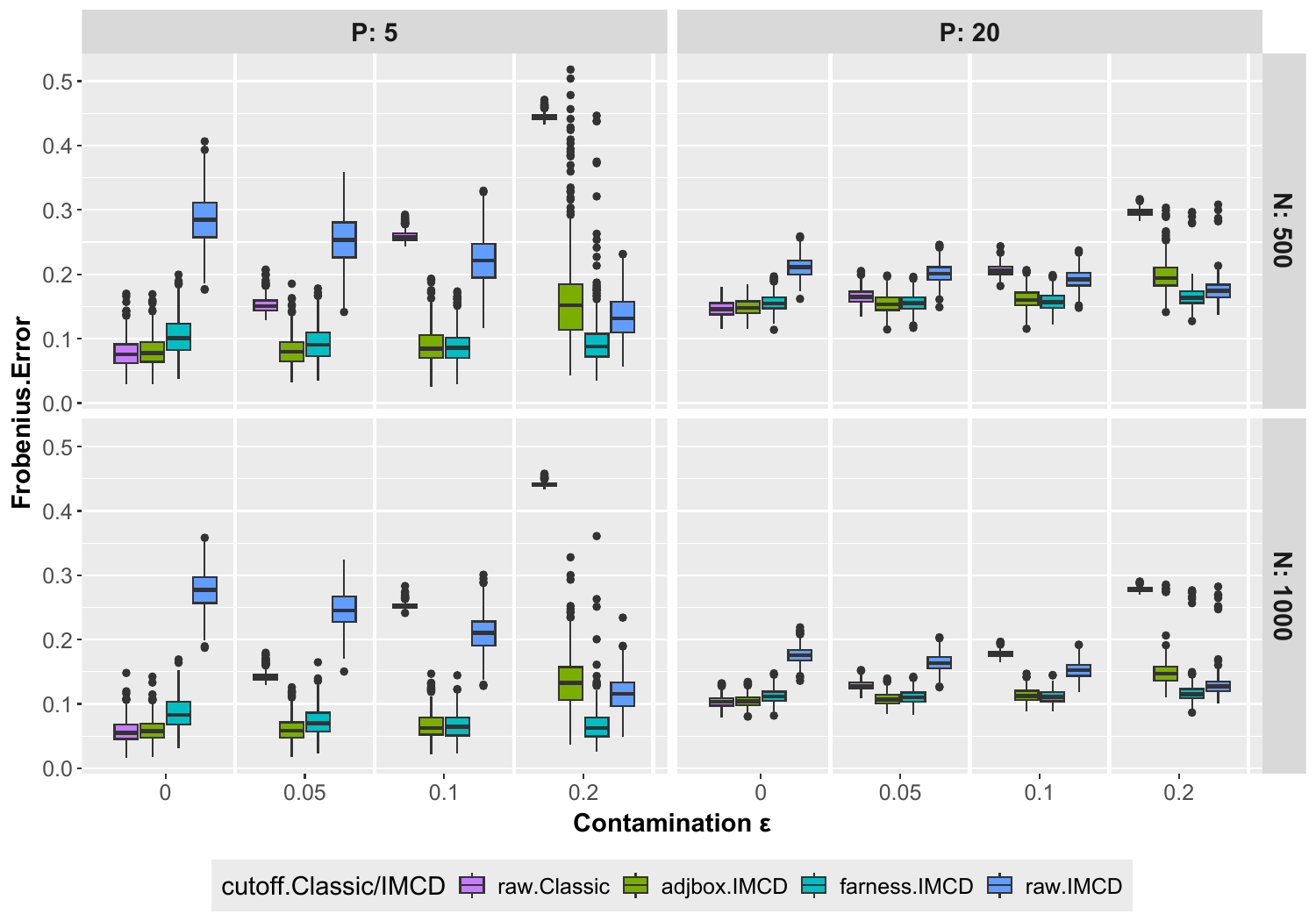}
        \caption{Scenario 2.} 
    \end{subfigure}
    \hfill
    \begin{subfigure}[b]{0.49\textwidth}
        \centering
        \includegraphics[width=\textwidth]{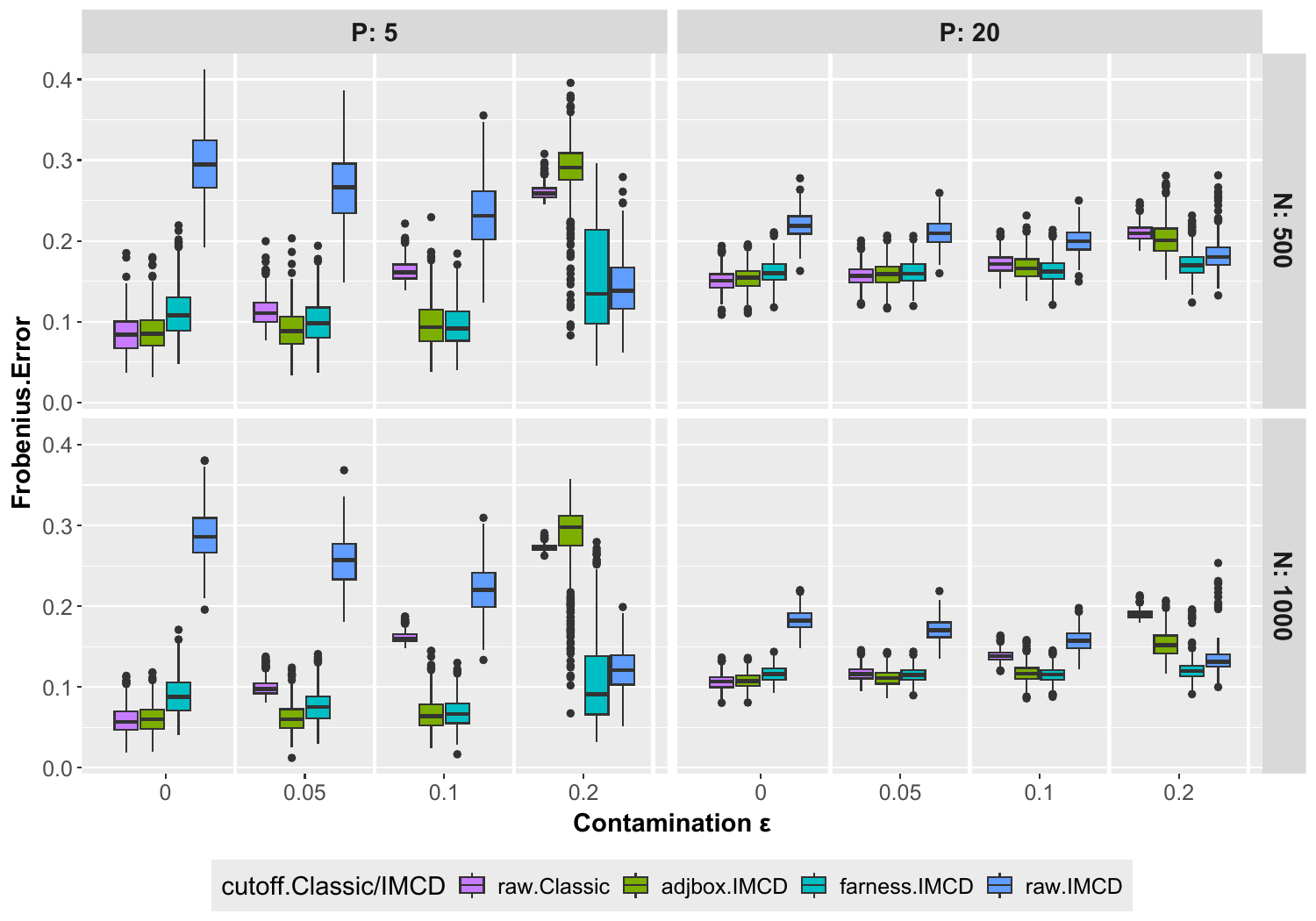}
        \caption{Scenario 5.} 
    \end{subfigure}
    \caption{Boxplots of the relative Frobenius error obtained for scenarios 1, 2, 4, and 5, and the different levels of contamination ($\epsilon$), number of variables ($P$), and sample size ($N$). For each case, we have four covariance matrix estimators: classic without reweighting (raw.Classic), IMCD with adjusted boxplot (adjbox.IMCD) and farness (farness.IMCD) reweighting, and IMCD without reweighting (raw.IMCD).}
    \label{fig:frobenius.error}
\end{figure}

Comparing now the four covariance matrix estimators, we can see that, when we do not have contamination, the classic estimator is the best one. This is to be expected since all the observations are regular ones and are used to estimate the covariance, resulting in a larger sample. As the contamination increases, we see that the classic estimator's error increases, getting further from the ground truth. Looking now to the IMCD estimators, we can see why the reweighting step is so important. The raw IMCD estimator consistently presents worse results than its reweighted counterparts, and in some cases even worse than the classic estimator. As for the reweighted IMCD estimators, they present good results, with the error very close to $0$ in multiple cases. The IMCD estimator with adjusted boxplot reweighting appears to breakdown when the contamination reaches $20\%$, which could be due to it starting to integrate outliers into the distribution. In general, the best estimator is the IMCD with farness reweighting.

In \autoref{fig:recall1}, we present the boxplots of the recall of class 1 (outliers) obtained for scenarios 1, 2, 4, and 5. For this metric, the higher the value, the better. We see similar patterns as in \autoref{fig:frobenius.error} from scenario to scenario. Scenarios 1 and 4 present similar results, although scenario 4 is slightly worse, which could once again be due to the asymmetry of the latent variables. Scenarios 2 and 5 once more show worse results, with scenario 5 being the worst. The same reasons apply here. Additionally, the results show that the classic estimator with adjusted boxplot cutoff has a low recall for every case, being unable to detect any outliers once the contamination is larger than $5\%$. The Mallows distance shows even worse results, basically being unable to detect any outliers even with $5\%$ contamination. As for the IMCD estimators, there seems to be no significant difference between reweighting and cutoff using the adjusted boxplot and using farness. For these two estimators, we have perfect recall for almost every case. The greatest exception is with scenario 4 and $P=5$, both IMCD estimators seem to completely break down for a contamination of $20\%$. This is the scenario where it is hardest to identify outliers, since the $\delta$ parameter is smaller and influences the detection of outliers in the ranges.

\begin{figure}[ht]
    \centering
    \begin{subfigure}[b]{0.49\textwidth}
        \centering
        \includegraphics[width=\textwidth]{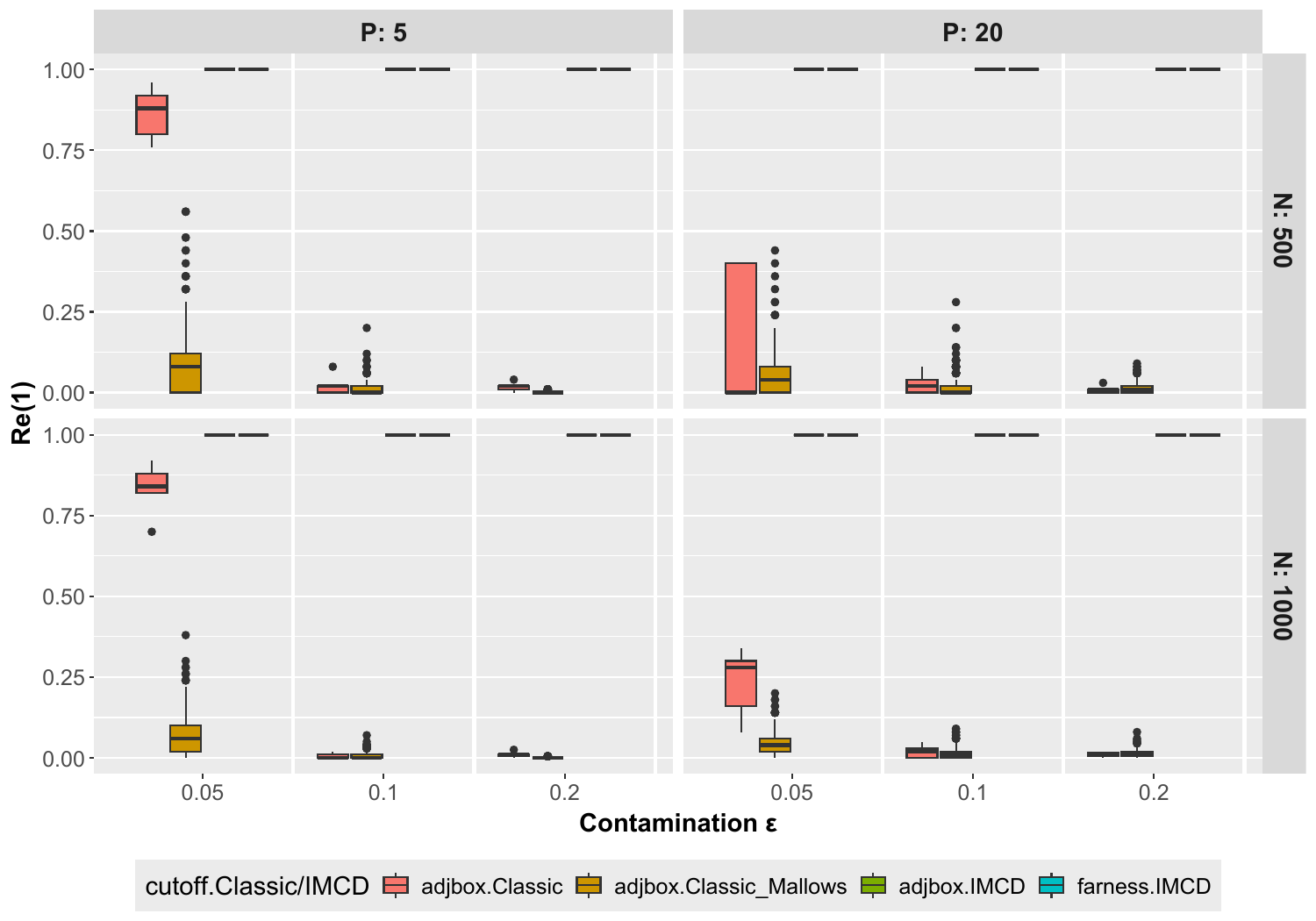}
        \caption{Scenario 1.} 
    \end{subfigure}
    \hfill
    \begin{subfigure}[b]{0.49\textwidth}
        \centering
        \includegraphics[width=\textwidth]{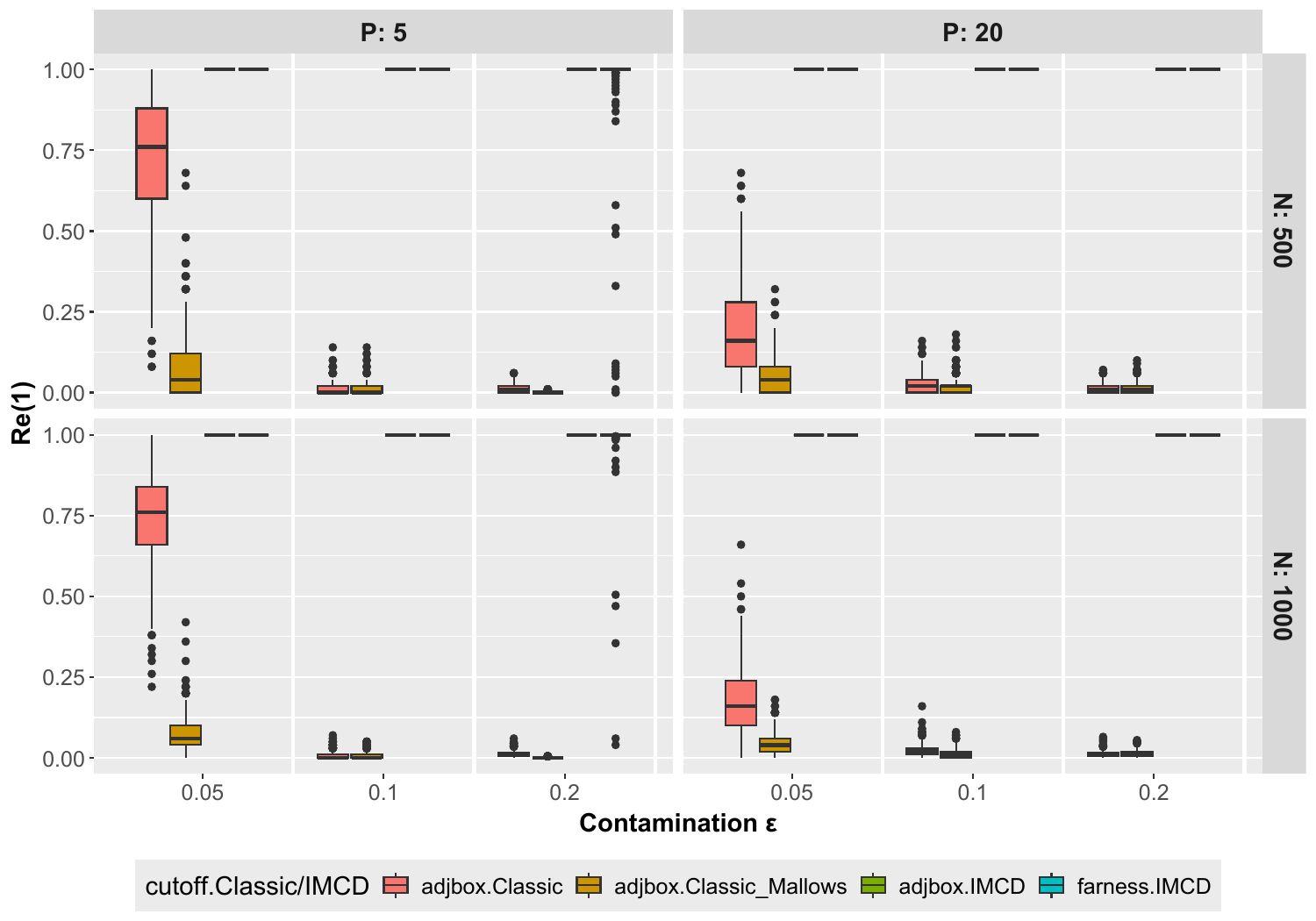}
        \caption{Scenario 4.} 
    \end{subfigure}
    \begin{subfigure}[b]{0.49\textwidth}
        \centering
        \includegraphics[width=\textwidth]{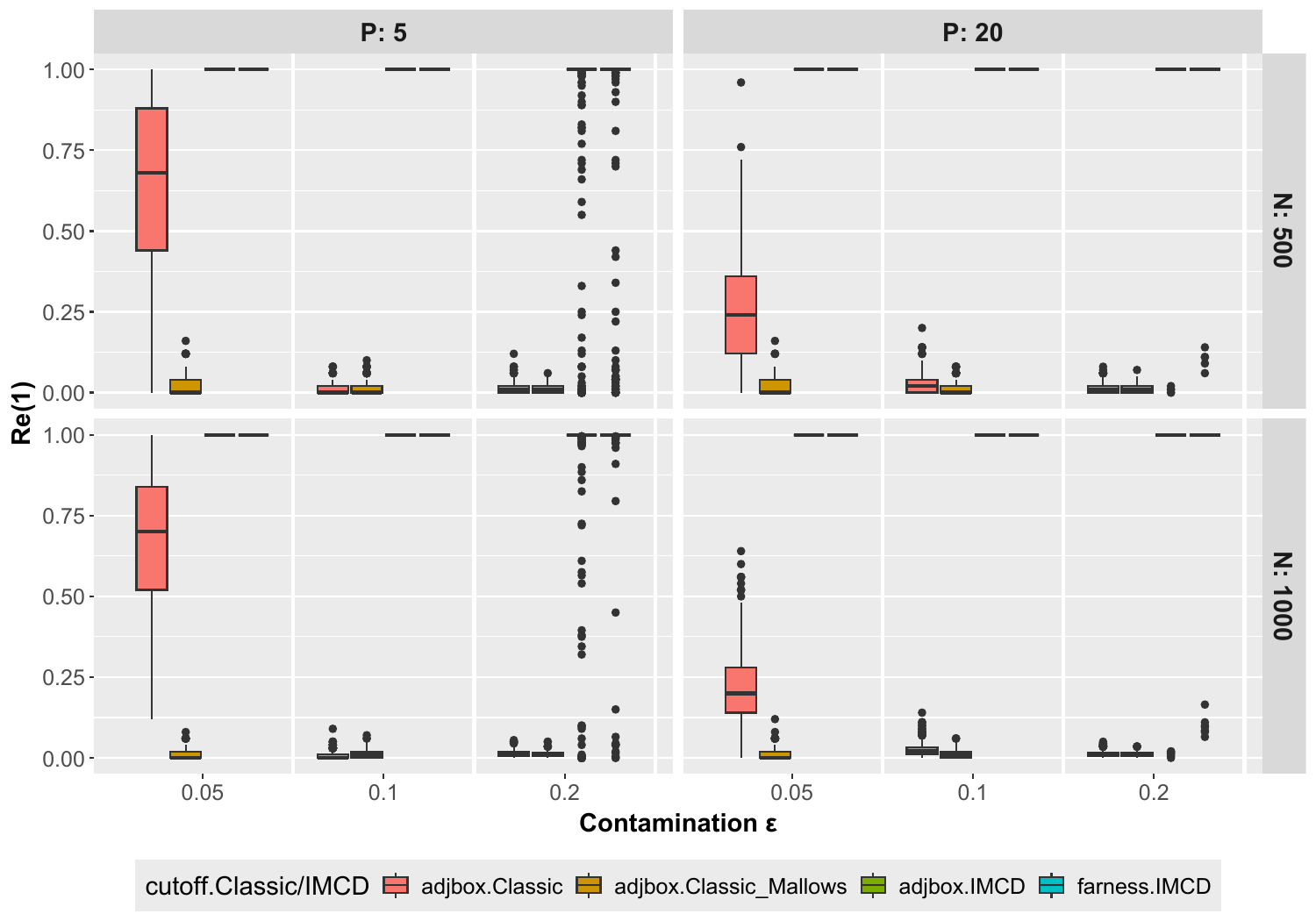}
        \caption{Scenario 2.} 
    \end{subfigure}
    \hfill
    \begin{subfigure}[b]{0.49\textwidth}
        \centering
        \includegraphics[width=\textwidth]{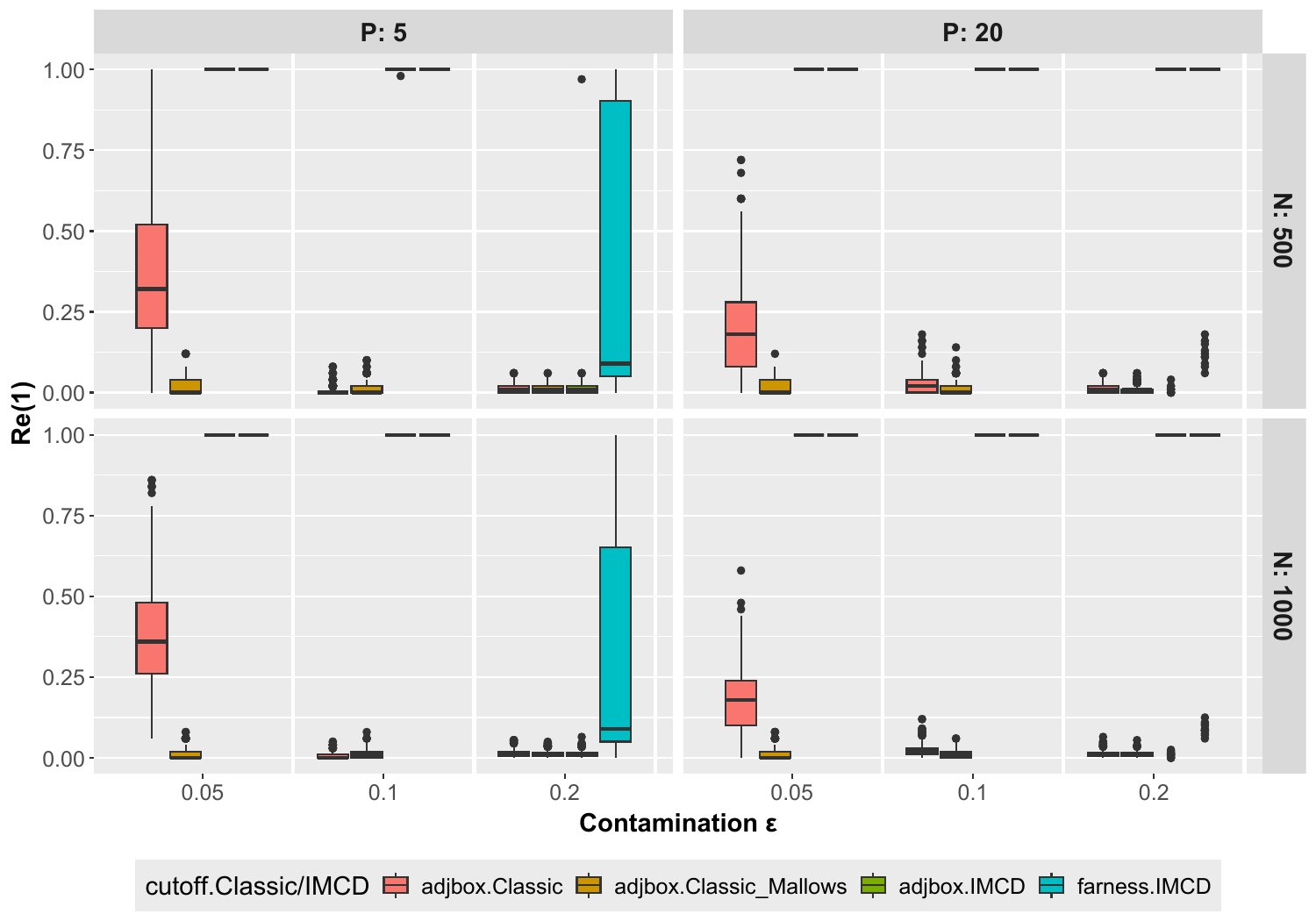}
        \caption{Scenario 5.} 
    \end{subfigure}
    \caption{Boxplots of the recall of class 1 (outliers) obtained for scenarios 1, 2, 4, and 5, and the different levels of contamination ($\epsilon$), number of variables ($P$), and sample size ($N$). For each case, we have four outlier detection methods: classic Interval-Mahalanobis distance with adjusted boxplot cutoff (adjbox.Classic), Mallows distance with adjusted boxplot cutoff (adjbox.Classic\_Mallows), robust Interval-Mahalanobis distance with adjusted boxplot (adjbox.IMCD), and farness (farness.IMCD) reweighting/cutoff.}
    \label{fig:recall1}
\end{figure}

\autoref{fig:precision1} displays the precision of class 1 (outliers) for scenarios 1, 2, 4, and 5. We note that the Mallows distance method has the worst performance in general, with the classic Interval-Mahalanobis distance being a close second. In addition, for the IMCD estimator with the adjusted boxplot reweighting and cutoff, the precision decreases as the contamination increases, while for the farness reweighting and cutoff, it increases. Given the results for the recall of class 1 discussed above, the methods are able to detect the outliers but are identifying false positives as well. Once again we see that the contamination in the ranges results in lower precision than the contamination in the centers. Similarly, the asymmetry of the latent variables also leads to slightly lower precision.

\begin{figure}[ht]
    \centering
    \begin{subfigure}[b]{0.49\textwidth}
        \centering
        \includegraphics[width=\textwidth]{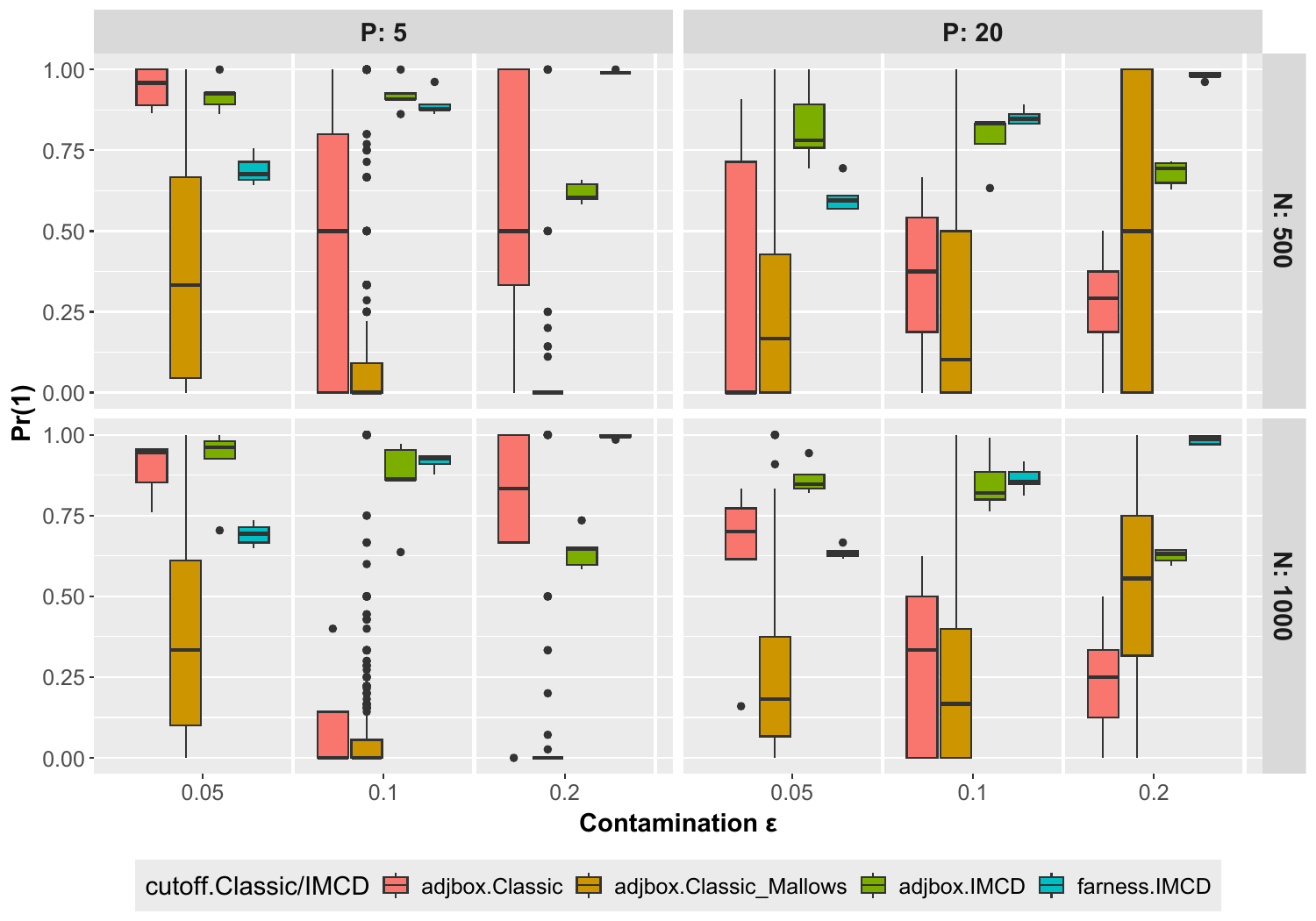}
        \caption{Scenario 1.} 
    \end{subfigure}
    \hfill
    \begin{subfigure}[b]{0.49\textwidth}
        \centering
        \includegraphics[width=\textwidth]{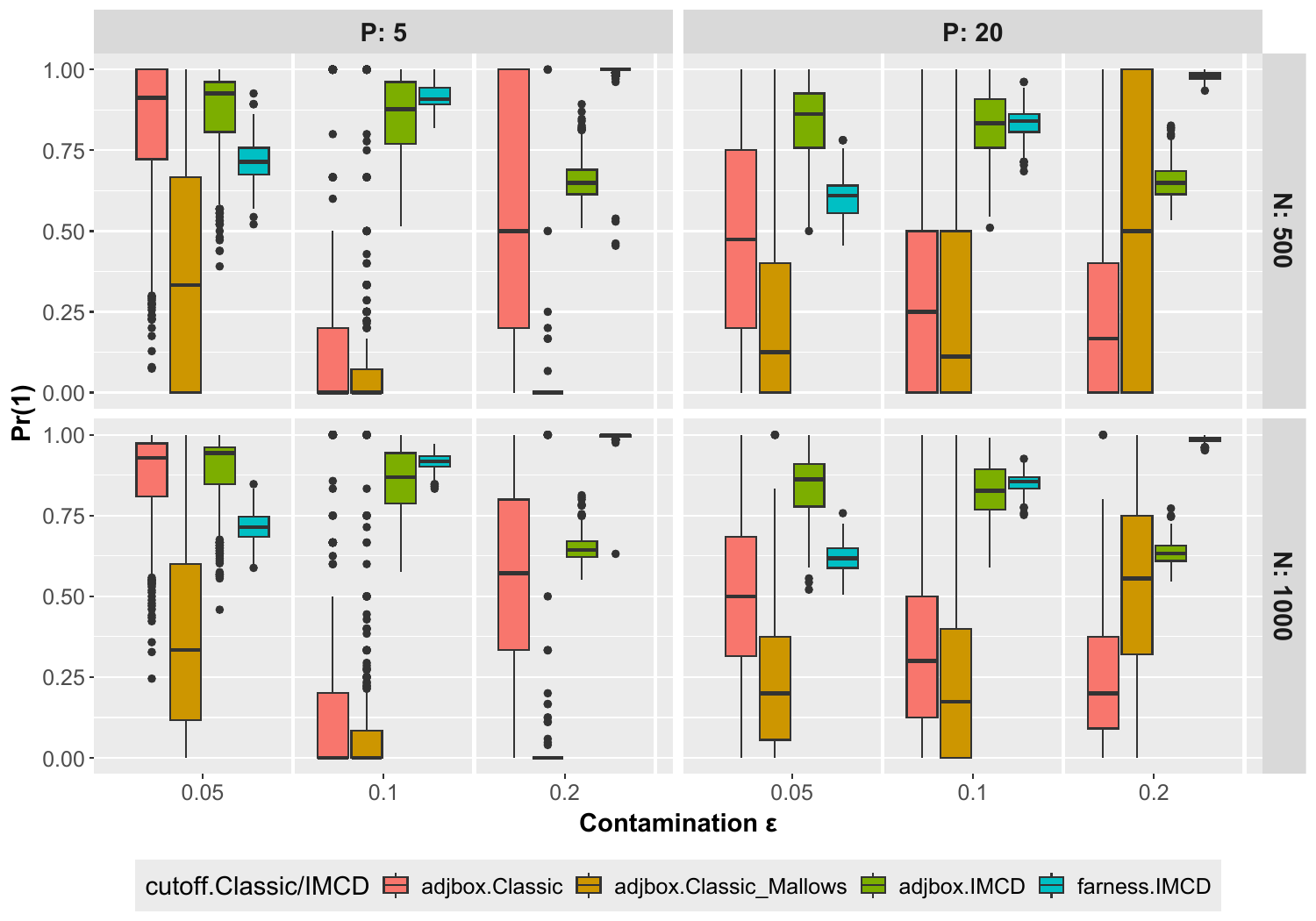}
        \caption{Scenario 4.} 
    \end{subfigure}
    \begin{subfigure}[b]{0.49\textwidth}
        \centering
        \includegraphics[width=\textwidth]{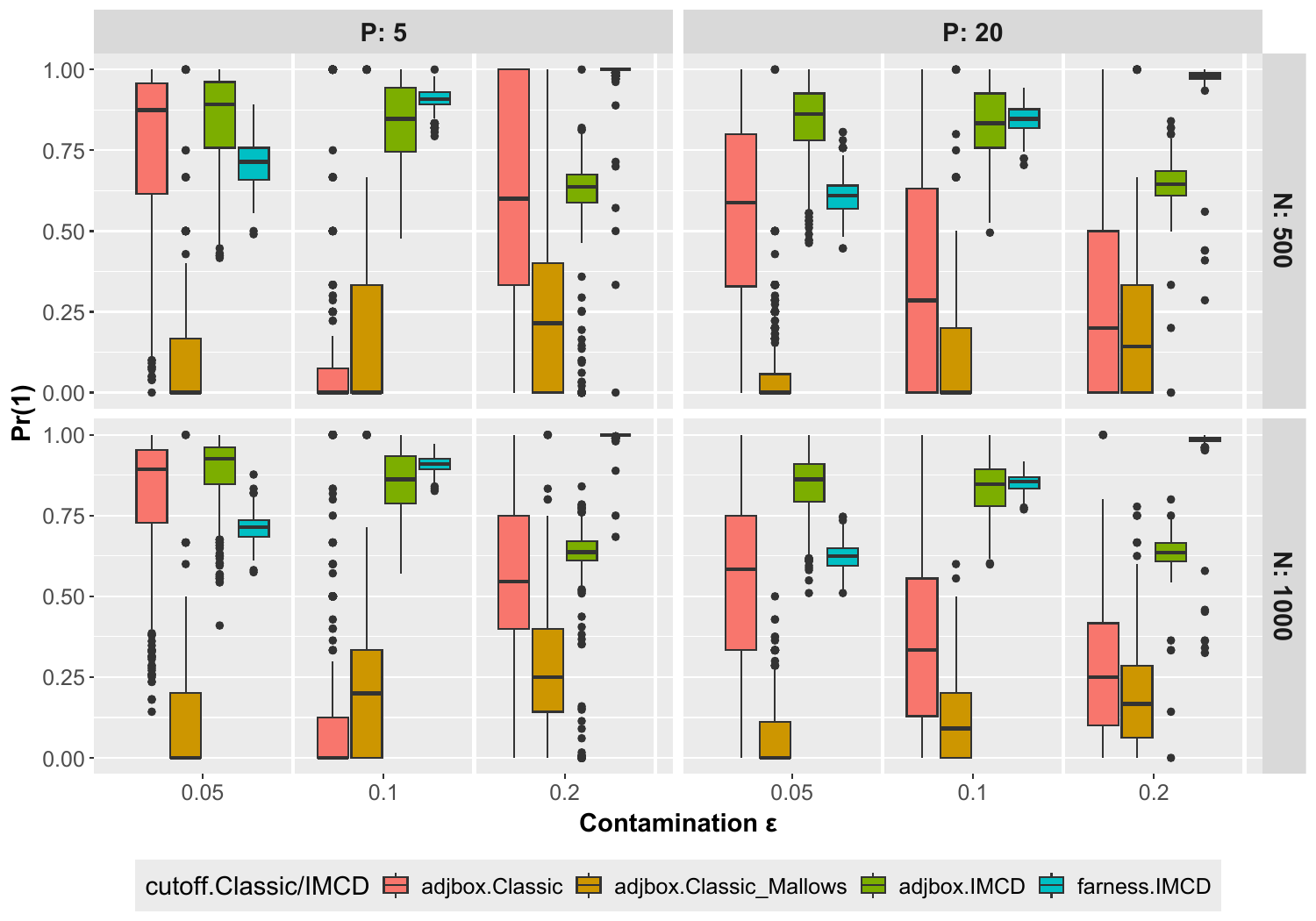}
        \caption{Scenario 2.} 
    \end{subfigure}
    \hfill
    \begin{subfigure}[b]{0.49\textwidth}
        \centering
        \includegraphics[width=\textwidth]{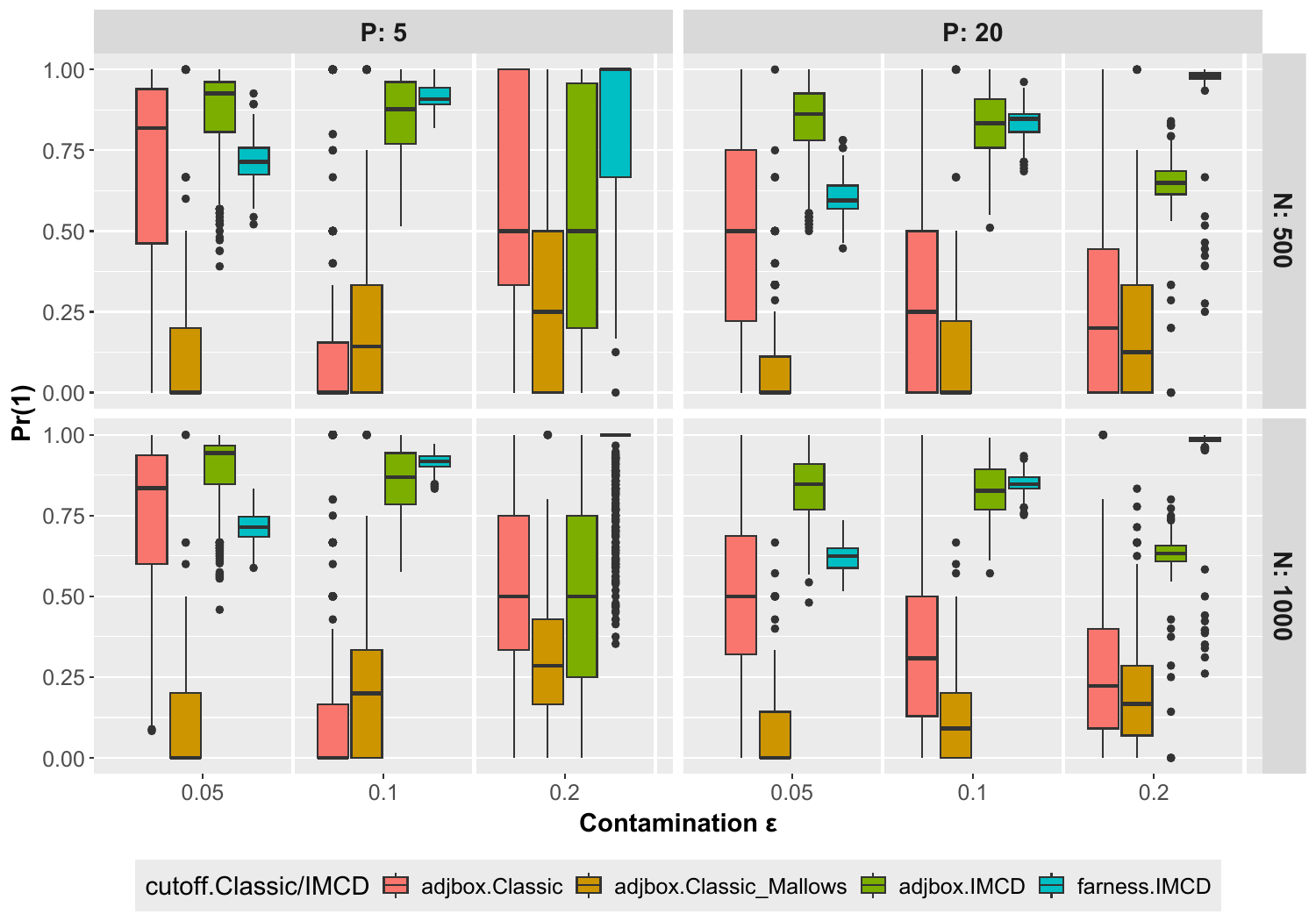}
        \caption{Scenario 5.} 
    \end{subfigure}
    \caption{Boxplots of the precision of class 1 (outliers) obtained for scenarios 1, 2, 4, and 5, and the different levels of contamination ($\epsilon$), number of variables ($P$), and sample size ($N$). For each case, we have four outlier detection methods: classic Interval-Mahalanobis distance with adjusted boxplot cutoff (adjbox.Classic), Mallows distance with adjusted boxplot cutoff (adjbox.Classic\_Mallows), robust Interval-Mahalanobis distance with adjusted boxplot (adjbox.IMCD), and farness (farness.IMCD) reweighting/cutoff.}
    \label{fig:precision1}
\end{figure}

The results for the remaining scenarios are presented in \ref{sec:appendix_simulations}, showing similar patterns as the ones discussed above. Overall, the IMCD estimator with farness reweighting and cutoff consistently outperforms the other methods across almost all scenarios, demonstrating its robustness and effectiveness in outlier detection for interval-valued data.

\section{Applications}
\label{sec:applications}

\subsection{Cars Dataset}
\label{sec:cars}
The Cars dataset consists of $27$ car models characterized by four interval-valued variables: \textit{Price}, \textit{Engine Capacity} (\textit{EngCap}), \textit{Top Speed}, and \textit{Acceleration}. Each car model belongs to one of four classes: \textit{Berlina}, \textit{Luxury}, \textit{Sportive}, and \textit{Utilitarian}. The dataset is available in the \texttt{MAINT.Data} R-package \citep{MAINT.Data}. Since this dataset was previously analyzed for outlier detection in \cite{DuarteSilva2018}, we applied a logarithmic transformation to the \textit{Price} (\textit{lnPrice}) variable to ensure comparability with that study.  As no microdata are available, and following the standard assumption in the literature, we adopted a continuous uniform distribution for the latent variables, corresponding to the symmetric and i.d. case with $\delta=1/12$.

We estimated the sample barycenter and the covariance matrix using the reweighted IMCD algorithm with a subset size of $\lfloor 0.75\times 27\rfloor=20$ and a farness cutoff of $0.9$. This cutoff was chosen to reflect the expected contamination level of about $10\%$ reported by \cite{DuarteSilva2018}. Then, using the robust distances of each observation from the barycenter and the farness cutoff of $0.9$, five observations were flagged as outliers: \textit{MercedesClasseS}, \textit{Ferrari}, \textit{Porsche}, \textit{MercedesSL}, and \textit{HondaNSK}. Notably, except for \textit{MercedesClasseS}, which belongs to the \textit{Sportive} class, the remaining four correspond to the \textit{Luxury} class.

\autoref{fig:cars_pairs} displays the correlation matrix obtained from the IMCD algorithm in the upper panel and scatter plots highlighting the outliers in the lower panel. The variable \textit{lnPrice} is strongly positively correlated with \textit{EngCap} and \textit{Top Speed}, as expected, while \textit{Acceleration} is negatively correlated with all other variables. These relationships are consistent with typical engineering and performance characteristics of car models.

In \autoref{fig:cars_dist_dist}, we compare the classical and robust squared Interval-Mahalanobis distances. For the classical version, we used the whole dataset to estimate the barycenter and covariance matrix, while for the robust version, we used the IMCD algorithm. \autoref{fig:cars_dist_dist} also shows the $0.9$ farness cutoff value for the robust method and the $1.5$ adjusted boxplot cutoff value for the classical method. In addition, the outliers are highlighted in red. While the robust approach clearly separates the five outliers from the remaining observations, the classical distance flags none. This discrepancy underscores the importance of robust methods.

\begin{figure}[ht]
    \centering
    \begin{subfigure}[t]{0.49\textwidth}
        \centering
        \includegraphics[width=\textwidth]{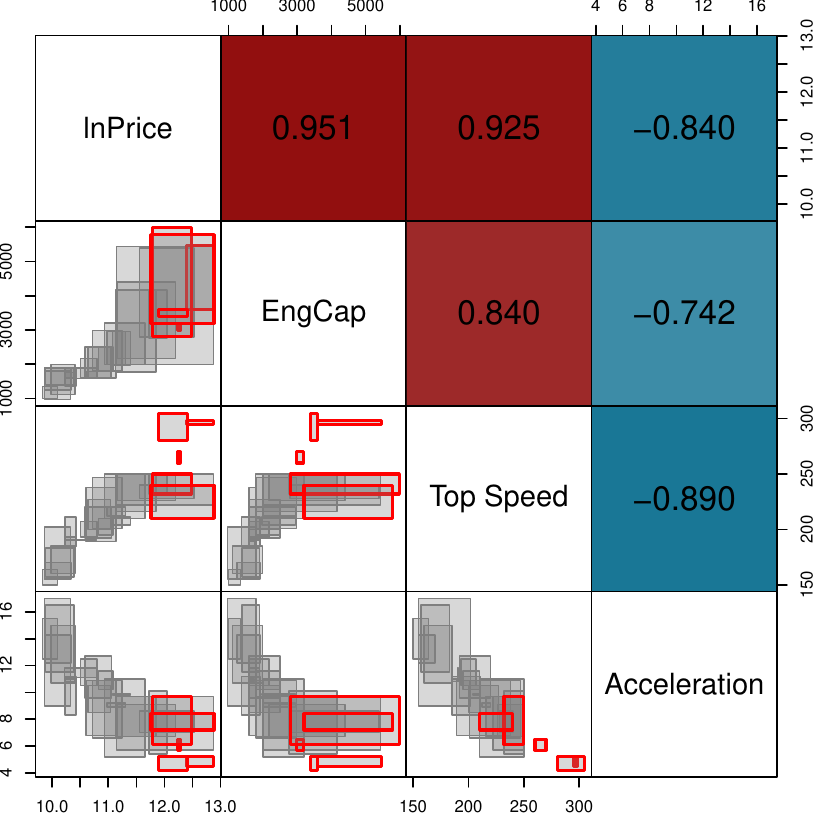}
        \caption{Pairs plot for the Cars dataset.}
        \label{fig:cars_pairs}
    \end{subfigure}
    \hfill
    \begin{subfigure}[t]{0.49\textwidth}
        \centering
        \includegraphics[width=\textwidth]{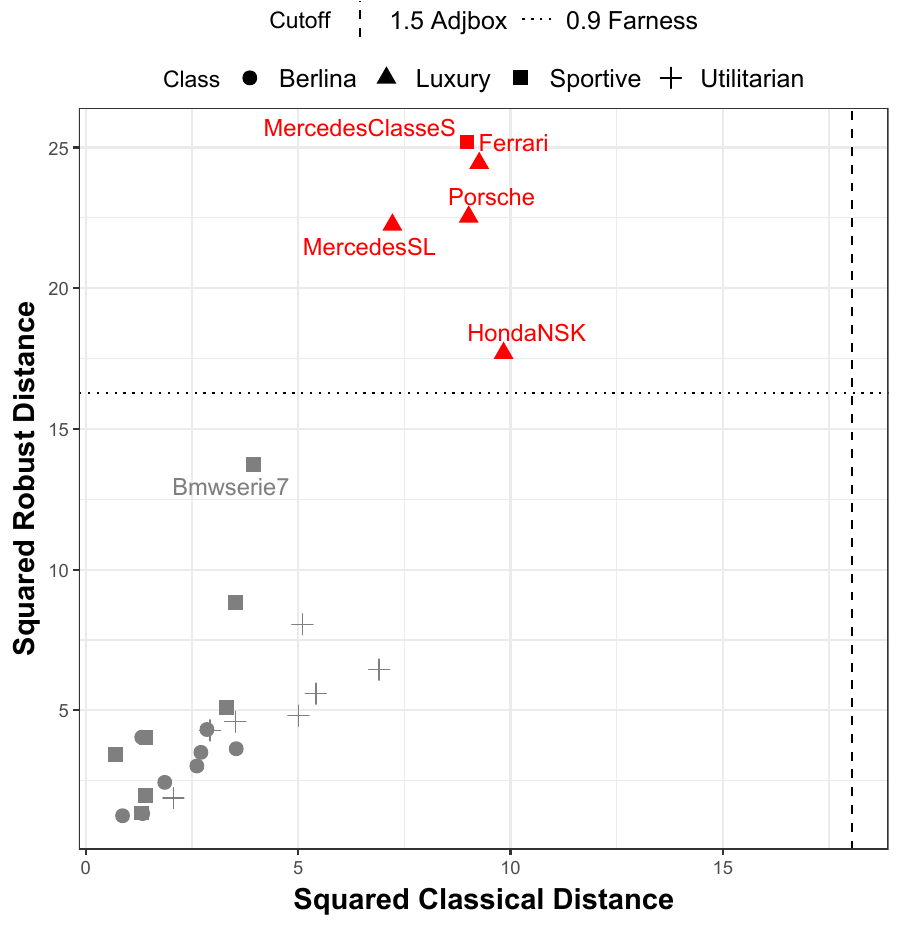}
        \caption{Distance-distance plot for the Cars dataset.}
        \label{fig:cars_dist_dist}
    \end{subfigure}
    \caption{Pairs plot (a) and distance-distance plot (b) of the classical versus robust squared Interval-Mahalanobis distances for the Cars dataset. In the pairs plot (a), the lower triangular shows scatter plots of the four variables, with the outlying observations highlighted in red, while the upper triangular shows the interval correlation matrix. In the distance-distance plot (b), the points' shape represents the car model class, the horizontal and vertical lines correspond to the $0.9$ farness cutoff value and the $1.5$ adjusted boxplot cutoff value, respectively, with the outliers marked in red.}
    \label{fig:cars1}
\end{figure}

Finally, when comparing our findings with those of \cite{DuarteSilva2018}, we observe substantial agreement between the two approaches. In particular, \textit{Ferrari}, \textit{Porsche}, and \textit{MercedesClasseS} are identified as outliers in both studies. \textit{HondaNSK} is also flagged in \cite{DuarteSilva2018}, but only with respect to \textit{lnPrice}, whereas our method detects it as a multivariate interval outlier. Our approach additionally identifies \textit{MercedesSL}, which was not detected in their analysis, while \textit{Skoda Octavia} (and \textit{Alfa 166}, flagged there in relation to \textit{Acceleration}) are not classified as outliers under our framework. These differences likely stem from the distinct modelling strategies adopted. Overall, this comparison suggests that explicitly accounting for the joint contribution of centers, ranges, and underlying latent distributions yields a fully multivariate assessment of outlyingness.

\subsection{Spotify Tracks Dataset}
The Spotify Tracks Dataset is publicly available on Kaggle \citep{kaggle.spotify2022}. The raw dataset contains $114\,000$ observations and $20$ audio and metadata variables retrieved from Spotify in October 2022. However, some tracks appear more than once or have multiple genres. To address this, we removed duplicates and retained the $11$ numerical variables: \textit{duration\_ms}, \textit{popularity}, \textit{danceability}, \textit{energy}, \textit{loudness}, \textit{speechiness}, \textit{acousticness}, \textit{instrumentalness}, \textit{liveness}, \textit{valence}, and \textit{tempo}. Then, using only the tracks with one genre associated, we computed the robust centroid using the MCD estimator for each genre. Next, for the tracks with more than one genre, we selected among those genres the one whose centroid was closest, using the conventional robust Mahalanobis distance. After this process, $81\,033$ tracks remained. We then applied some transformations: \textit{tempo} and \textit{loudness} were logarithmically transformed; \textit{popularity} was rescaled from $[0,100]$ to $[0,1]$; \textit{duration\_ms} was converted from milliseconds into minutes, becoming \textit{duration}. Finally, the data were aggregated by genre into intervals, taking as lower and upper bounds the $1\%$ and $99\%$ quantiles, respectively. This resulted in an interval-valued dataset with $111$ observations and $11$ variables. The number of tracks grouped into each genre ranges from $62$ to $1\,822$ with a median of $800$, where some tracks have missing values across variables.

Unlike in the Cars dataset, microdata are available, enabling estimation of the latent variables' parameters. In particular, $\boldsymbol{\mathfrak{E}}_{UU}$ was estimated via Kernel Density Estimation (KDE) using the R-package \texttt{kde1d} \citep{kde1d}. Then, we applied the proposed reweighted IMCD estimator with a subset size of $\lfloor 0.75\times111\rfloor=83$ and a farness cutoff of $0.95$. The estimated robust correlation matrix is shown in \autoref{fig:spotify_corr}. We observe medium to strong positive correlations among energy-related features (\textit{energy}, \textit{loudness}, and \textit{tempo}) and somewhat meaningful positive associations between \textit{danceability} and \textit{valence}, and between \textit{danceability} and \textit{loudness}. In contrast, \textit{acousticness} is moderately negatively correlated with \textit{energy}, while other features like \textit{speechiness} and \textit{liveness} show correlations close to zero with most features. These relationships are consistent with musical intuition, supporting the validity of the robust correlation estimates obtained from the IMCD algorithm.

\begin{figure}[ht]
    \centering
    \begin{subfigure}[t]{0.49\textwidth}
        \centering
        \includegraphics[width=\textwidth]{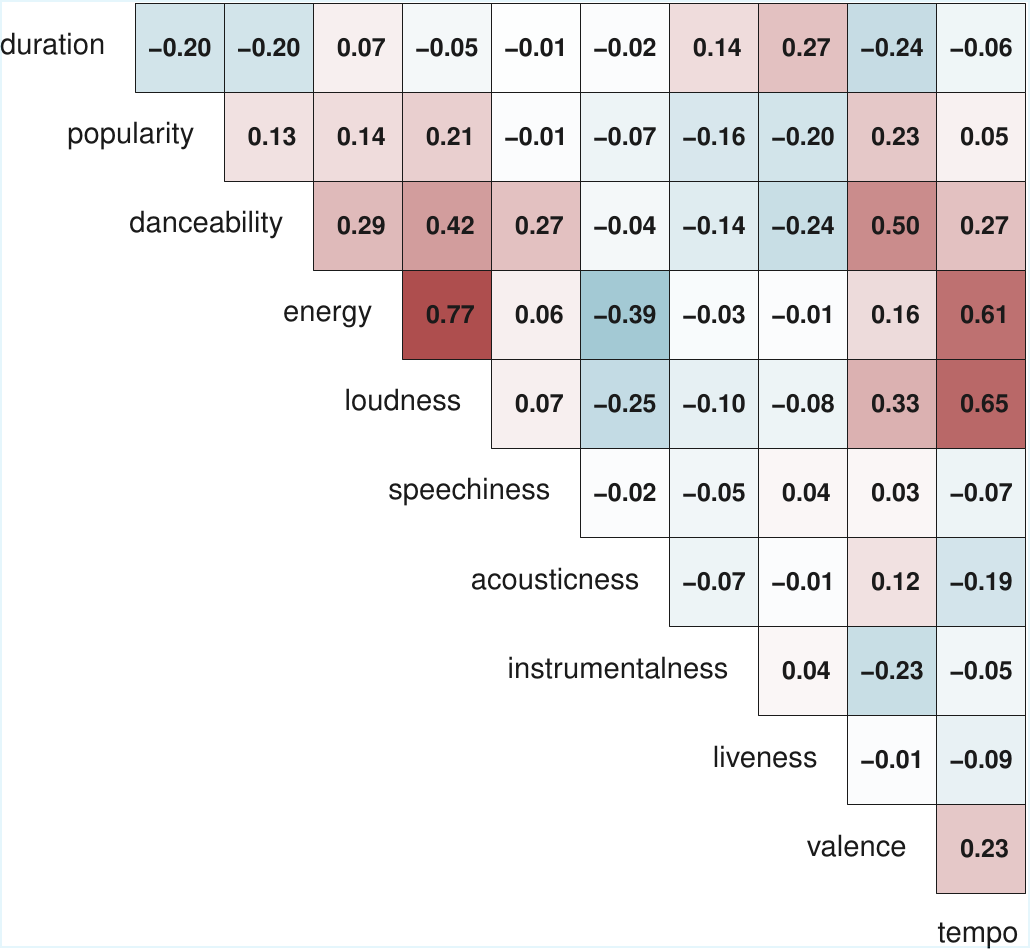}
        \caption{Interval correlation matrix plot.}
        \label{fig:spotify_corr}
    \end{subfigure}
    \hfill
    \begin{subfigure}[t]{0.49\textwidth}
        \centering
        \includegraphics[width=\textwidth]{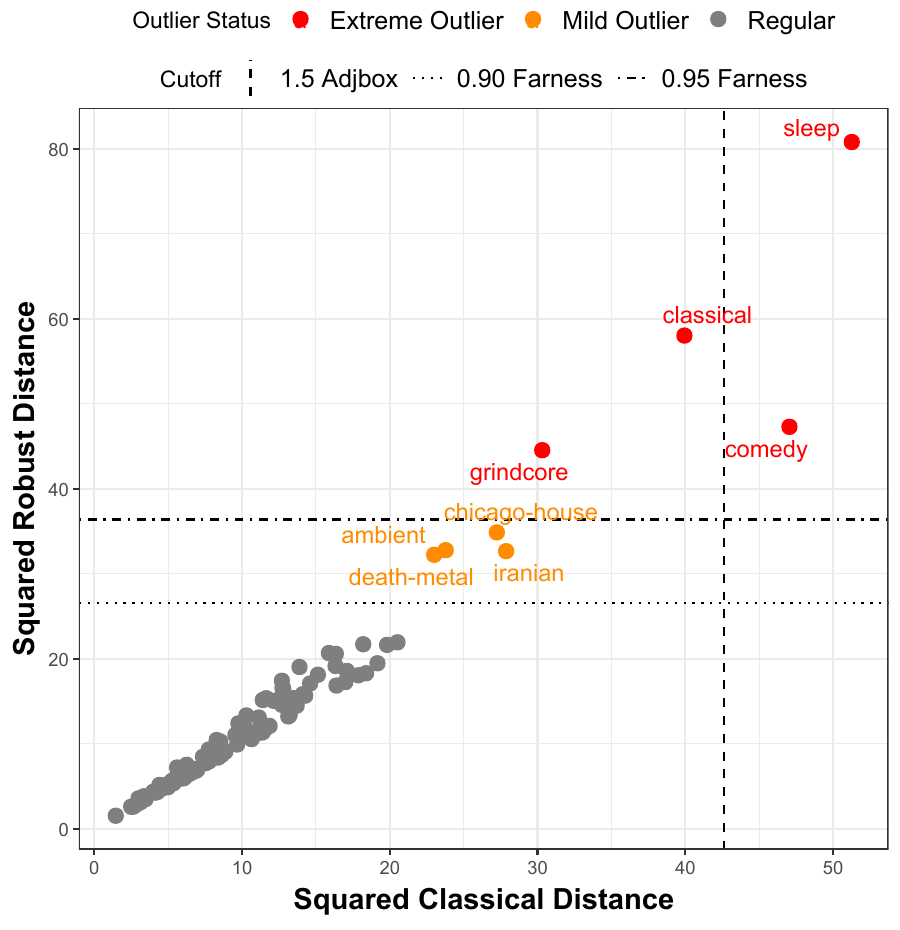}
        \caption{Distance-distance plot.}
        \label{fig:spotify_dist_dist}
    \end{subfigure}
    \caption{IMCD correlation estimate (a) and distance-distance (b) plot of the classical versus robust squared Interval-Mahalanobis distances for the Spotify dataset. In the distance-distance plot (b), the horizontal lines correspond to the $0.9$ and $0.95$ farness cutoff values, with the extreme outliers marked in red and the mild outliers in orange, while the vertical line represents the $1.5$ adjusted boxplot cutoff.}
    \label{fig:spotify1}
\end{figure}

Next, we applied the outlier detection rule using a farness cutoff of $0.95$ to identify extreme outliers and $0.9$ for mild outliers. The distance-distance plot in \autoref{fig:spotify_dist_dist} highlights these outlying genres, with extreme and mild outliers marked in red and orange, respectively. The classical method only detects \textit{sleep} and \textit{comedy} as outlying genres, whereas the robust method flags \textit{sleep}, \textit{comedy}, \textit{classical}, and \textit{grindcore} as extreme outliers, and \textit{chicago-house}, \textit{ambient}, \textit{iranian}, and \textit{death-metal} as mild ones. These results are musically interpretable: \textit{sleep} and \textit{classical} typically exhibit very low energy and loudness with high acousticness (and, for \textit{classical}, long duration), placing them far from most genres; \textit{comedy} commonly stands out due to unusually high speechiness and low instrumentalness; and \textit{grindcore} is often characterized by high energy and danceability. The mild outliers also present coherent deviations from the majority of genres and thus the dataset's barycenter. Compared with the classical estimator, the robust approach captures a broader and more meaningful set of atypical genres, highlighting its improved sensitivity to interval-valued audio features. Moreover, even though the intervals were trimmed to the $1\%$ and $99\%$ quantiles, the robust method still identifies outlying genres, demonstrating that interval-valued outliers go beyond extreme values in the microdata.

\section{Conclusion}
\label{sec:conclusion}
This work introduces a robust framework for analyzing interval-valued data in the presence of anomalous observations. Building on the expected value and covariance matrix definitions based on the Mallows distance, we propose an extension of the MCD estimator to interval-valued data that jointly accounts for centers, ranges, and the underlying microdata structure. From the resulting robust location and scatter estimates, we define a robust Interval-Mahalanobis distance and suggest fully non-parametric cutoff strategies for outlier detection, including adjusted boxplots and farness-based approaches.

Extensive simulations demonstrate that the IMCD estimator consistently outperforms its classical counterpart under contamination, yielding substantially lower estimation errors for the symbolic covariance matrix and markedly improved outlier detection rates. The reweighting step proves essential for recovering efficiency, with farness based reweighting showing particularly stable performance across scenarios. The two real world applications further illustrate the practical utility of the method, successfully identifying meaningful outliers in interval-valued datasets from the automotive and music domains.

In contrast to the existing approach to the MCD estimator for interval-valued data, the proposed method is specifically designed to operate within a unified framework that explicitly incorporates both macrodata (interval bounds) and microdata variability through the Mallows-distance-based covariance structure. This distinction is central to the methodological contribution of the paper: rather than adapting classical MCD ideas using only interval summaries such as bounds or centers, we construct a robust estimator that accounts for the internal distributional information encoded in the intervals. Within this framework, the proposed Interval-Mahalanobis distance provides a coherent tool for outlier detection that reflects all relevant sources of variability.

Several important directions remain for future research. From a theoretical perspective, key robustness properties of the estimator, such as its breakdown point and influence function, have not yet been formally established and warrant further investigation. In addition, as with the classical MCD estimator, the proposed method is subject to limitations in high-dimensional settings where the number of variables becomes large relative to the sample size. This issue may lead to instability or singularity of the covariance estimates. In the conventional multivariate setting, this limitation has been addressed through regularized versions of the MCD estimator \citep{regularizedMCD}, and extending such regularization strategies to the interval-valued framework represents a promising direction for future research. Furthermore, given the central role of the conventional MCD estimator in robust multivariate analysis, we expect the proposed IMCD estimator to serve a similar purpose as a building block for a wide range of robust methods for multivariate interval-valued data.

The R-package \texttt{AIDA}, including the proposed methods and applications, is available from CRAN at \url{https://cran.r-project.org/package=AIDA}. In particular, vignette ``IMCD estimator examples'' reproduces the results from Section \ref{sec:applications}. For more details, such as the code used to generate the simulation results and pre-process the datasets, the reader may refer to the GitHub repository at \url{https://github.com/catarinaploureiro/AIDA}.

\section*{Acknowledgments} 
\noindent The authors gratefully acknowledge Professor João Xavier for his valuable assistance with the optimization problem. Generative AI tools (GPT-4 and GPT-5) were used for language editing. This work was supported by FCT — Funda\c{c}\~ao para a Ci\^encia e Tecnologia, I.P. and the Recovery and Resilience Plan (PRR) through the grant \href{https://doi.org/10.54499/UI/BD/153720/2021}{UI/\discretionary{}{}{}BD/\discretionary{}{}{}153720/\discretionary{}{}{}2021}, and through the projects \href{https://doi.org/10.54499/UID/PRR/04621/2025}{UID/\discretionary{}{}{}04621/\discretionary{}{}{}2025}, \href{https://doi.org/10.54499/UID/50014/2025}{UID/\discretionary{}{}{}50014/\discretionary{}{}{}2025}, \href{https://doi.org/10.54499/UID/04459/2025}{UID/\discretionary{}{}{}04459/\discretionary{}{}{}2025}, and \href{https://doi.org/10.54499/UID/PRR/04459/2025}{UID/\discretionary{}{}{}PRR/\discretionary{}{}{}04459/\discretionary{}{}{}2025}.

\small{
\bibliographystyle{elsarticle-harv}
\bibliography{cas-refs}}

\normalsize

\bigskip
\bigskip
\bigskip

\appendix
\section{Proofs from Section \ref{sec:MCD}}
\label{sec:appendix_mcd}

\subsection{Proof of Proposition \ref{prop:cov_matrix_rewritten}}
\label{sec:proof_cov_matrix_rewritten}
Noticing that
\begin{equation*}
    \begin{bmatrix}
        \boldsymbol{I} & \boldsymbol{\Psi}/2
    \end{bmatrix}
    \begin{bmatrix}
        \boldsymbol{\Sigma}_{CC} & \boldsymbol{\Sigma}_{CR}\\
        \boldsymbol{\Sigma}_{RC} & \boldsymbol{\Sigma}_{RR}
    \end{bmatrix}
    \begin{bmatrix}
        \boldsymbol{I} \\
        \boldsymbol{\Psi}/2
    \end{bmatrix}=
    \boldsymbol{\Sigma}_{CC}+\frac{1}{2}\boldsymbol{\Sigma}_{CR}\boldsymbol{\Psi}+\frac{1}{2}\boldsymbol{\Psi}\boldsymbol{\Sigma}_{RC}+\frac{1}{4}\boldsymbol{\Psi}\boldsymbol{\Sigma}_{RR}\boldsymbol{\Psi},
\end{equation*}
and setting $\boldsymbol{\Lambda}=\begin{bmatrix}
            \boldsymbol{I} & \boldsymbol{\Psi}/2
        \end{bmatrix}$, we have that
\begin{equation}
    \label{eq:proof1}
    \boldsymbol{\Sigma}_B=\boldsymbol{\Lambda}\boldsymbol{\Sigma}_{2p}\boldsymbol{\Lambda}^\top+
    \frac{1}{4}\boldsymbol{\mathfrak{E}}_{UU}\bullet\boldsymbol{\Sigma}_{RR}-\frac{1}{4}\boldsymbol{\Psi}\boldsymbol{\Sigma}_{RR}\boldsymbol{\Psi}.
\end{equation}
It is also the case that
\begin{align*}
    [\boldsymbol{\mathfrak{E}}_{UU}\bullet\boldsymbol{\Sigma}_{RR}-\boldsymbol{\Psi}\boldsymbol{\Sigma}_{RR}\boldsymbol{\Psi}]_{j\ell}&=
    \begin{cases}
        \mathrm{Var}(U_j)\mathrm{Var}(R_j), &\mathrm{if} \quad j=\ell,\\
        \left[\mathcal{E}(U_j,U_\ell)-\mathbb{E}(U_j)\mathbb{E}(U_\ell)\right]\mathrm{Cov}(R_j,R_\ell),&\mathrm{if} \quad j\neq \ell,
    \end{cases}\\
    &=\mathrm{Cov}(F_{U_j}^{-1}(T), F_{U_\ell}^{-1}(T))\mathrm{Cov}(R_j,R_\ell),
\end{align*}
where $T\sim\mathrm{Unif}(0,1)$. For $[\boldsymbol{\Xi}]_{j\ell}=\mathrm{Cov}(F_{U_j}^{-1}(T), F_{U_\ell}^{-1}(T))=\widetilde{\mathrm{Cov}}(U_j,U_\ell)$, $i,\ell=1,\dots,p$, with $\widetilde{\mathrm{Cov}}(U_j,U_\ell)$ as in \eqref{eq:cov_xi}, we can conclude that 
\begin{equation}
    \label{eq:proof2}
    \boldsymbol{\mathfrak{E}}_{UU}\bullet\boldsymbol{\Sigma}_{RR}-\boldsymbol{\Psi}\boldsymbol{\Sigma}_{RR}\boldsymbol{\Psi}=\boldsymbol{\Xi}\bullet\boldsymbol{\Sigma}_{RR}.
\end{equation}
It now follows from \eqref{eq:proof1} and \eqref{eq:proof2} that \eqref{eq:cov_matrix_rewritten} holds. This concludes the proof.
\qed

\subsection{Proof of Proposition \ref{prop:concavity}}
\label{sec:proof_matrix_concave}
We begin by proving that $\boldsymbol{S}_B(\boldsymbol{z})$ defined in \eqref{eq:cov_matrix_z} is matrix concave. That is, for any $\boldsymbol{z},\boldsymbol{w}\in[0,1]^n$ and $\alpha\in[0,1]$, we have $\boldsymbol{S}_B((1-\alpha)\boldsymbol{z}+\alpha\boldsymbol{w})\succcurlyeq(1-\alpha)\boldsymbol{S}_B(\boldsymbol{z})+\alpha\boldsymbol{S}_B(\boldsymbol{w})$. Here, for symmetric matrices $\boldsymbol{A},\boldsymbol{B}\in\mathbb{R}^{p\times p}$, $\boldsymbol{A}\succcurlyeq\boldsymbol{B}$ denotes generalized matrix inequality, i.e., $\boldsymbol{A}-\boldsymbol{B}$ is positive semidefinite. Taking \eqref{eq:cov_matrix_rewritten}, \eqref{eq:barycenter_z}, and \eqref{eq:cov_matrix_z} into account, we can express $\boldsymbol{S}_B(\boldsymbol{z})$ as
\begin{equation}
    \label{eq:S_B}
    \boldsymbol{S}_B(\boldsymbol{z})=\boldsymbol{\Lambda}\boldsymbol{S}_{2p}(\boldsymbol{z})\boldsymbol{\Lambda}^\top+\frac{1}{4}\boldsymbol{\Xi}\bullet\boldsymbol{S}_{RR}(\boldsymbol{z}),
\end{equation}
where 
\begin{equation}
    \label{eq:S_2p}
    \boldsymbol{S}_{2p}(\boldsymbol{z})=\frac{1}{m}\sum_{i=1}^nz_i\boldsymbol{x}_i\boldsymbol{x}_i^\top-\frac{1}{m^2}\boldsymbol{X}^{\top}\boldsymbol{z}\boldsymbol{z}^\top\boldsymbol{X},
\end{equation}
and
\begin{equation}
    \label{eq:S_RR}
    \boldsymbol{S}_{RR}(\boldsymbol{z})=\frac{1}{m}\sum_{i=1}^nz_i\boldsymbol{r}_i\boldsymbol{r}_i^\top-\frac{1}{m^2}\boldsymbol{R}^{\top}\boldsymbol{z}\boldsymbol{z}^\top\boldsymbol{R}.
\end{equation}

First, we can prove that, for any $\boldsymbol{z},\boldsymbol{w}\in[0,1]^n$ and $\alpha\in[0,1]$, $\boldsymbol{S}_{2p}(\boldsymbol{z})$ is matrix concave. Let us assume this property is true, then we have
\begin{gather*}
    \boldsymbol{S}_{2p}((1-\alpha)\boldsymbol{z}+\alpha\boldsymbol{w})\succcurlyeq(1-\alpha)\boldsymbol{S}_{2p}(\boldsymbol{z})+\alpha\boldsymbol{S}_{2p}(\boldsymbol{w})\Leftrightarrow\\
    \frac{1}{m}\sum_{i=1}^n\left((1-\alpha)z_i+\alpha w_i\right)\boldsymbol{x}_i\boldsymbol{x}_i^\top-\frac{1}{m^2}\boldsymbol{X}^{\top}\left((1-\alpha)\boldsymbol{z}+\alpha\boldsymbol{w}\right)\left((1-\alpha)\boldsymbol{z}+\alpha\boldsymbol{w}\right)^\top\boldsymbol{X}\qquad\qquad\qquad\\
    \qquad\quad\succcurlyeq(1-\alpha)\left[\frac{1}{m}\sum_{i=1}^nz_i\boldsymbol{x}_i\boldsymbol{x}_i^\top-\frac{1}{m^2}\boldsymbol{X}^{\top}\boldsymbol{z}\boldsymbol{z}^\top\boldsymbol{X}\right]+\alpha\left[\frac{1}{m}\sum_{i=1}^nw_i\boldsymbol{x}_i\boldsymbol{x}_i^\top-\frac{1}{m^2}\boldsymbol{X}^{\top}\boldsymbol{w}\boldsymbol{w}^\top\boldsymbol{X}\right]\Leftrightarrow\\
    \boldsymbol{X}^{\top}\left((1-\alpha)\boldsymbol{z}+\alpha\boldsymbol{w}\right)\left((1-\alpha)\boldsymbol{z}+\alpha\boldsymbol{w}\right)^\top\boldsymbol{X}\preccurlyeq(1-\alpha)\boldsymbol{X}^{\top}\boldsymbol{z}\boldsymbol{z}^\top\boldsymbol{X}+\alpha\boldsymbol{X}^{\top}\boldsymbol{w}\boldsymbol{w}^\top\boldsymbol{X}\Leftrightarrow\\
    (1-2\alpha+\alpha^2)\boldsymbol{X}^{\top}\boldsymbol{z}\boldsymbol{z}^\top\boldsymbol{X}+\alpha(1-\alpha)\boldsymbol{X}^{\top}\boldsymbol{z}\boldsymbol{w}^\top\boldsymbol{X}+\alpha(1-\alpha)\boldsymbol{X}^{\top}\boldsymbol{w}\boldsymbol{z}^\top\boldsymbol{X}+\alpha^2\boldsymbol{X}^{\top}\boldsymbol{w}\boldsymbol{w}^\top\boldsymbol{X}\qquad\\
    \qquad\qquad\qquad\qquad\qquad\qquad\qquad\qquad\qquad\qquad\qquad\qquad\preccurlyeq(1-\alpha)\boldsymbol{X}^{\top}\boldsymbol{z}\boldsymbol{z}^\top\boldsymbol{X}+\alpha\boldsymbol{X}^{\top}\boldsymbol{w}\boldsymbol{w}^\top\boldsymbol{X}\Leftrightarrow\\
    \boldsymbol{0}\preccurlyeq\boldsymbol{X}^{\top}\boldsymbol{z}\boldsymbol{z}^\top\boldsymbol{X}-\boldsymbol{X}^{\top}\boldsymbol{z}\boldsymbol{w}^\top\boldsymbol{X}-\boldsymbol{X}^{\top}\boldsymbol{w}\boldsymbol{z}^\top\boldsymbol{X}+\boldsymbol{X}^{\top}\boldsymbol{w}\boldsymbol{w}^\top\boldsymbol{X}\Leftrightarrow\\
    \boldsymbol{0}\preccurlyeq(\boldsymbol{X}^{\top}\boldsymbol{z}-\boldsymbol{X}^{\top}\boldsymbol{w})(\boldsymbol{X}^{\top}\boldsymbol{z}-\boldsymbol{X}^{\top}\boldsymbol{w})^\top
\end{gather*}
which always holds, since $(\boldsymbol{X}^{\top}\boldsymbol{z}-\boldsymbol{X}^{\top}\boldsymbol{w})(\boldsymbol{X}^{\top}\boldsymbol{z}-\boldsymbol{X}^{\top}\boldsymbol{w})^\top$ is positive semidefinite. Then, $\boldsymbol{\Lambda}\boldsymbol{S}_{2p}(\boldsymbol{z})\boldsymbol{\Lambda}^\top$ is matrix concave. A similar reasoning shows that $\boldsymbol{S}_{RR}(\boldsymbol{z})$ is also matrix concave. In addition, we know that $\boldsymbol{\Xi}$ defined in \eqref{eq:cov_xi} is positive semidefinite, since it is a covariance matrix of the quantile functions $F_{U_j}^{-1}(T)$, $T\sim\mathrm{Unif}(0,1)$. Consequently, given that the Schur product and the addition operation preserve positive semidefiniteness, we can conclude that $\boldsymbol{S}_B(\boldsymbol{z})$ is also matrix concave.

Additionally, it is known that the function $h(\boldsymbol{A})=\log\det(\boldsymbol{A})$ is concave and non-decreasing \citep{convex_book}. Hence, the Composition Theorem \citep{convex_book} applies: a non-decreasing concave function composed with a matrix concave function is also concave. Therefore, we can conclude that the objective function $g(\boldsymbol{z})=\log\det\left(\boldsymbol{S}_B(\boldsymbol{z})\right)$ is concave. Its concavity holds on the domain of $\boldsymbol{z}$ such that $\boldsymbol{S}_B(\boldsymbol{z})$ is positive definite, ensuring that the logarithm of the determinant is well-defined.
\qed

\subsection{Proof of Proposition \ref{prop:gradient}}
\label{sec:proof_gradient}
As will be seen, the formula for the gradient of $g(\boldsymbol{z})$ defined in \eqref{eq:gradient} follows from
\begin{equation}
    \label{eq:derivative_logdet}
    \frac{\partial\log\det(\boldsymbol{A})}{\partial x}=\mathrm{tr}\left(\boldsymbol{A}^{-1}\frac{\partial\boldsymbol{A}}{\partial x}\right),
\end{equation}
(see, e.g., \cite{Harville1998}) and properties of the trace. Firstly, we show that, for $\boldsymbol{c}=(c_1,\dots,c_p)^\top\in\mathbb{R}^p$, $\boldsymbol{A}=[a_{j\ell}]\in\mathbb{R}^{p\times p}$, and $\boldsymbol{B}=[b_{j\ell}]\in\mathbb{R}^{p\times p}$ symmetric, we have
\begin{equation}
    \label{eq:trace_equality}
    \mathrm{tr}(\boldsymbol{A}(\boldsymbol{B}\bullet(\boldsymbol{cc}^\top)))=\mathrm{tr}(\boldsymbol{c}^\top(\boldsymbol{A}\bullet\boldsymbol{B})\boldsymbol{c})=\mathrm{tr}(\boldsymbol{c}^\top(\boldsymbol{B}\bullet\boldsymbol{A})\boldsymbol{c}).
\end{equation}
To prove that $\mathrm{tr}(\boldsymbol{A}(\boldsymbol{B}\bullet(\boldsymbol{cc}^\top)))=\mathrm{tr}(\boldsymbol{c}^\top(\boldsymbol{A}\bullet\boldsymbol{B})\boldsymbol{c})$, we can write
\begin{align*}
    \mathrm{tr}(\boldsymbol{A}(\boldsymbol{B}\bullet(\boldsymbol{cc}^\top)))&=\mathrm{tr}(\boldsymbol{A}(\boldsymbol{B}\bullet[c_jc_\ell]))=\mathrm{tr}(\boldsymbol{A}[b_{j\ell}c_jc_\ell])\\
    &=\mathrm{tr}\left(\mathrm{diag}\left(\sum_{\ell=1}^pa_{1\ell}b_{\ell 1}c_\ell c_1,\dots,\sum_{\ell=1}^pa_{p\ell}b_{\ell p}c_\ell c_p\right)\right),
\end{align*}
where $[\cdot]$ denotes the matrix with $(j,\ell)$-entry given by the enclosed expression, and, since $\boldsymbol{B}$ is a symmetric matrix,
\begin{equation*}
    \mathrm{tr}(\boldsymbol{A}(\boldsymbol{B}\bullet(\boldsymbol{cc}^\top)))=
    \mathrm{tr}\left(\mathrm{diag}\left(c_1\sum_{\ell=1}^pa_{1\ell}b_{1\ell}c_\ell,\dots,c_p\sum_{\ell=1}^pa_{p\ell}b_{p\ell}c_\ell\right)\right)=\sum_{j=1}^pc_j\sum_{\ell=1}^pa_{j\ell}b_{j\ell}c_\ell.
\end{equation*}
However,
\begin{align*}
    \mathrm{tr}(\boldsymbol{c}^\top(\boldsymbol{A}\bullet\boldsymbol{B})\boldsymbol{c})&=\boldsymbol{c}^\top(\boldsymbol{A}\bullet\boldsymbol{B})\boldsymbol{c}=\boldsymbol{c}^\top[a_{j\ell}b_{j\ell}]\boldsymbol{c}\\
    &=\boldsymbol{c}^\top\left(\sum_{\ell=1}^pa_{1\ell}b_{1\ell}c_\ell,\dots,\sum_{\ell=1}^pa_{p\ell}b_{p\ell}c_\ell\right)^\top=\sum_{j=1}^pc_j\sum_{\ell=1}^pa_{j\ell}b_{j\ell}c_\ell,
\end{align*}
which shows that the first equality in \eqref{eq:trace_equality} holds. The second equality in \eqref{eq:trace_equality} follows from the commutative property of the Schur product. 

We now prove \eqref{eq:gradient}. From \eqref{eq:derivative_logdet} and \eqref{eq:S_B}, we have
\begin{align*}
    \frac{\partial g(\boldsymbol{z})}{\partial z_i}&=\mathrm{tr}\left(\boldsymbol{S}_B(\boldsymbol{z})^{-1}\frac{\partial\boldsymbol{S}_B(\boldsymbol{z})}{\partial z_i}\right)\\
    &=\mathrm{tr}\left(\boldsymbol{S}_B(\boldsymbol{z})^{-1}\left(\boldsymbol{\Lambda}\left(\frac{\partial\boldsymbol{S}_{2p}(\boldsymbol{z})}{\partial z_i}\right)\boldsymbol{\Lambda}^\top+\frac{1}{4}\boldsymbol{\Xi}\bullet\left(\frac{\partial\boldsymbol{S}_{RR}(\boldsymbol{z})}{\partial z_i}\right)\right)\right)\\
    &=\frac{1}{m}\mathrm{tr}\Bigg(\boldsymbol{S}_B(\boldsymbol{z})^{-1}\boldsymbol{\Lambda}\boldsymbol{x}_i\boldsymbol{x}_i^\top\boldsymbol{\Lambda}^\top-\boldsymbol{S}_B(\boldsymbol{z})^{-1}\boldsymbol{\Lambda}\boldsymbol{x}_i\overline{\boldsymbol{x}}_B(\boldsymbol{z})^\top\boldsymbol{\Lambda}^\top\\
    &\qquad\qquad-\boldsymbol{S}_B(\boldsymbol{z})^{-1}\boldsymbol{\Lambda}\overline{\boldsymbol{x}}_B(\boldsymbol{z})\boldsymbol{x}_i^\top\boldsymbol{\Lambda}^\top+\frac{1}{4}\boldsymbol{S}_B(\boldsymbol{z})^{-1}\left(\boldsymbol{\Xi}\bullet\left(\boldsymbol{r}_i\boldsymbol{r}_i^\top\right)\right)\\
    &\qquad\qquad-\frac{1}{4}\boldsymbol{S}_B(\boldsymbol{z})^{-1}\left(\boldsymbol{\Xi}\bullet\left(\boldsymbol{r}_i\overline{\boldsymbol{r}}(\boldsymbol{z})^\top\right)\right)-\frac{1}{4}\boldsymbol{S}_B(\boldsymbol{z})^{-1}\left(\boldsymbol{\Xi}\bullet\left(\overline{\boldsymbol{r}}(\boldsymbol{z})\boldsymbol{r}_i^\top\right)\right)\Bigg).\\
\end{align*}
Finally, by the cyclic property of the trace and \eqref{eq:trace_equality}, we can write
\begin{align*}
    \frac{\partial g(\boldsymbol{z})}{\partial z_i}&=\frac{1}{m}\Bigg(\boldsymbol{x}_i^\top\boldsymbol{\Lambda}^\top\boldsymbol{S}_B(\boldsymbol{z})^{-1}\boldsymbol{\Lambda}\boldsymbol{x}_i-\overline{\boldsymbol{x}}_B(\boldsymbol{z})^\top\boldsymbol{\Lambda}^\top\boldsymbol{S}_B(\boldsymbol{z})^{-1}\boldsymbol{\Lambda}\boldsymbol{x}_i\\
    &\qquad\quad-\boldsymbol{x}_i^\top\boldsymbol{\Lambda}^\top\boldsymbol{S}_B(\boldsymbol{z})^{-1}\boldsymbol{\Lambda}\overline{\boldsymbol{x}}_B(\boldsymbol{z})+\frac{1}{4}\boldsymbol{r}_i^\top\left(\boldsymbol{\Xi}\bullet\boldsymbol{S}_B(\boldsymbol{z})^{-1}\right)\boldsymbol{r}_i\\
    &\qquad\quad-\frac{1}{4}\overline{\boldsymbol{r}}(\boldsymbol{z})^\top\left(\boldsymbol{\Xi}\bullet\boldsymbol{S}_B(\boldsymbol{z})^{-1}\right)\boldsymbol{r}_i-\frac{1}{4}\boldsymbol{r}_i^\top\left(\boldsymbol{\Xi}\bullet\boldsymbol{S}_B(\boldsymbol{z})^{-1}\right)\overline{\boldsymbol{r}}(\boldsymbol{z})\Bigg)\\
    &=\frac{1}{m}\Bigg((\boldsymbol{x}_i-\overline{\boldsymbol{x}}_B(\boldsymbol{z}))^\top\boldsymbol{\Lambda}^\top\boldsymbol{S}_B(\boldsymbol{z})^{-1}\boldsymbol{\Lambda}(\boldsymbol{x}_i-\overline{\boldsymbol{x}}_B(\boldsymbol{z}))-\overline{\boldsymbol{x}}_B(\boldsymbol{z})^\top\boldsymbol{\Lambda}^\top\boldsymbol{S}_B(\boldsymbol{z})^{-1}\boldsymbol{\Lambda}\overline{\boldsymbol{x}}_B(\boldsymbol{z})\\
    &\qquad\quad+\frac{1}{4}(\boldsymbol{r}_i-\overline{\boldsymbol{r}}(\boldsymbol{z}))^\top\left(\boldsymbol{\Xi}\bullet\boldsymbol{S}_B(\boldsymbol{z})^{-1}\right)(\boldsymbol{r}_i-\overline{\boldsymbol{r}}(\boldsymbol{z}))-\frac{1}{4}\overline{\boldsymbol{r}}(\boldsymbol{z})^\top\left(\boldsymbol{\Xi}\bullet\boldsymbol{S}_B(\boldsymbol{z})^{-1}\right)\overline{\boldsymbol{r}}(\boldsymbol{z})\Bigg).
\end{align*}
\qed

\section{IMCD Estimator Algorithm Pseudocode}
\label{sec:appendix_mcd_algorithm}

\begin{algorithm}[H]
\caption{Minorization Step}
\label{alg:minorization}
\KwIn{Interval-valued observations $\{\boldsymbol{x}_1,\dots,\boldsymbol{x}_n\}$, subset size $m$, weight vector $\boldsymbol{z}$}
\KwOut{Updated $\boldsymbol{z}$}
    
    Compute distances between each observation $\boldsymbol{x}_i$ and the current barycenter estimate:
    \begin{align*}
        d_i^2(\boldsymbol{z})&=\left(\boldsymbol{x}_i-\overline{\boldsymbol{x}}_B(\boldsymbol{z})\right)^\top\boldsymbol{\Lambda}^\top\boldsymbol{S}_B(\boldsymbol{z})^{-1}\boldsymbol{\Lambda}\left(\boldsymbol{x}_i-\overline{\boldsymbol{x}}_B(\boldsymbol{z})\right)\\
        &\quad+\frac{1}{4}\left(\boldsymbol{r}_i-\overline{\boldsymbol{r}}(\boldsymbol{z})\right)^\top\left(\boldsymbol{\Xi}\bullet\boldsymbol{S}_B(\boldsymbol{z})^{-1}\right)\left(\boldsymbol{r}_i-\overline{\boldsymbol{r}}(\boldsymbol{z})\right)
    \end{align*}

    Sort distances in ascending order: $$d_{(1)}^2(\boldsymbol{z})\leqslant \dots\leqslant d_{(m)}^2(\boldsymbol{z})\leqslant d_{(m+1)}^2(\boldsymbol{z})\leqslant\dots\leqslant d_{(n)}^2(\boldsymbol{z})$$
    
    Select $m$ observations with smallest $d_i^2$ to update $\boldsymbol{z}$ as:
    $$z_i \leftarrow  \begin{cases}
        1, \quad \mathrm{if} \ d_i^2(\boldsymbol{z}^{(t)})\leqslant d_{(m)}^2(\boldsymbol{z}^{(t)}),\\
        0, \quad \mathrm{otherwise.}
    \end{cases}$$
\end{algorithm}

\begin{algorithm}[H]
\caption{IMCD Estimator (adapted from \cite{fastMCD})}
\label{alg:IMCD}

\KwIn{Interval-valued observations $\{\boldsymbol{x}_1,\dots,\boldsymbol{x}_n\}$, subset size $m$}
\KwOut{Raw and reweighted IMCD estimates}

\uIf{$m=n$}{
  $\boldsymbol{z}_{\mathrm{final}} \leftarrow (1,1,\dots,1)^\top$\;
}
\uElseIf{$n \le 600$}{

  Draw 500 random $(p+1)$-subsets with nonsingular covariance\;

  \ForEach{subset}{
    Expand to size $m$ using the $m$ smallest $d_i^2(\boldsymbol{z})$\;
    Apply the Minorization Step twice\;
    Compute $\log \det \boldsymbol{S}_B(\boldsymbol{z})$\;
  }

  Retain 10 subsets with smallest $\log \det \boldsymbol{S}_B(\boldsymbol{z})$\;

  \ForEach{of these 10 subsets}{
    Iterate Minorization Step until convergence\;
    Recompute $\log \det \boldsymbol{S}_B(\boldsymbol{z})$\;
  }

  $\boldsymbol{z}_{\mathrm{final}} \leftarrow$ subset minimizing $\log \det \boldsymbol{S}_B(\boldsymbol{z})$\;
}
\Else{

  Draw a random subset of size 
  $n_{\mathrm{merge}}=\min(n,1500)$\;

  Partition this reduced dataset into 
  $k=\min(5,\lceil n/300\rceil)$ disjoint $n_{sub}$-subsets\;

  Let $m_{\text{sub}}=\lfloor n_{\text{sub}} m/n \rfloor$ and $m_{\mathrm{merge}}=\lfloor n_{\mathrm{merge}} m/n \rfloor$\;

  \ForEach{partition}{
    Draw $(500/k)$ random $(p+1)$-subsets with nonsingular covariance\;

    \ForEach{subset}{
      Expand to size $m_{\text{sub}}$ using the $m_{\text{sub}}$ smallest $d_i^2(\boldsymbol{z})$\;
      Apply Minorization Step twice\;
      Compute $\log \det \boldsymbol{S}_B(\boldsymbol{z})$\;
    }

    Retain 10 subsets with smallest $\log \det \boldsymbol{S}_B(\boldsymbol{z})$ within this partition\;
  }

  Embed each of the $10k$ retained subsets into a vector 
    of length $n_{\mathrm{merge}}$\;
  
  \ForEach{of these $10k$ subsets}{
    Apply the Minorization Step twice on the reduced dataset with $m_{\mathrm{merge}}$\;
    Compute $\log \det \boldsymbol{S}_B(\boldsymbol{z})$\;
  }

  Retain 10 subsets with smallest $\log \det \boldsymbol{S}_B(\boldsymbol{z})$\;

  Map retained subsets back to the full dataset of size $n$\;

  \ForEach{of these 10 subsets}{
    Iterate Minorization Step on the full dataset until convergence\;
    Recompute $\log \det \boldsymbol{S}_B(\boldsymbol{z})$\;
  }

  $\boldsymbol{z}_{\mathrm{final}} \leftarrow$ subset minimizing $\log \det \boldsymbol{S}_B(\boldsymbol{z})$\;
}

Compute  raw IMCD estimates based on 
$\boldsymbol{z}_{\mathrm{final}}$: $\overline{\boldsymbol{x}}_{\mathrm{rawIMCD}}$ and $\boldsymbol{S}_{\mathrm{rawIMCD}}$\;

Perform one-step reweighting to obtain final IMCD estimates: $\overline{\boldsymbol{x}}_{\mathrm{IMCD}}$ and $\boldsymbol{S}_{\mathrm{IMCD}}$\;

\end{algorithm}

\section{Additional Simulation Results}
\label{sec:appendix_simulations}
In addition to the metrics presented in the main text, we report further results for both the covariance matrix estimation task and the outlier detection task.

\subsection{Covariance matrix estimation}
\label{sec:appendix_sim_matrix}
The results for the remaining scenarios of the relative Frobenius error are shown in \autoref{fig:frobenius.error_2}. For the matrix estimation scenarios, we also consider the Kullback-Leibler (KL) divergence and the angle error between the estimated and ground-truth covariance matrices:
\begin{itemize}
    \item \textbf{KL divergence}:
        The KL divergence between two $p$-dimensional Gaussian distributions 
        $\mathcal{N}(\boldsymbol{\mu}, \hat{\boldsymbol{\Sigma}})$ and 
        $\mathcal{N}(\boldsymbol{\mu}, \boldsymbol{\Sigma})$ is given by
        \begin{equation}
        \label{eq:kl_divergence}
            \mathrm{KL}(\hat{\boldsymbol{\Sigma}} \,\|\, \boldsymbol{\Sigma}) 
            = \frac{1}{2} \left(
            \mathrm{tr}(\boldsymbol{\Sigma}^{-1} \hat{\boldsymbol{\Sigma}})
            + \log\!\left( \frac{\det(\boldsymbol{\Sigma})}{\det(\hat{\boldsymbol{\Sigma}})} \right)
            - p
            \right),
        \end{equation}
        where $\hat{\boldsymbol{\Sigma}}$ and $\boldsymbol{\Sigma}$ are the estimated and ground truth covariance matrices, respectively. Although KL divergence assumes multivariate Gaussian distributions \citep{kldivergence2023} — making it not directly applicable in our setting — it still provides a useful numerical comparison.
    \item \textbf{Angle error}:
        The angle error quantifies the discrepancy between the eigenvalue vectors of the estimated and ground-truth covariance matrices:
        \begin{equation}
        \label{eq:angle_error}
            1 - \frac{\hat{\boldsymbol{a}}^\top \boldsymbol{a}}
            {\sqrt{\hat{\boldsymbol{a}}^\top \hat{\boldsymbol{a}}} \,
            \sqrt{\boldsymbol{a}^\top \boldsymbol{a}}},
        \end{equation}
        where $\hat{\boldsymbol{a}}$ and $\boldsymbol{a}$ contain the eigenvalues of
        $\hat{\boldsymbol{\Sigma}}$ and $\boldsymbol{\Sigma}$, respectively.
\end{itemize}

The results for the KL divergence and the angle error are presented in \autoref{fig:kl.divergence} and \autoref{fig:angle.error}, respectively.

\subsection{Outlier detection}
\label{sec:appendix_sim_outlier}
The results for the remaining scenarios of the recall and precision of class 1 are shown in \autoref{fig:recall1_2} and \autoref{fig:precision1_2}, respectively. For completeness, we also report the precision and recall of the regular class (0), along with overall accuracy, F$_1$-score and G-mean. These results are presented in \autoref{fig:recall0}--\ref{fig:gmean}, with the metrics defined as follows:
\begin{itemize}
    \item \textbf{Precision} and \textbf{recall} of the regular class:
        \begin{equation}
            \label{eq:pr_re_0}
            Pr(0) = \frac{\mathrm{TN}}{\mathrm{TN} + \mathrm{FN}},
            \quad
            Re(0) = \frac{\mathrm{TN}}{\mathrm{TN} + \mathrm{FP}},
        \end{equation}
        where TN is the number of true negatives, FN is the number of false negatives, and FP is the number of false positives.
    \item \textbf{Accuracy}:
        \begin{equation}
            \label{eq:accuracy}
            Acc = \frac{\mathrm{TP} + \mathrm{TN}}{\mathrm{TP} + \mathrm{TN} + \mathrm{FP} + \mathrm{FN}},
        \end{equation}
        where TP is the number of true positives.
    \item \textbf{F$_1$-score}, the harmonic mean of precision and recall for the outlier class: 
        \begin{equation}
            \label{eq:f1score}
            F_1(1) = \frac{2Pr(1)Re(1)}{Pr(1) + Re(1)}.
        \end{equation}
    \item \textbf{G-mean}, the geometric mean of recall for both classes:
        \begin{equation}
            \label{eq:gmean}
            Gmean = \sqrt{Re(0)Re(1)}
        \end{equation}
        The G-mean is often used in imbalanced learning problems \citep{Gmean}, of which outlier detection is an example.
\end{itemize}

\begin{figure}[ht]
    \centering
    \begin{subfigure}[b]{0.49\textwidth}
        \centering
        \includegraphics[width=\textwidth]{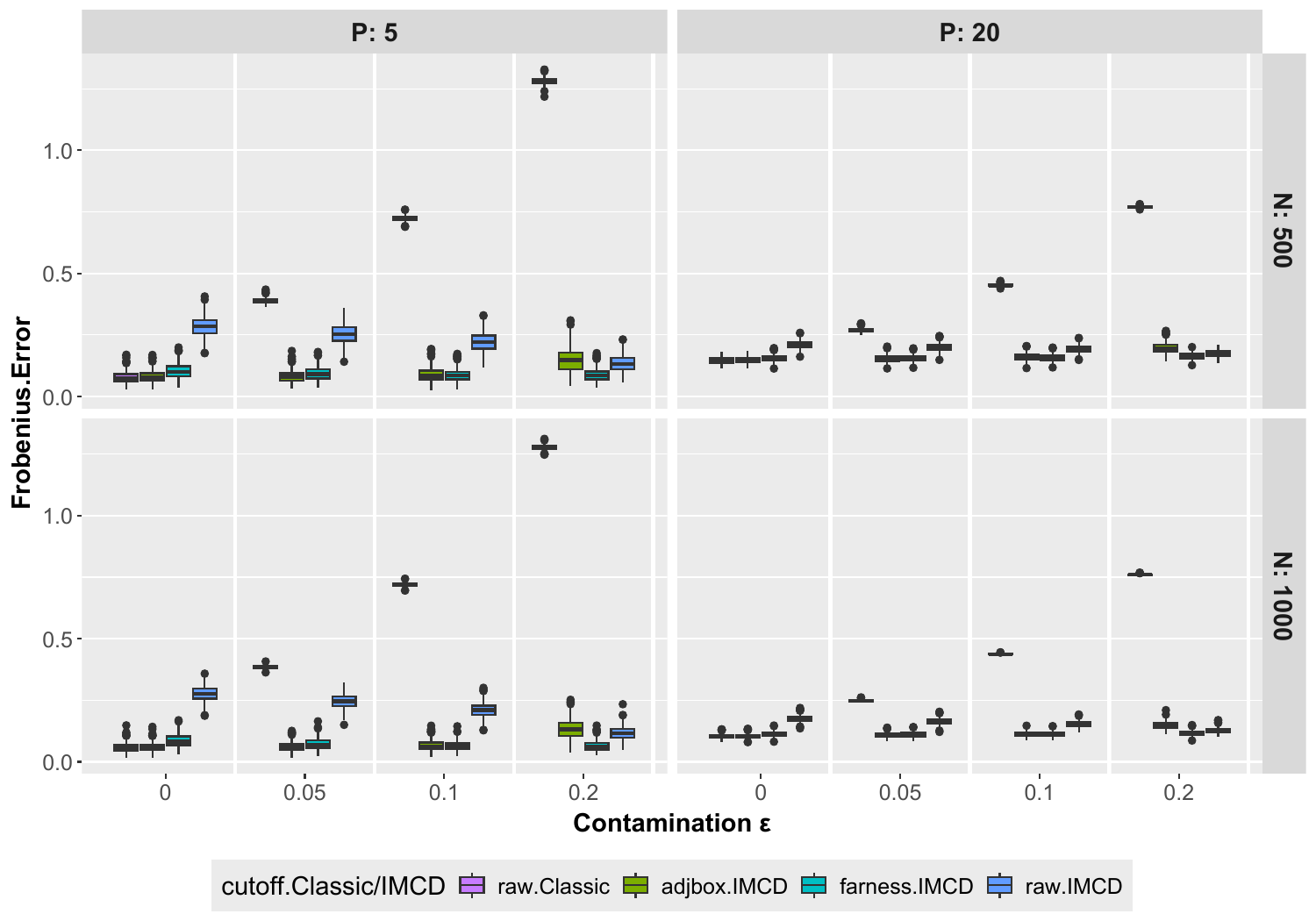}
        \caption{Scenario 3.} 
    \end{subfigure}
    \hfill
    \begin{subfigure}[b]{0.49\textwidth}
        \centering
        \includegraphics[width=\textwidth]{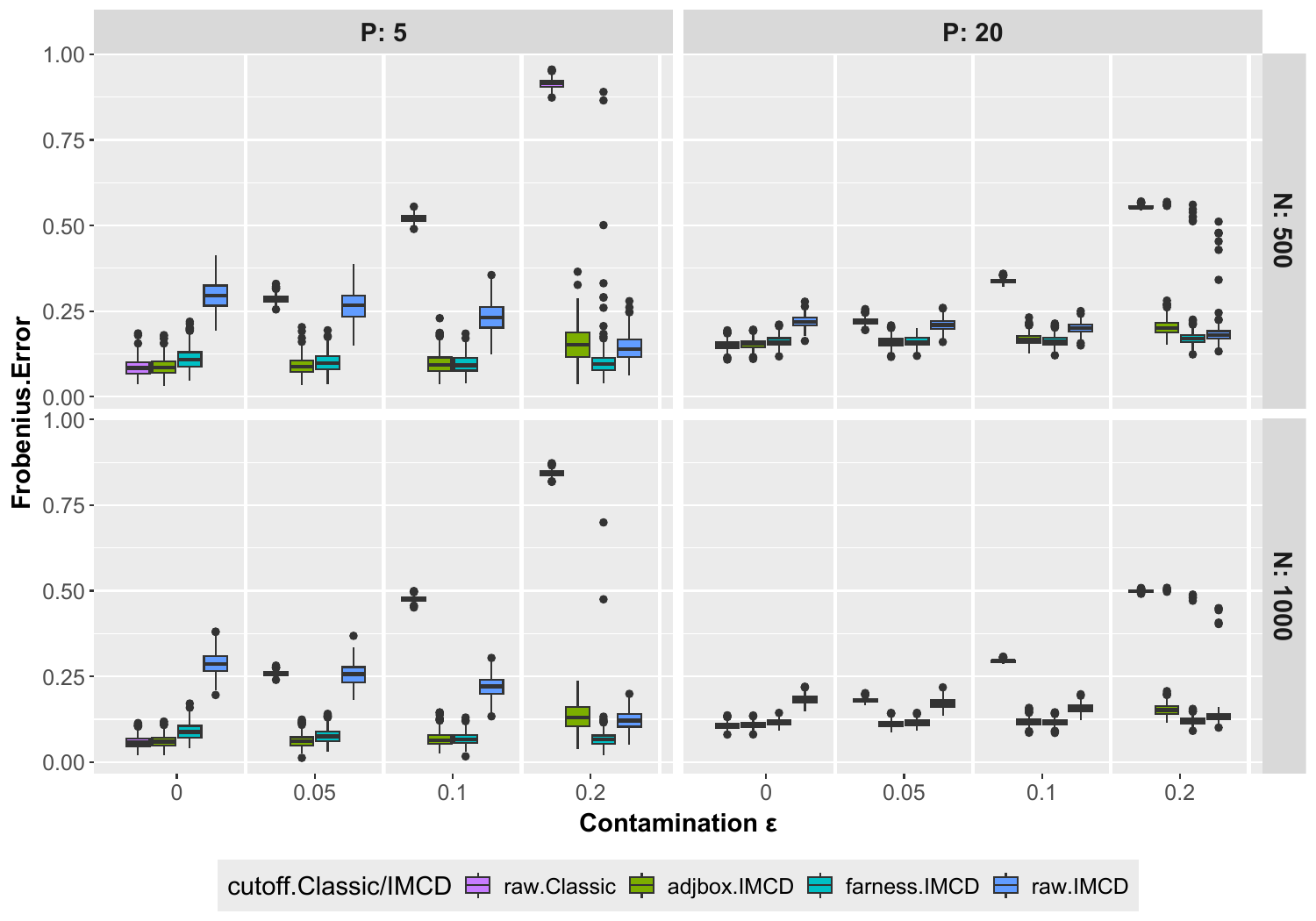}
        \caption{Scenario 6.} 
    \end{subfigure}
    \caption{Boxplots of the relative Frobenius error obtained for scenarios 3 and 6 and the different levels of contamination ($\epsilon$), number of variables ($P$), and sample size ($N$). For each case, we have four covariance matrix estimators: classic without reweighting (raw.Classic), IMCD with adjusted boxplot (adjbox.IMCD) and farness (farness.IMCD) reweighting, and IMCD without reweighting (raw.IMCD).}
    \label{fig:frobenius.error_2}
\end{figure}

\begin{figure}[ht]
    \centering
    \begin{subfigure}[b]{0.49\textwidth}
        \centering
        \includegraphics[width=\textwidth]{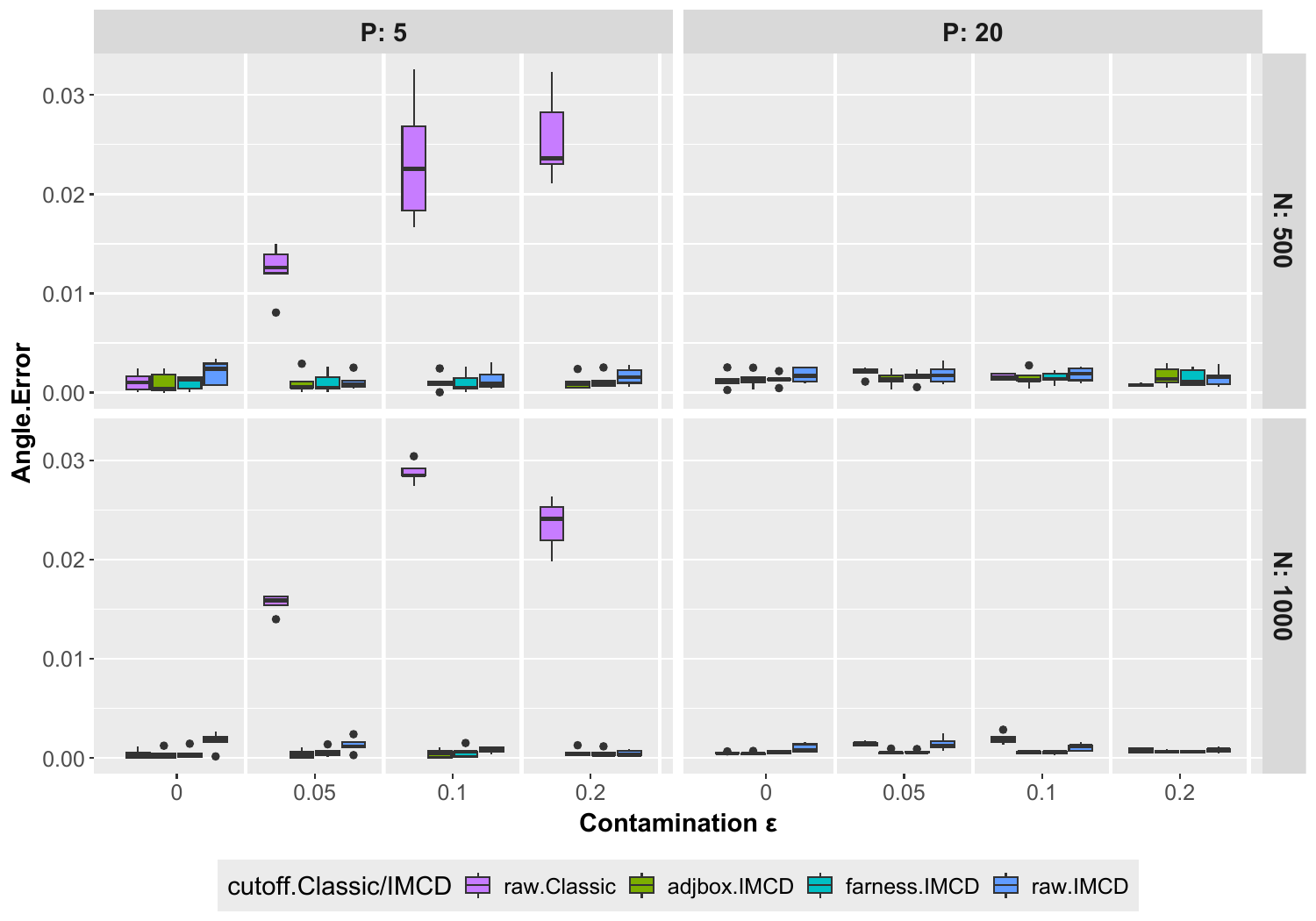}
        \caption{Scenario 1.} 
    \end{subfigure}
    \hfill
    \begin{subfigure}[b]{0.49\textwidth}
        \centering
        \includegraphics[width=\textwidth]{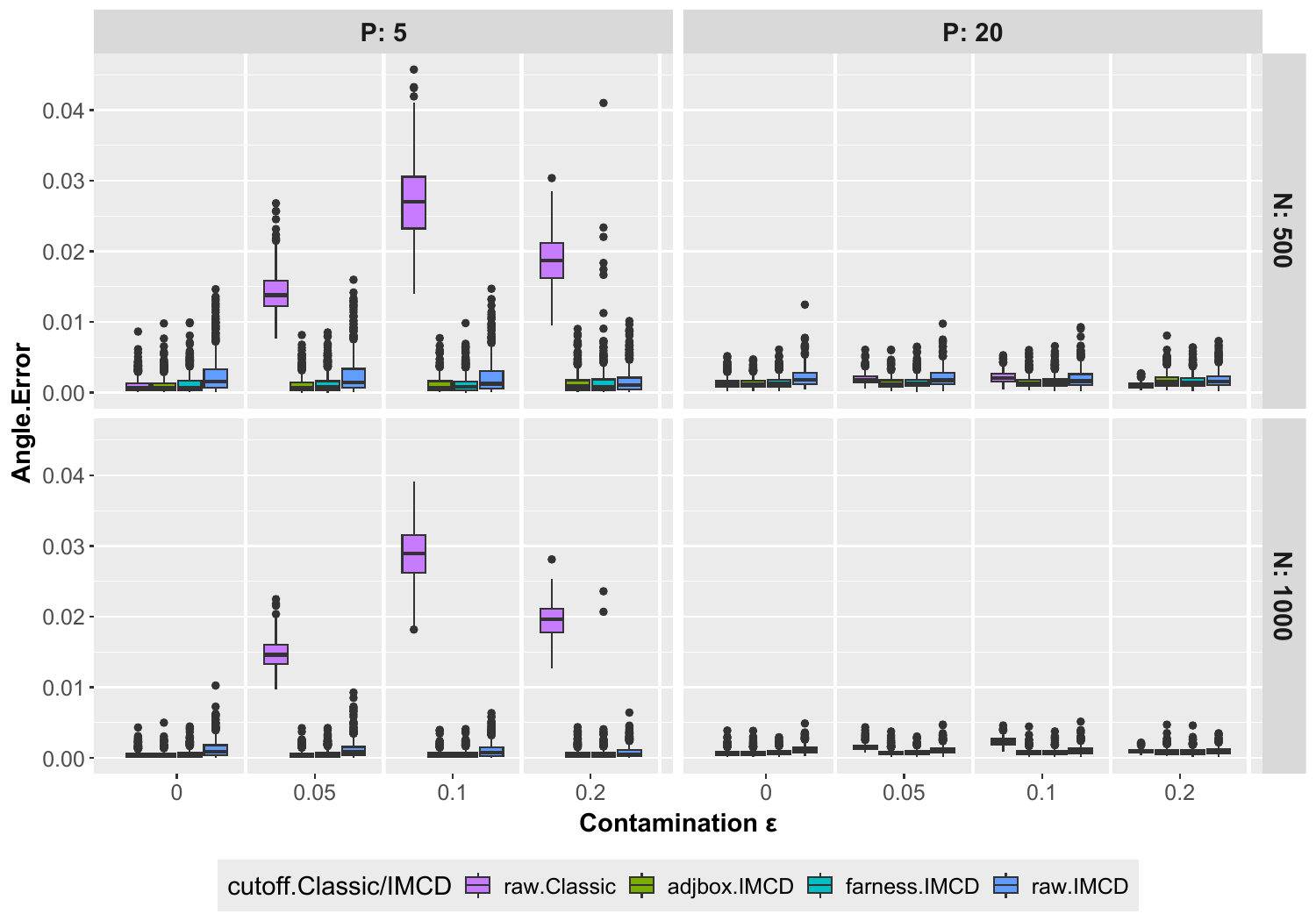}
        \caption{Scenario 4.} 
    \end{subfigure}
    \begin{subfigure}[b]{0.49\textwidth}
        \centering
        \includegraphics[width=\textwidth]{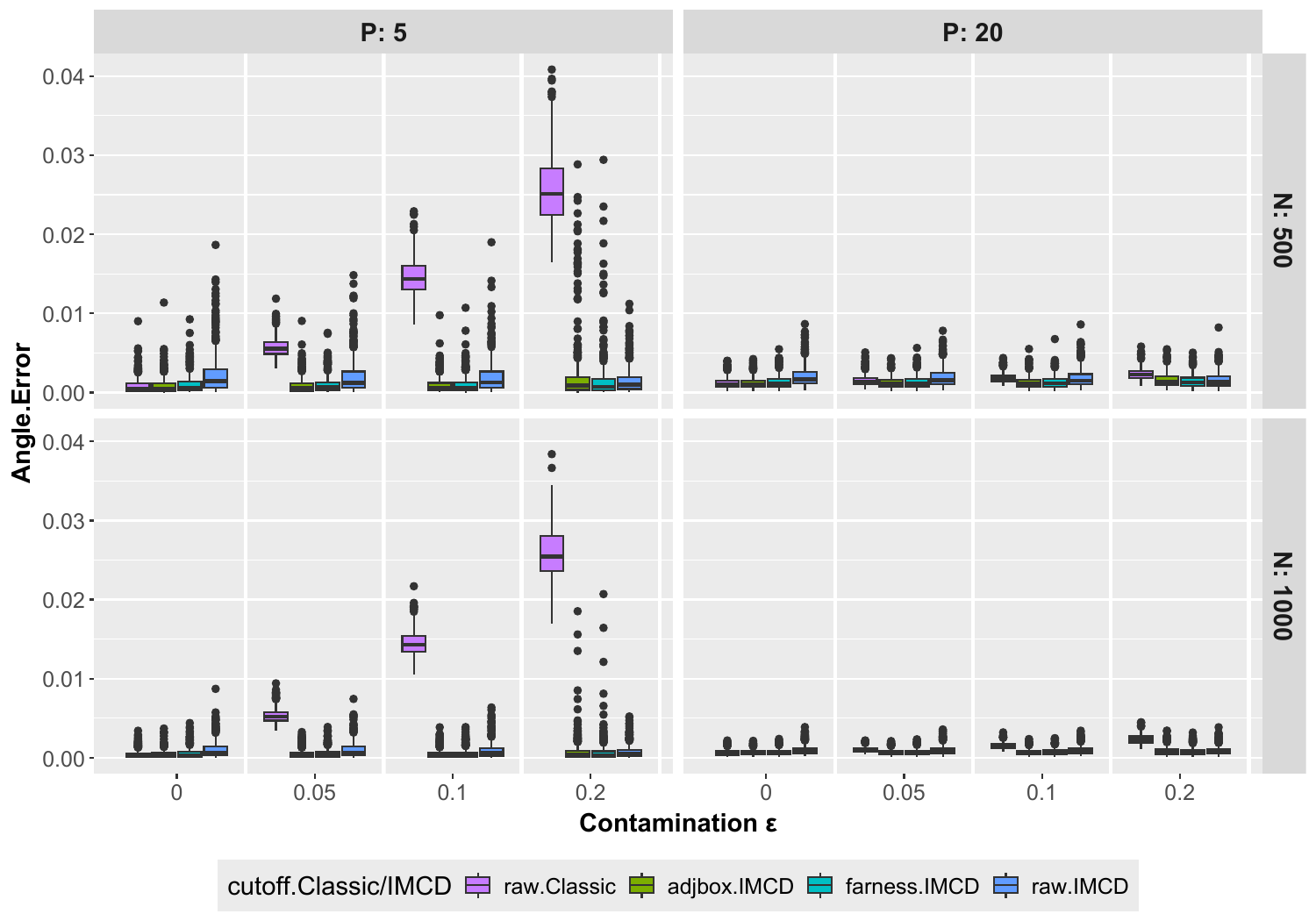}
        \caption{Scenario 2.} 
    \end{subfigure}
    \hfill
    \begin{subfigure}[b]{0.49\textwidth}
        \centering
        \includegraphics[width=\textwidth]{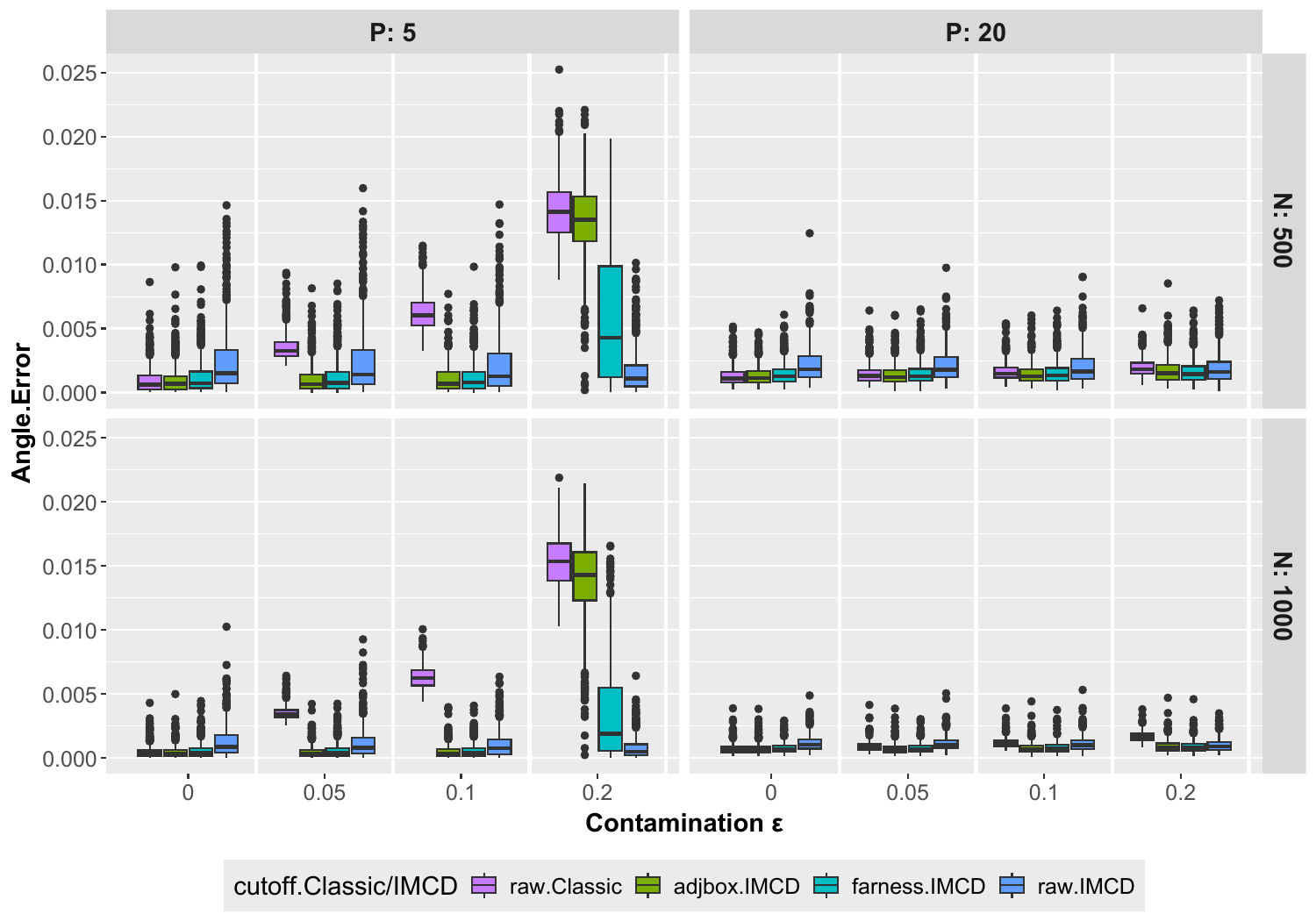}
        \caption{Scenario 5.} 
    \end{subfigure}
    \begin{subfigure}[b]{0.49\textwidth}
        \centering
        \includegraphics[width=\textwidth]{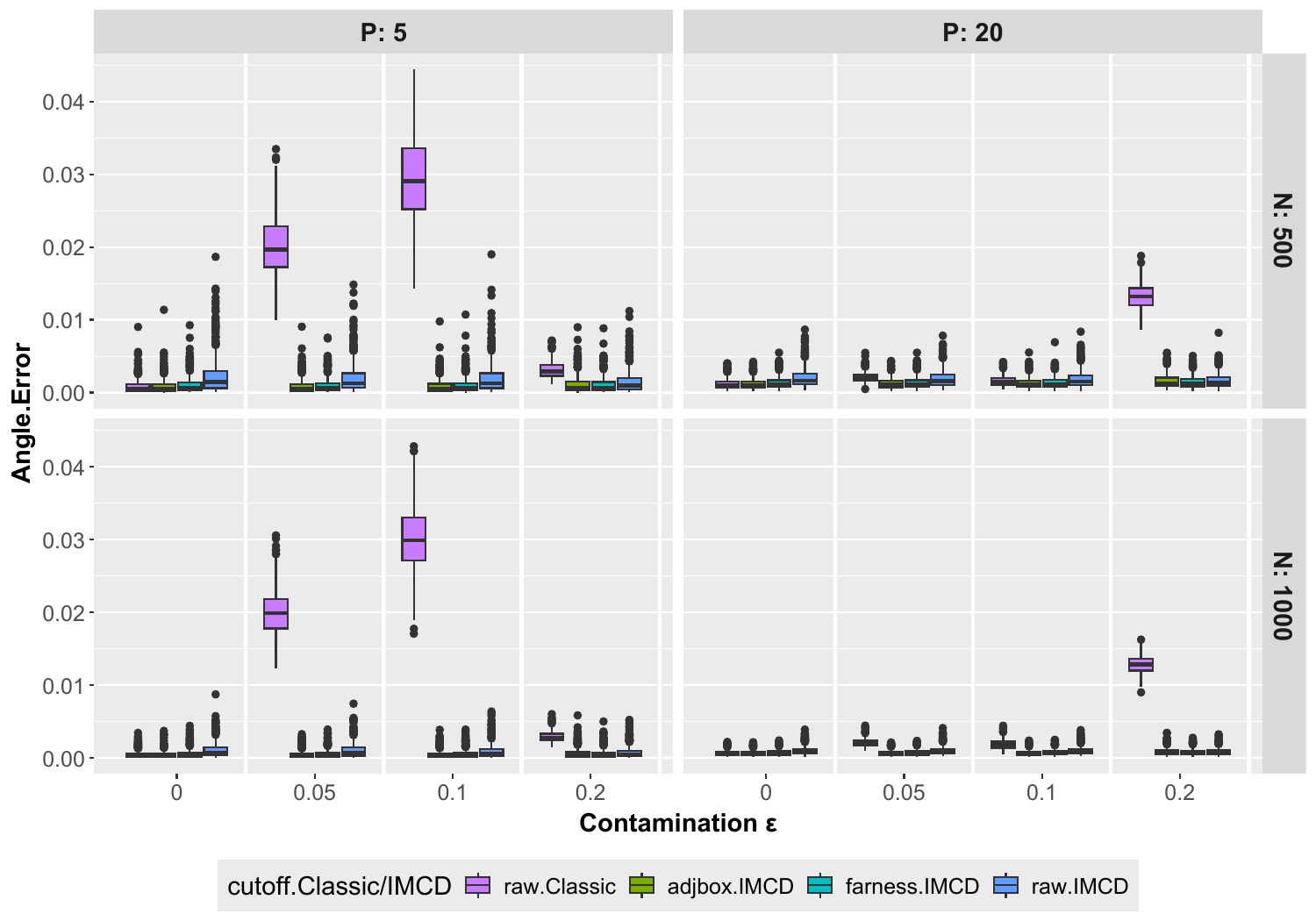}
        \caption{Scenario 3.} 
    \end{subfigure}
    \hfill
    \begin{subfigure}[b]{0.49\textwidth}
        \centering
        \includegraphics[width=\textwidth]{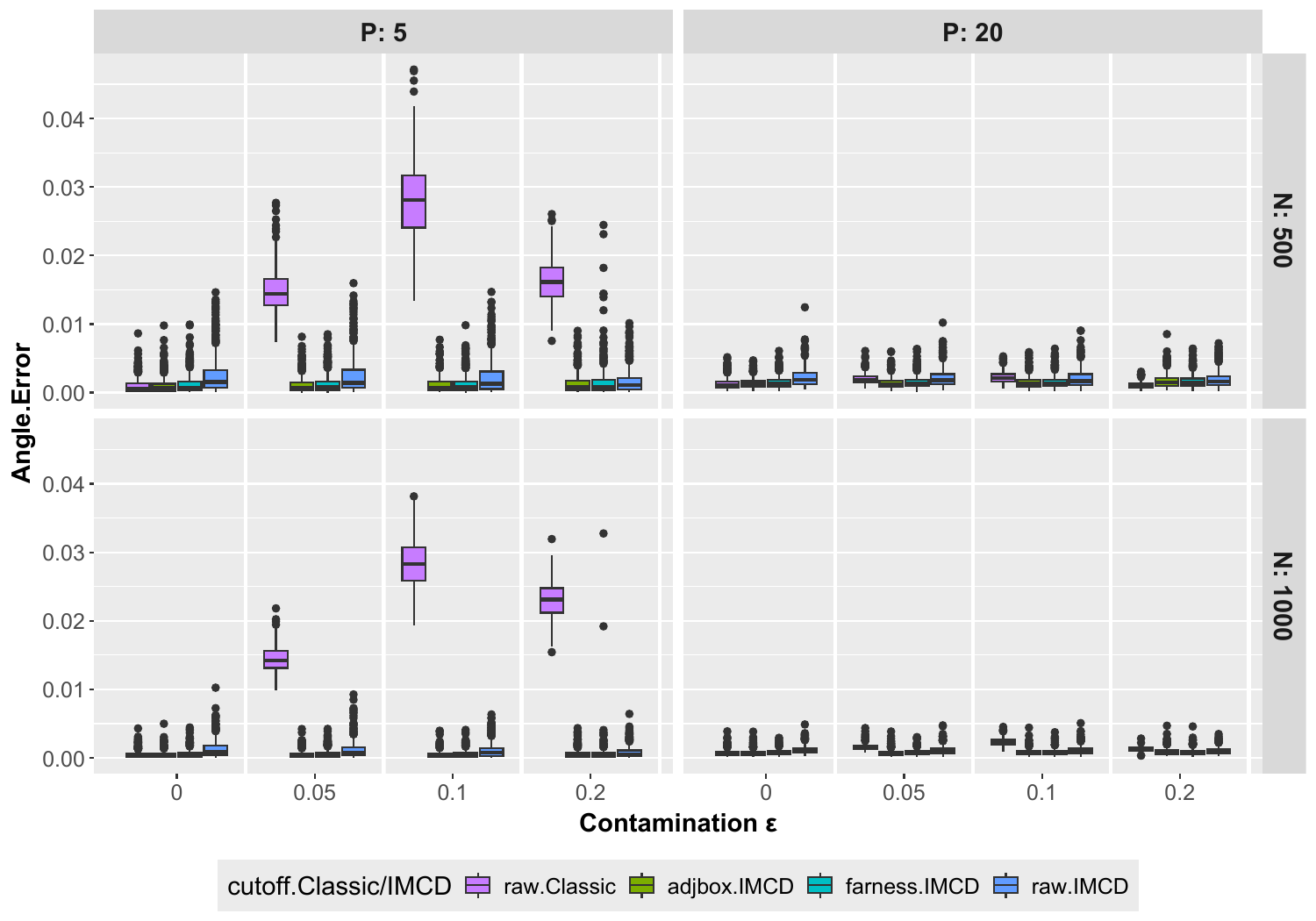}
        \caption{Scenario 6.} 
    \end{subfigure}
    \caption{Boxplots of the angle error \eqref{eq:angle_error} obtained for the six different scenarios, levels of contamination ($\epsilon$), number of variables ($P$), and sample size ($N$). For each case, we have four covariance matrix estimators: classic without reweighting (raw.Classic), IMCD with adjusted boxplot (adjbox.IMCD) and farness (farness.IMCD) reweighting, and IMCD without reweighting (raw.IMCD).}
    \label{fig:angle.error}
\end{figure}

\begin{figure}[ht]
    \centering
    \begin{subfigure}[b]{0.49\textwidth}
        \centering
        \includegraphics[width=\textwidth]{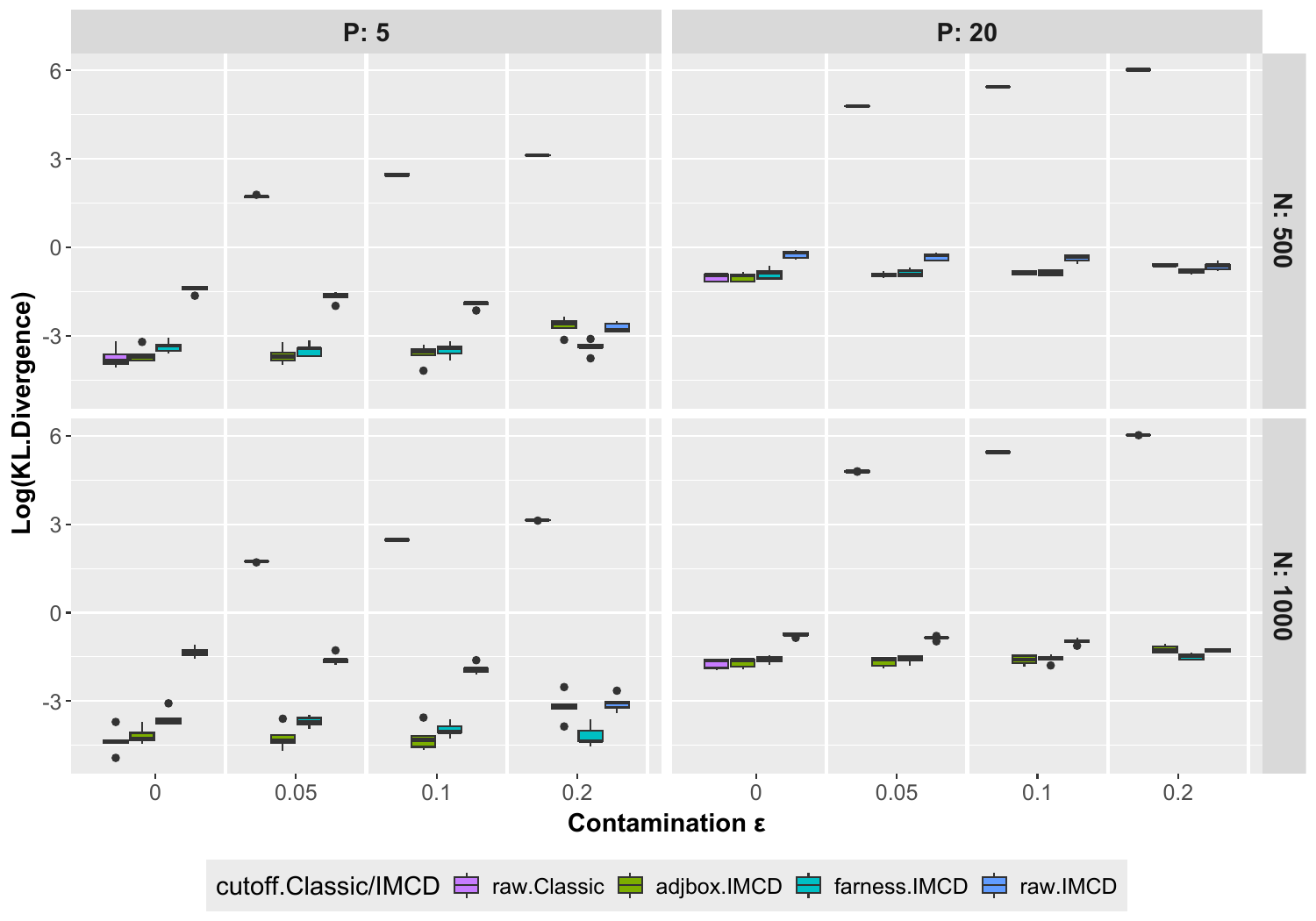}
        \caption{Scenario 1.} 
    \end{subfigure}
    \hfill
    \begin{subfigure}[b]{0.49\textwidth}
        \centering
        \includegraphics[width=\textwidth]{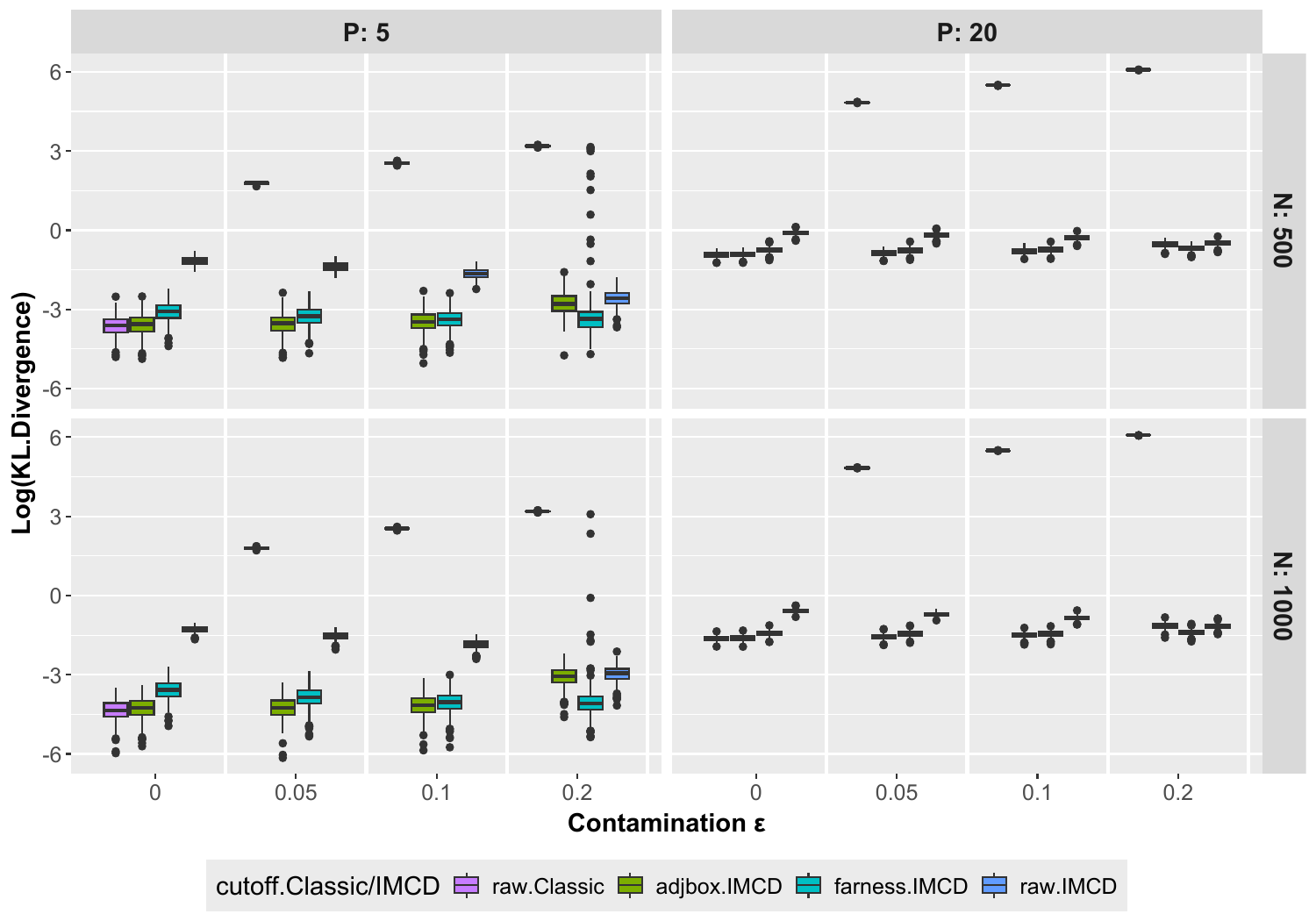}
        \caption{Scenario 4.} 
    \end{subfigure}
    \begin{subfigure}[b]{0.49\textwidth}
        \centering
        \includegraphics[width=\textwidth]{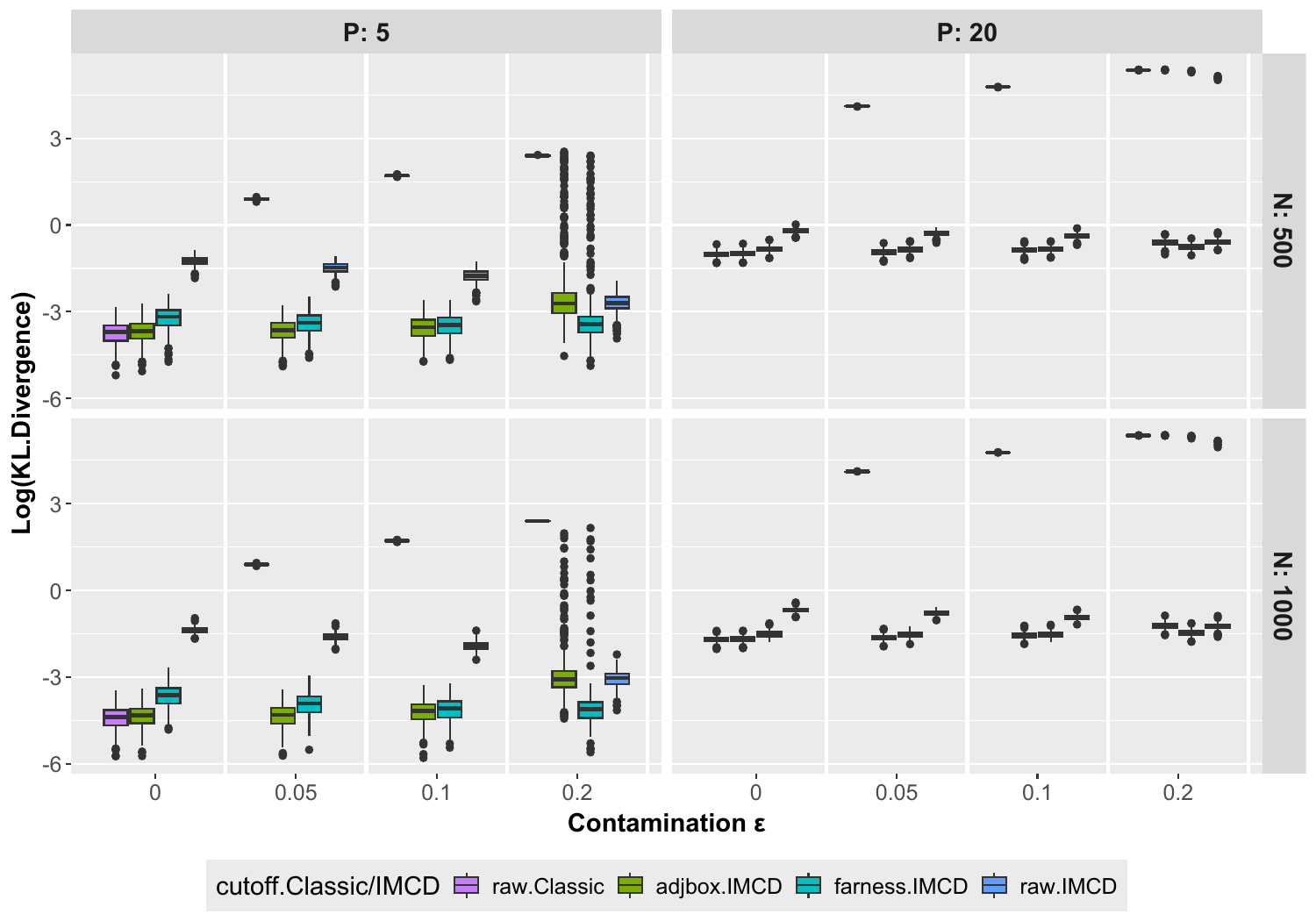}
        \caption{Scenario 2.} 
    \end{subfigure}
    \hfill
    \begin{subfigure}[b]{0.49\textwidth}
        \centering
        \includegraphics[width=\textwidth]{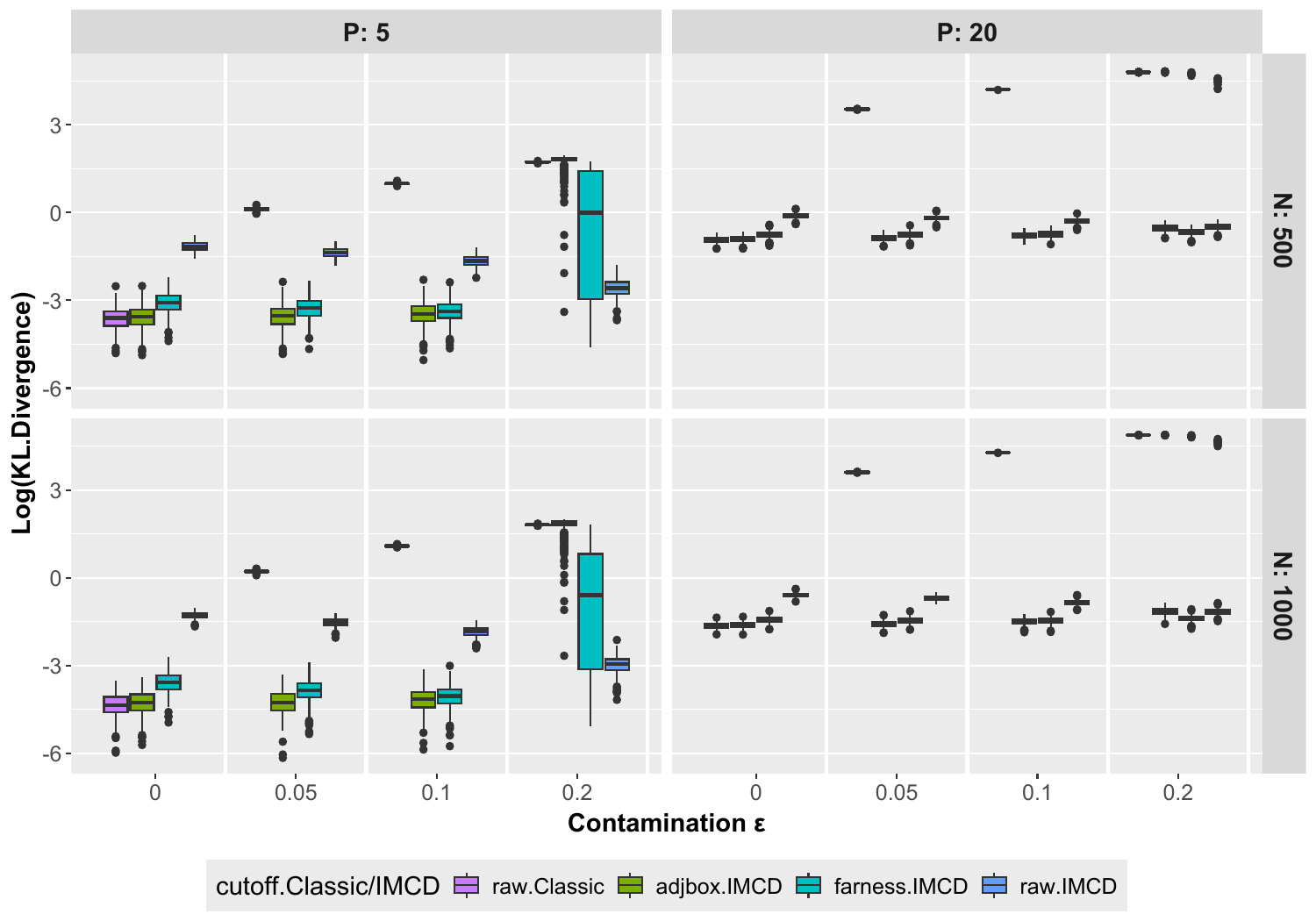}
        \caption{Scenario 5.} 
    \end{subfigure}
    \begin{subfigure}[b]{0.49\textwidth}
        \centering
        \includegraphics[width=\textwidth]{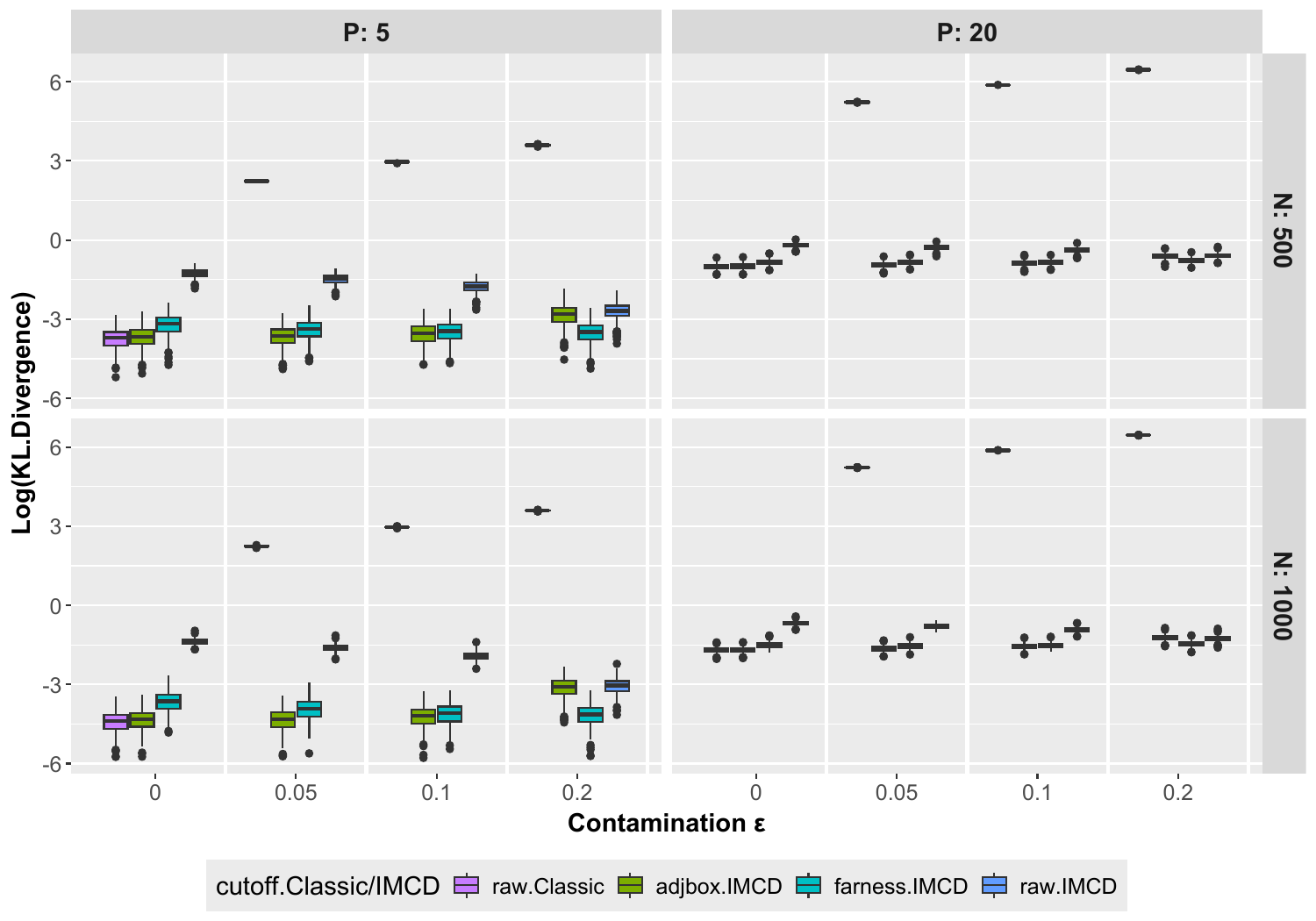}
        \caption{Scenario 3.} 
    \end{subfigure}
    \hfill
    \begin{subfigure}[b]{0.49\textwidth}
        \centering
        \includegraphics[width=\textwidth]{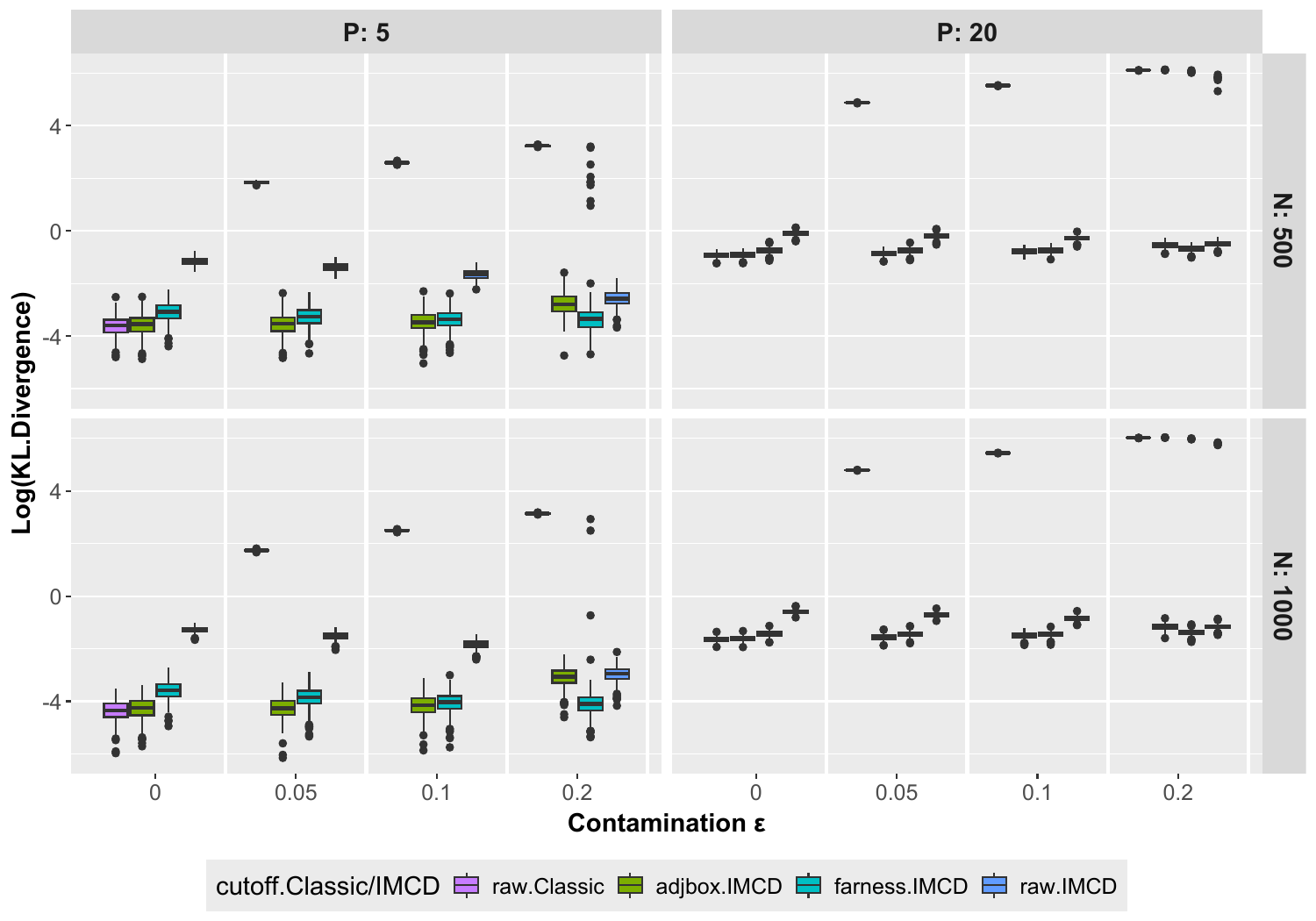}
        \caption{Scenario 6.} 
    \end{subfigure}
    \caption{Boxplots of the logarithm of the KL divergence \eqref{eq:kl_divergence} obtained for the six different scenarios, levels of contamination ($\epsilon$), number of variables ($P$), and sample size ($N$). For each case, we have four covariance matrix estimators: classic without reweighting (raw.Classic), IMCD with adjusted boxplot (adjbox.IMCD) and farness (farness.IMCD) reweighting, and IMCD without reweighting (raw.IMCD).}
    \label{fig:kl.divergence}
\end{figure}

\begin{figure}[ht]
    \centering
    \begin{subfigure}[b]{0.49\textwidth}
        \centering
        \includegraphics[width=\textwidth]{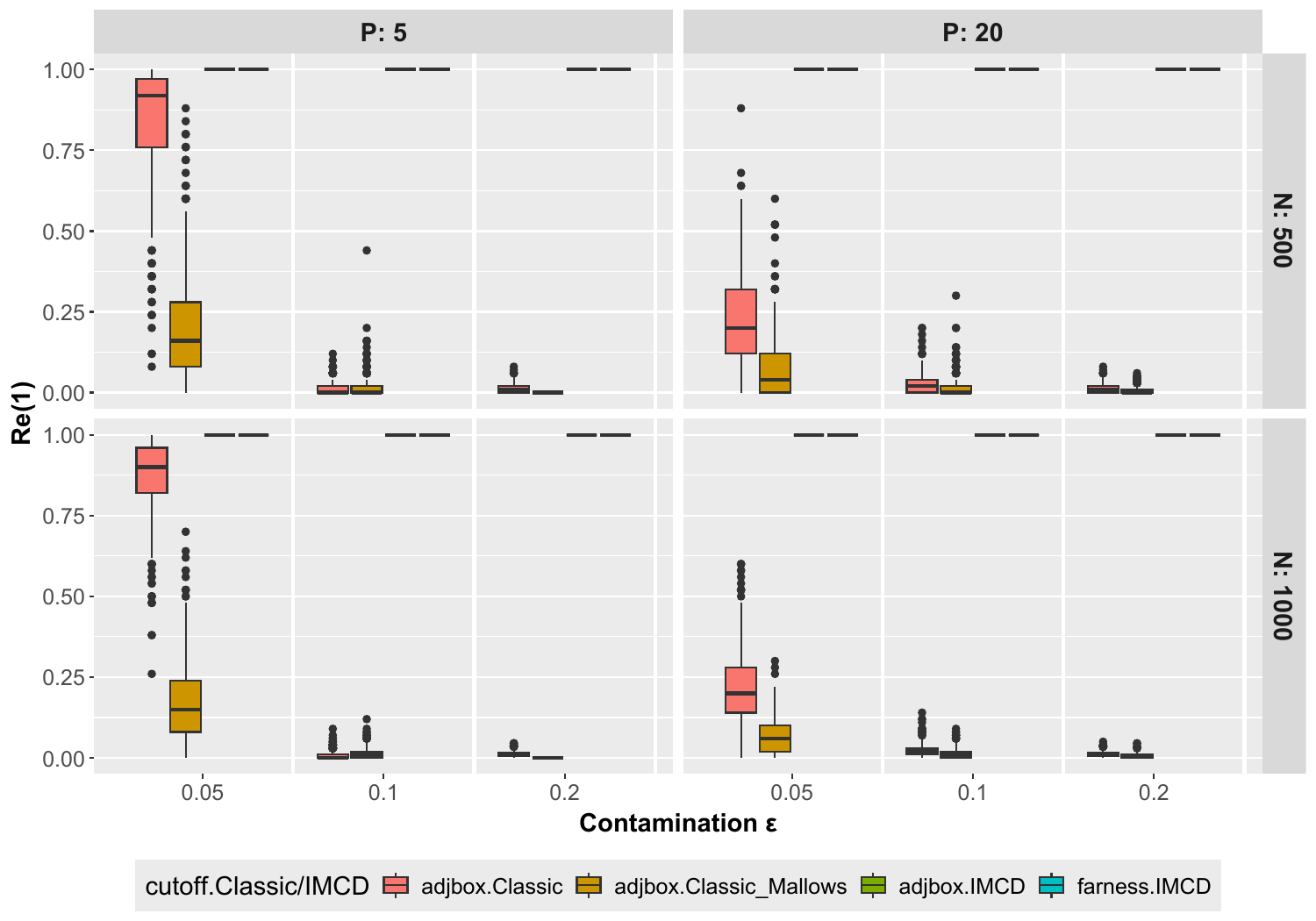}
        \caption{Scenario 3.} 
    \end{subfigure}
    \hfill
    \begin{subfigure}[b]{0.49\textwidth}
        \centering
        \includegraphics[width=\textwidth]{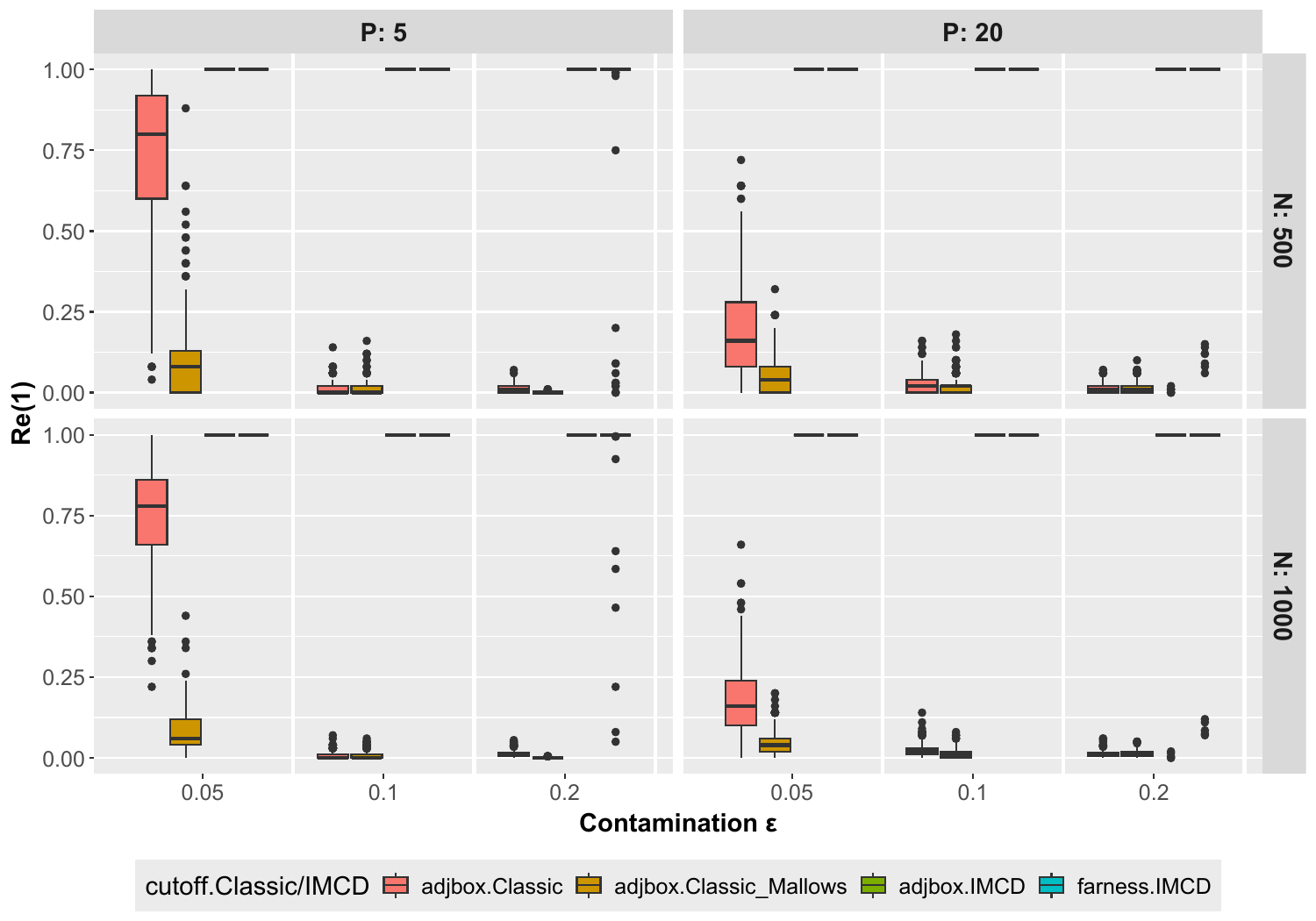}
        \caption{Scenario 6.} 
    \end{subfigure}
    \caption{Boxplots of the recall of class 1 (outliers) obtained for scenarios 3 and 6, the different levels of contamination ($\epsilon$), number of variables ($P$), and sample size ($N$). For each case, we have four outlier detection methods: classic Interval-Mahalanobis distance with adjusted boxplot cutoff (adjbox.Classic), Mallows distance with adjusted boxplot cutoff (adjbox.Classic\_Mallows), robust Interval-Mahalanobis distance with adjusted boxplot (adjbox.IMCD), and farness (farness.IMCD) reweighting/cutoff.}
    \label{fig:recall1_2}
\end{figure}

\begin{figure}[ht]
    \centering
    \begin{subfigure}[b]{0.49\textwidth}
        \centering
        \includegraphics[width=\textwidth]{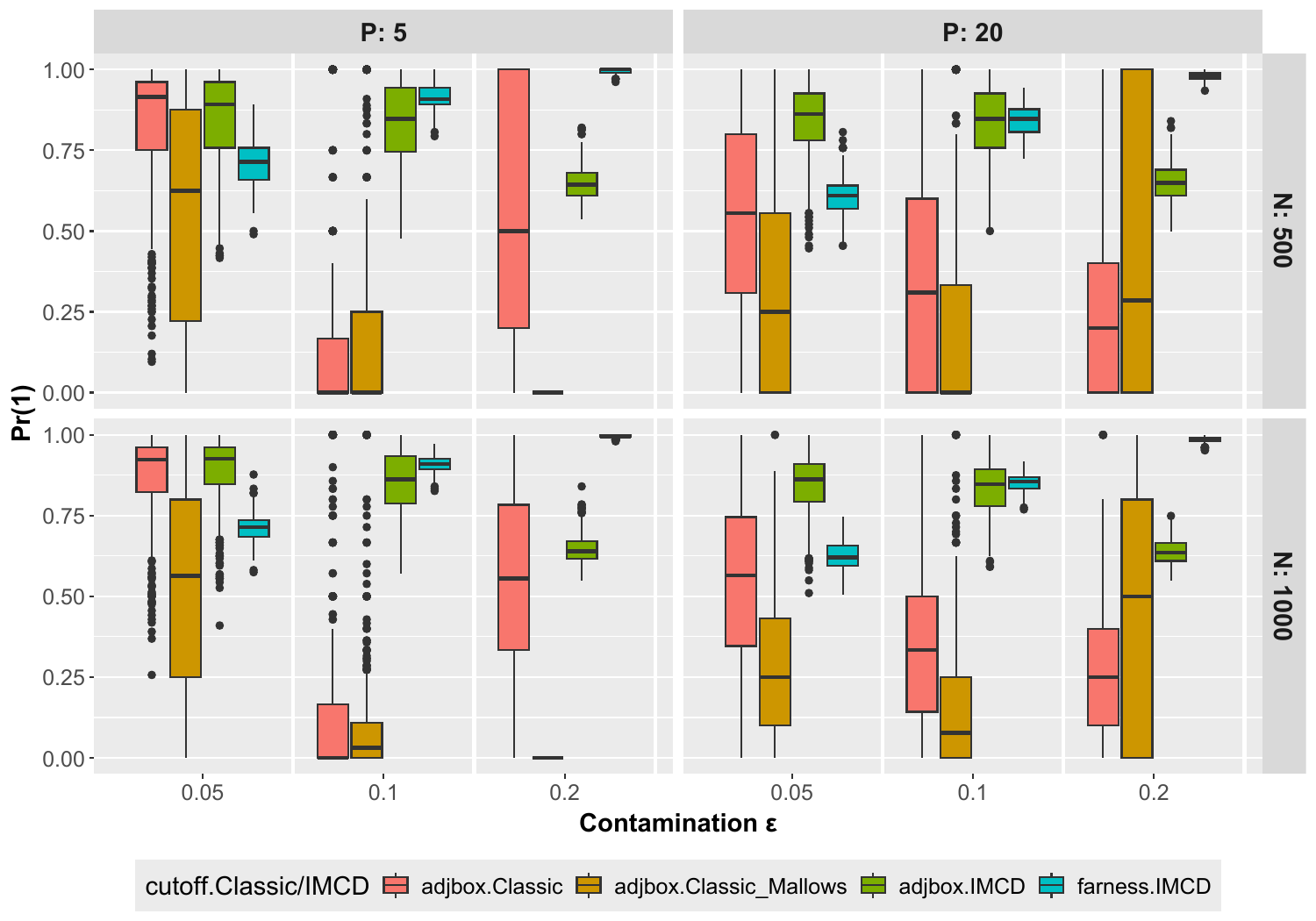}
        \caption{Scenario 3.} 
    \end{subfigure}
    \hfill
    \begin{subfigure}[b]{0.49\textwidth}
        \centering
        \includegraphics[width=\textwidth]{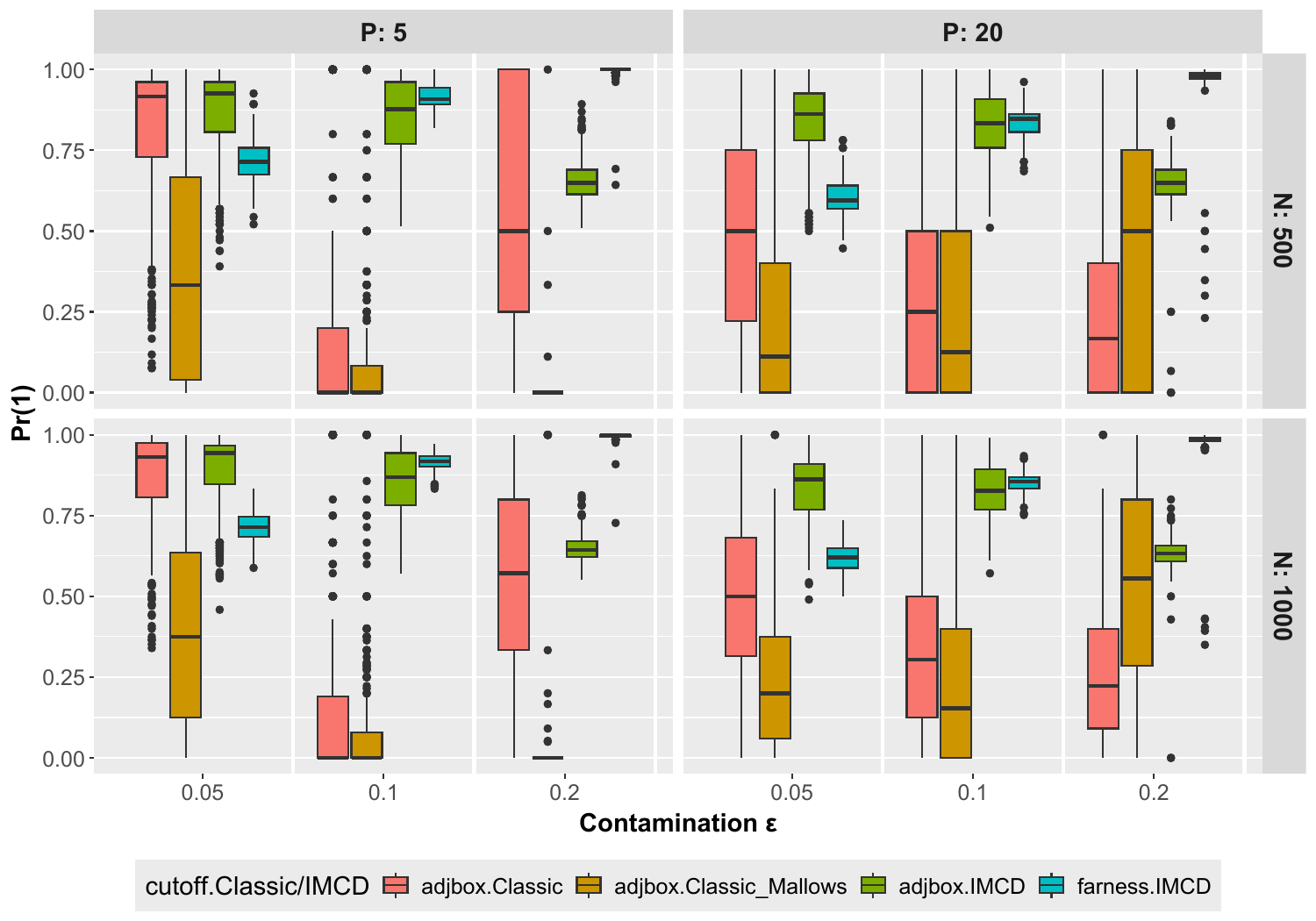}
        \caption{Scenario 6.} 
    \end{subfigure}
    \caption{Boxplots of the precision of class 1 (outliers) obtained for scenarios 3 and 6, the different levels of contamination ($\epsilon$), number of variables ($P$), and sample size ($N$). For each case, we have four outlier detection methods: classic Interval-Mahalanobis distance with adjusted boxplot cutoff (adjbox.Classic), Mallows distance with adjusted boxplot cutoff (adjbox.Classic\_Mallows), robust Interval-Mahalanobis distance with adjusted boxplot (adjbox.IMCD), and farness (farness.IMCD) reweighting/cutoff.}
    \label{fig:precision1_2}
\end{figure}

\begin{figure}[ht]
    \centering
    \begin{subfigure}[b]{0.49\textwidth}
        \centering
        \includegraphics[width=\textwidth]{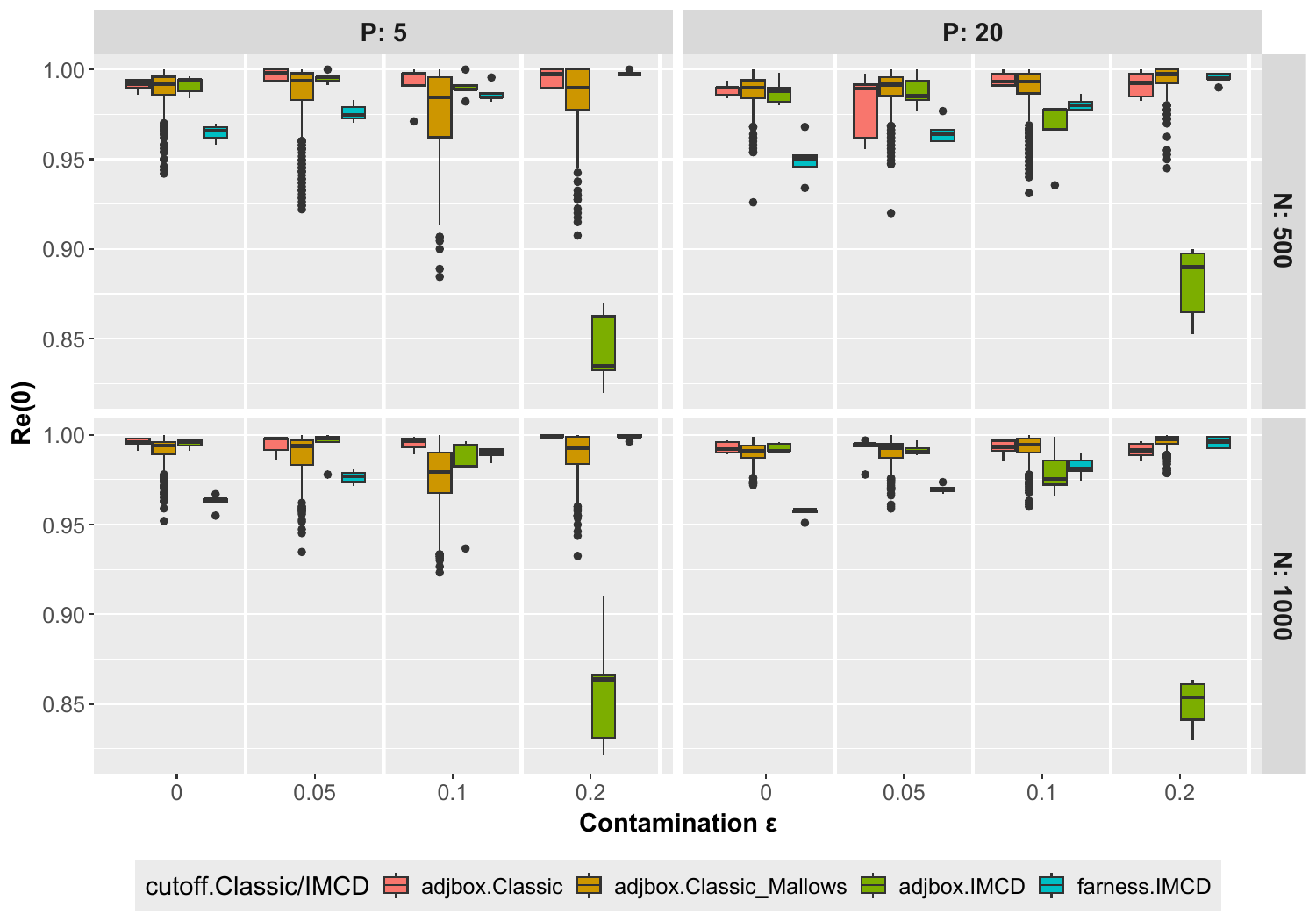}
        \caption{Scenario 1.} 
    \end{subfigure}
    \hfill
    \begin{subfigure}[b]{0.49\textwidth}
        \centering
        \includegraphics[width=\textwidth]{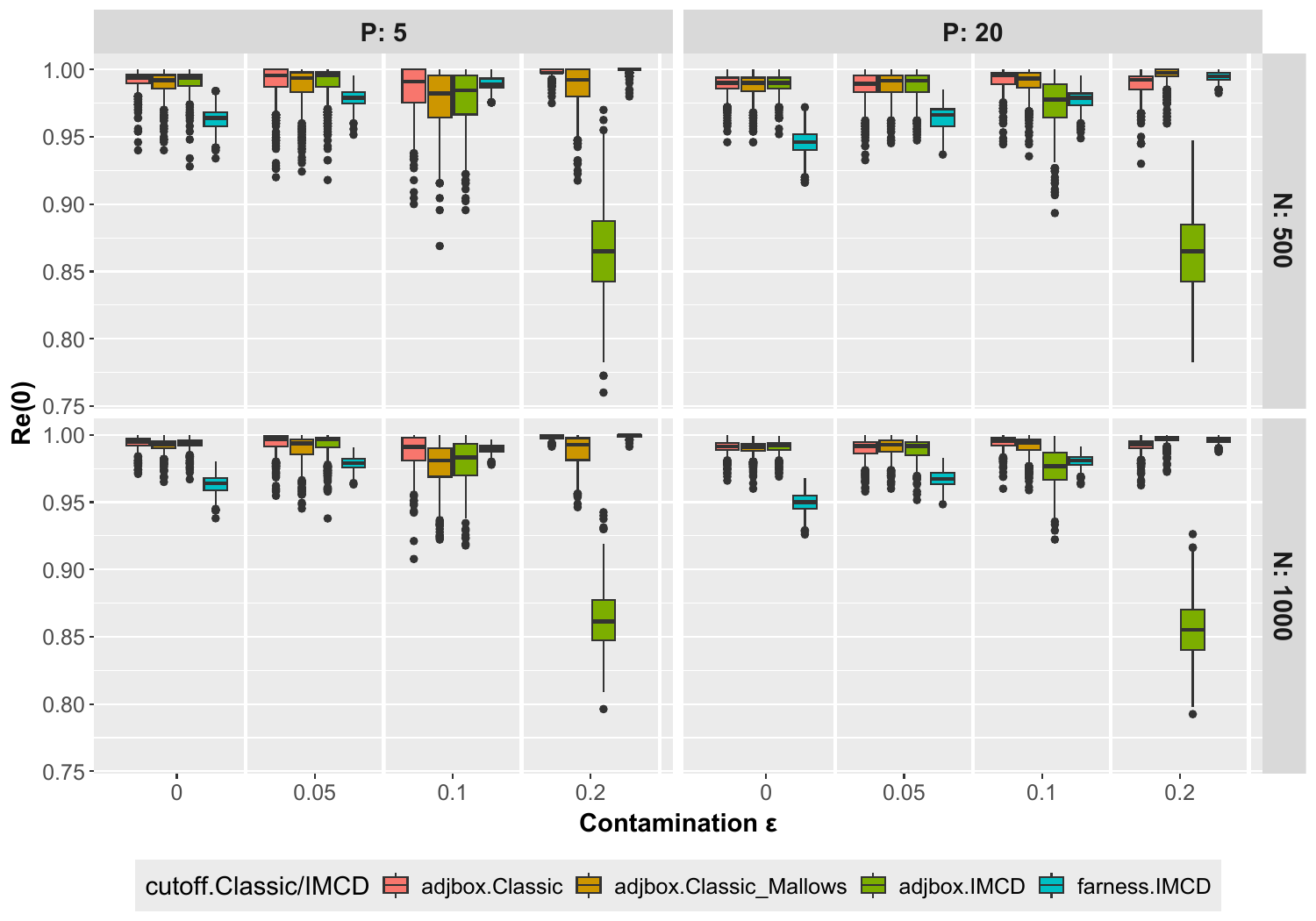}
        \caption{Scenario 4.} 
    \end{subfigure}
    \begin{subfigure}[b]{0.49\textwidth}
        \centering
        \includegraphics[width=\textwidth]{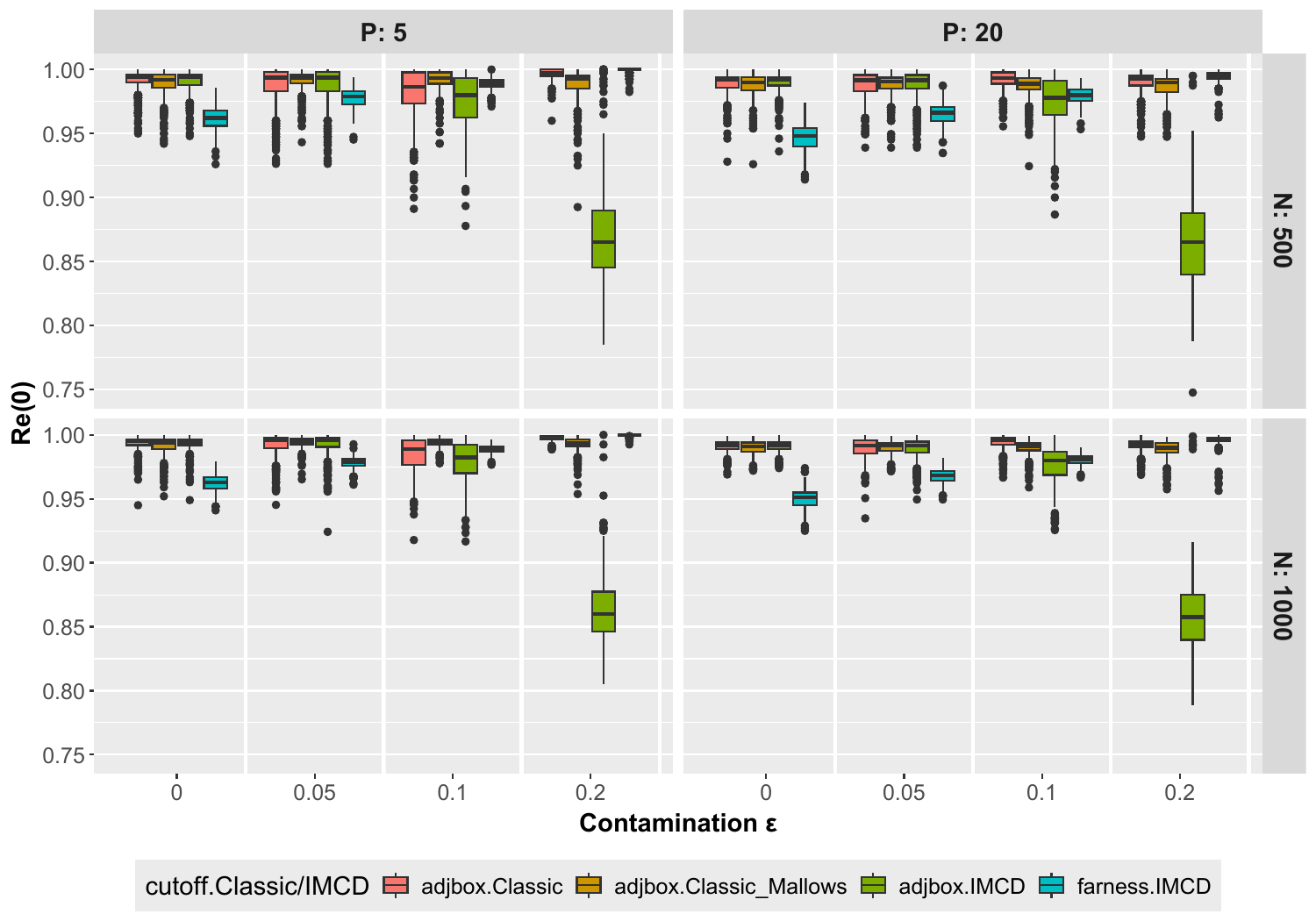}
        \caption{Scenario 2.} 
    \end{subfigure}
    \hfill
    \begin{subfigure}[b]{0.49\textwidth}
        \centering
        \includegraphics[width=\textwidth]{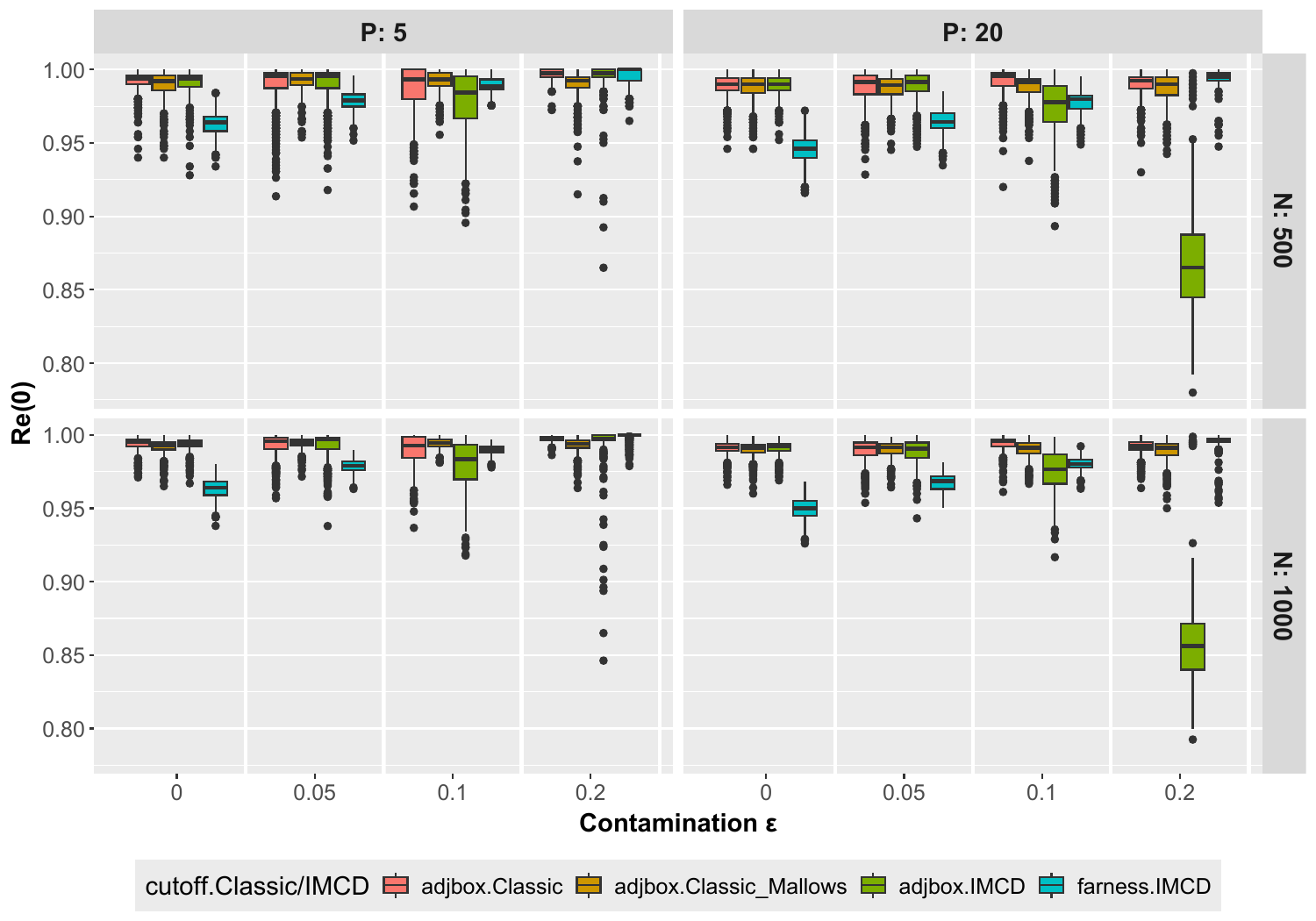}
        \caption{Scenario 5.} 
    \end{subfigure}
    \begin{subfigure}[b]{0.49\textwidth}
        \centering
        \includegraphics[width=\textwidth]{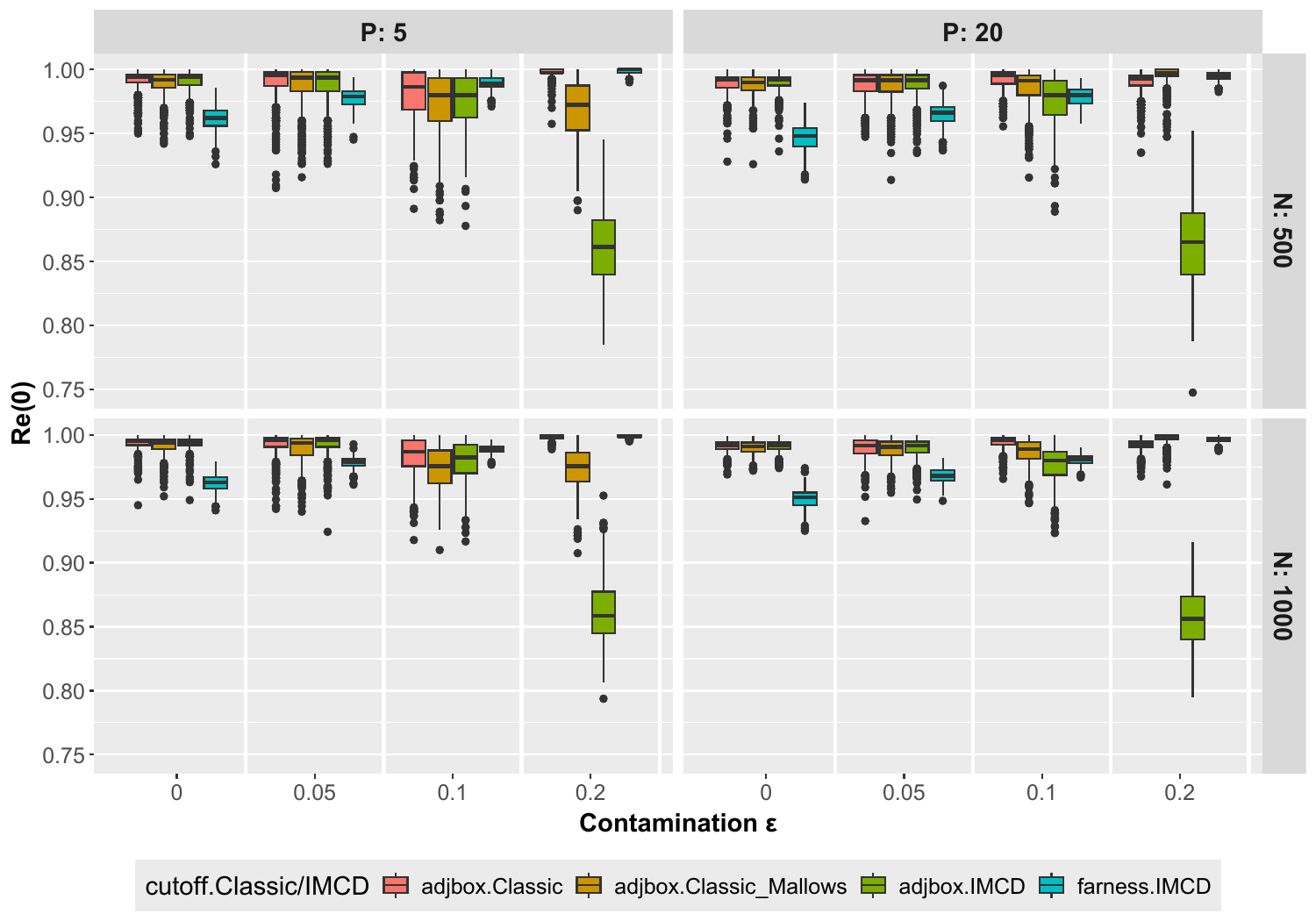}
        \caption{Scenario 3.} 
    \end{subfigure}
    \hfill
    \begin{subfigure}[b]{0.49\textwidth}
        \centering
        \includegraphics[width=\textwidth]{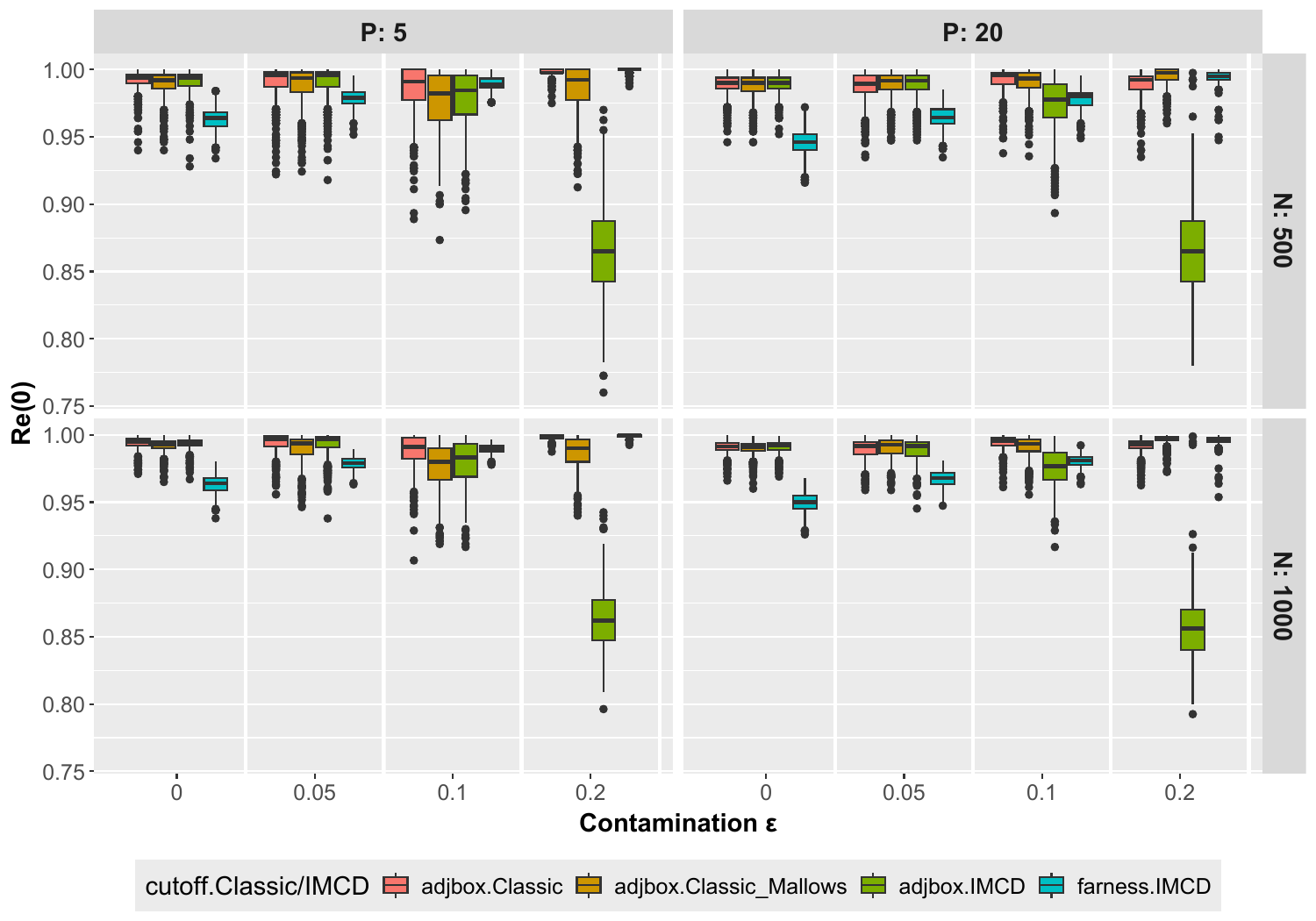}
        \caption{Scenario 6.} 
    \end{subfigure}
    \caption{Boxplots of the recall of class 0 (regular) \eqref{eq:pr_re_0} obtained for the six different scenarios, levels of contamination ($\epsilon$), number of variables ($P$), and sample size ($N$). For each case, we have four outlier detection methods: classic Interval-Mahalanobis distance with adjusted boxplot cutoff (adjbox.Classic), Mallows distance with adjusted boxplot cutoff (adjbox.Classic\_Mallows), robust Interval-Mahalanobis distance with adjusted boxplot (adjbox.IMCD), and farness (farness.IMCD) reweighting/cutoff.}
    \label{fig:recall0}
\end{figure}

\begin{figure}[ht]
    \centering
    \begin{subfigure}[b]{0.49\textwidth}
        \centering
        \includegraphics[width=\textwidth]{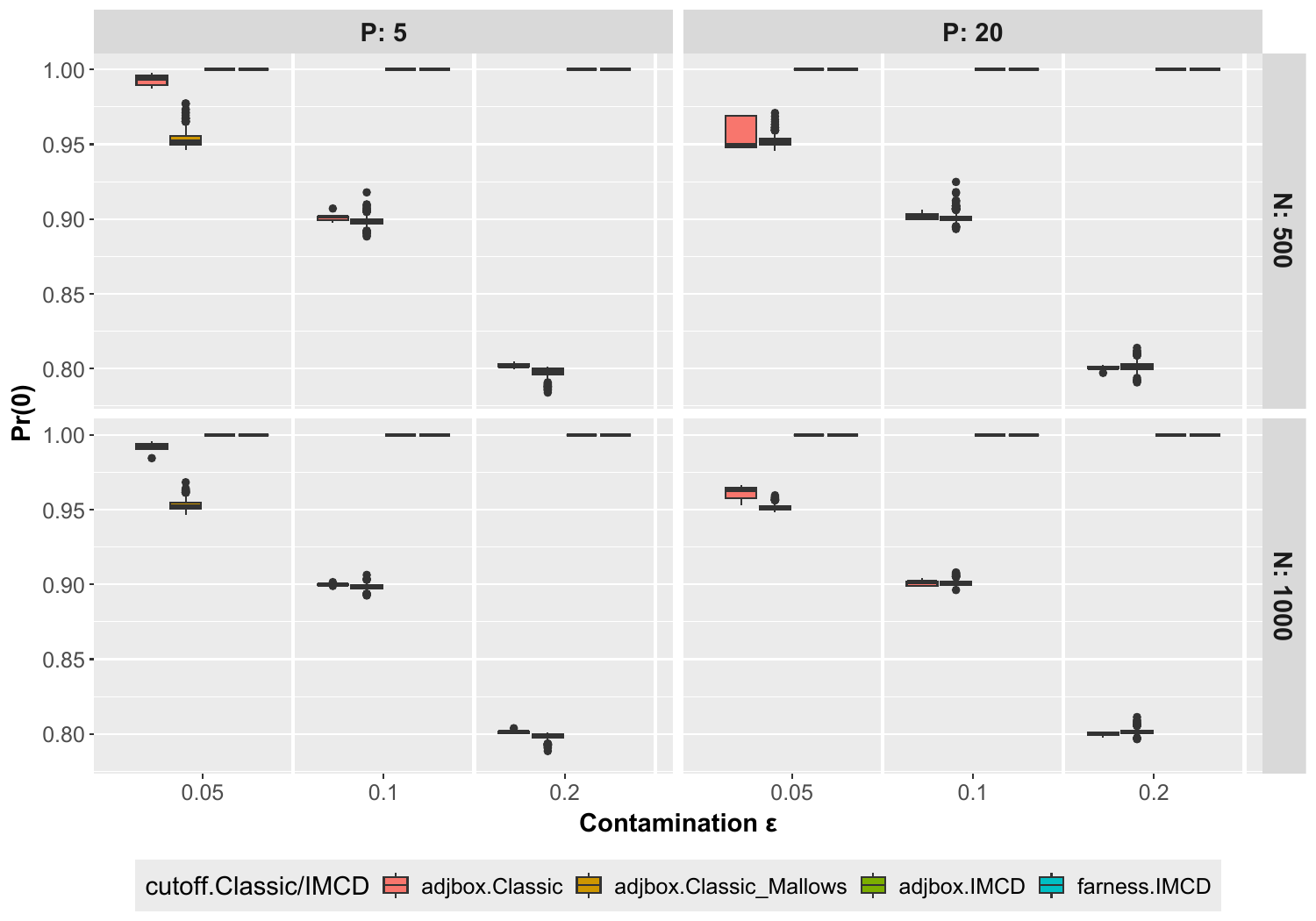}
        \caption{Scenario 1.} 
    \end{subfigure}
    \hfill
    \begin{subfigure}[b]{0.49\textwidth}
        \centering
        \includegraphics[width=\textwidth]{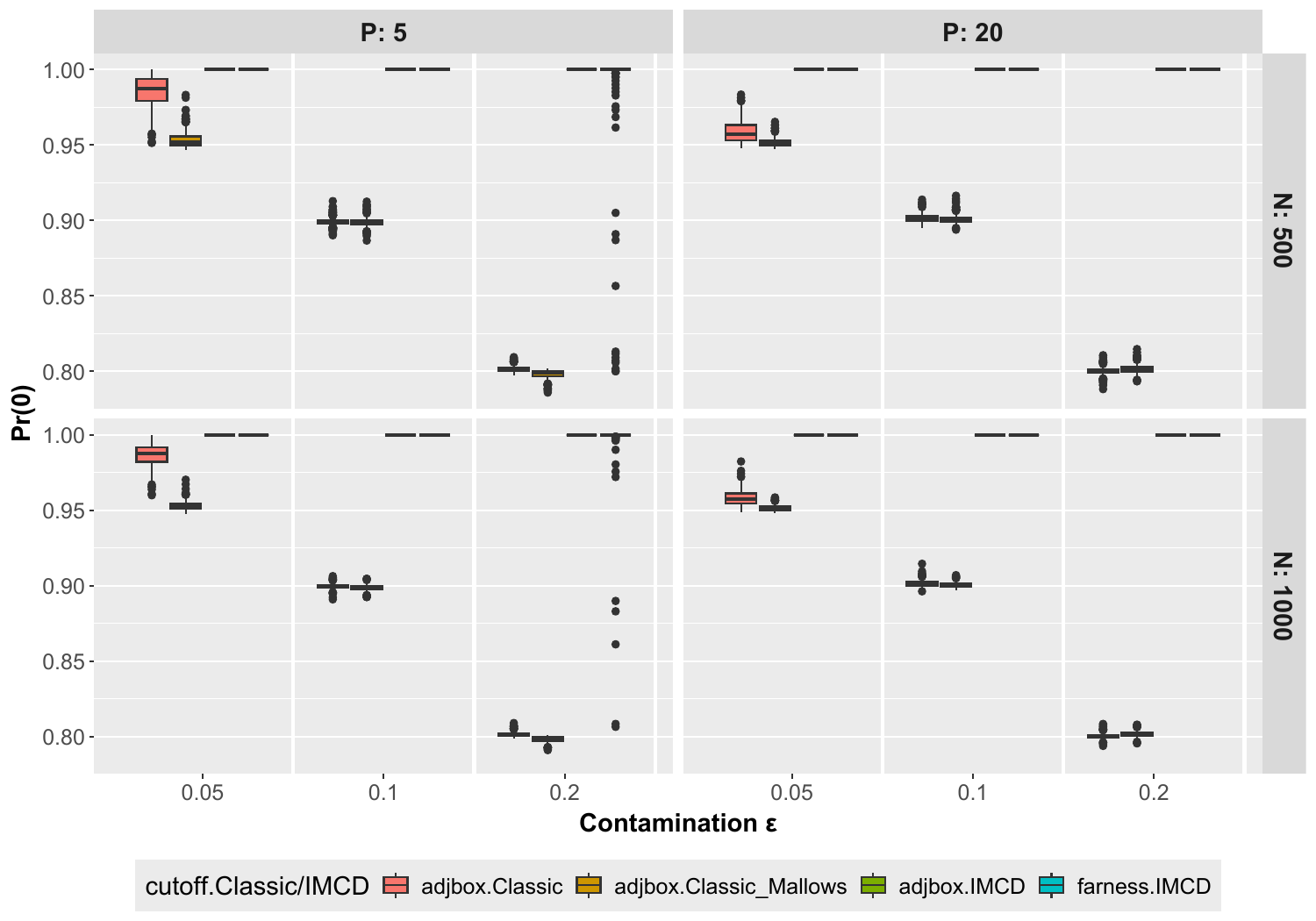}
        \caption{Scenario 4.} 
    \end{subfigure}
    \begin{subfigure}[b]{0.49\textwidth}
        \centering
        \includegraphics[width=\textwidth]{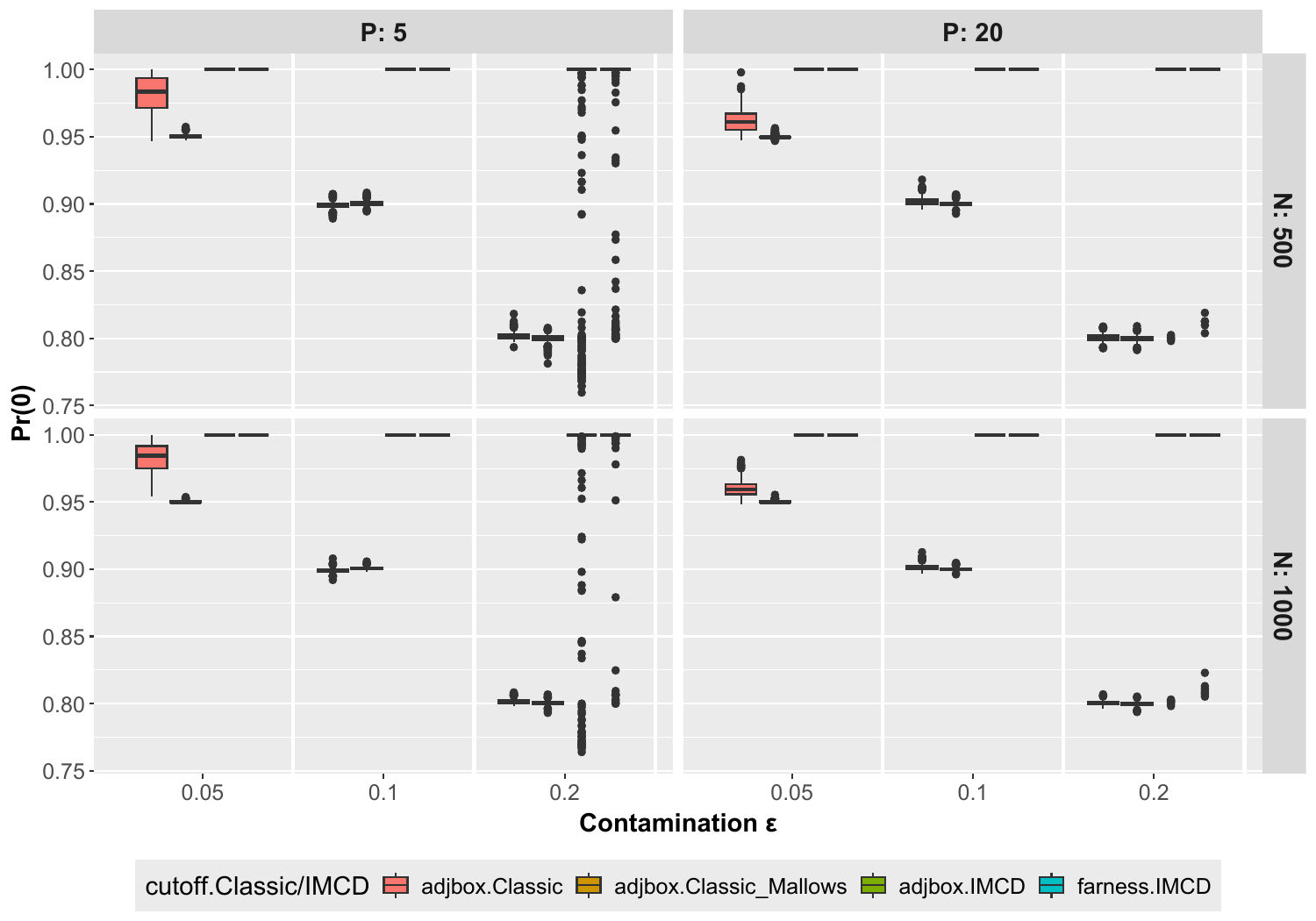}
        \caption{Scenario 2.} 
    \end{subfigure}
    \hfill
    \begin{subfigure}[b]{0.49\textwidth}
        \centering
        \includegraphics[width=\textwidth]{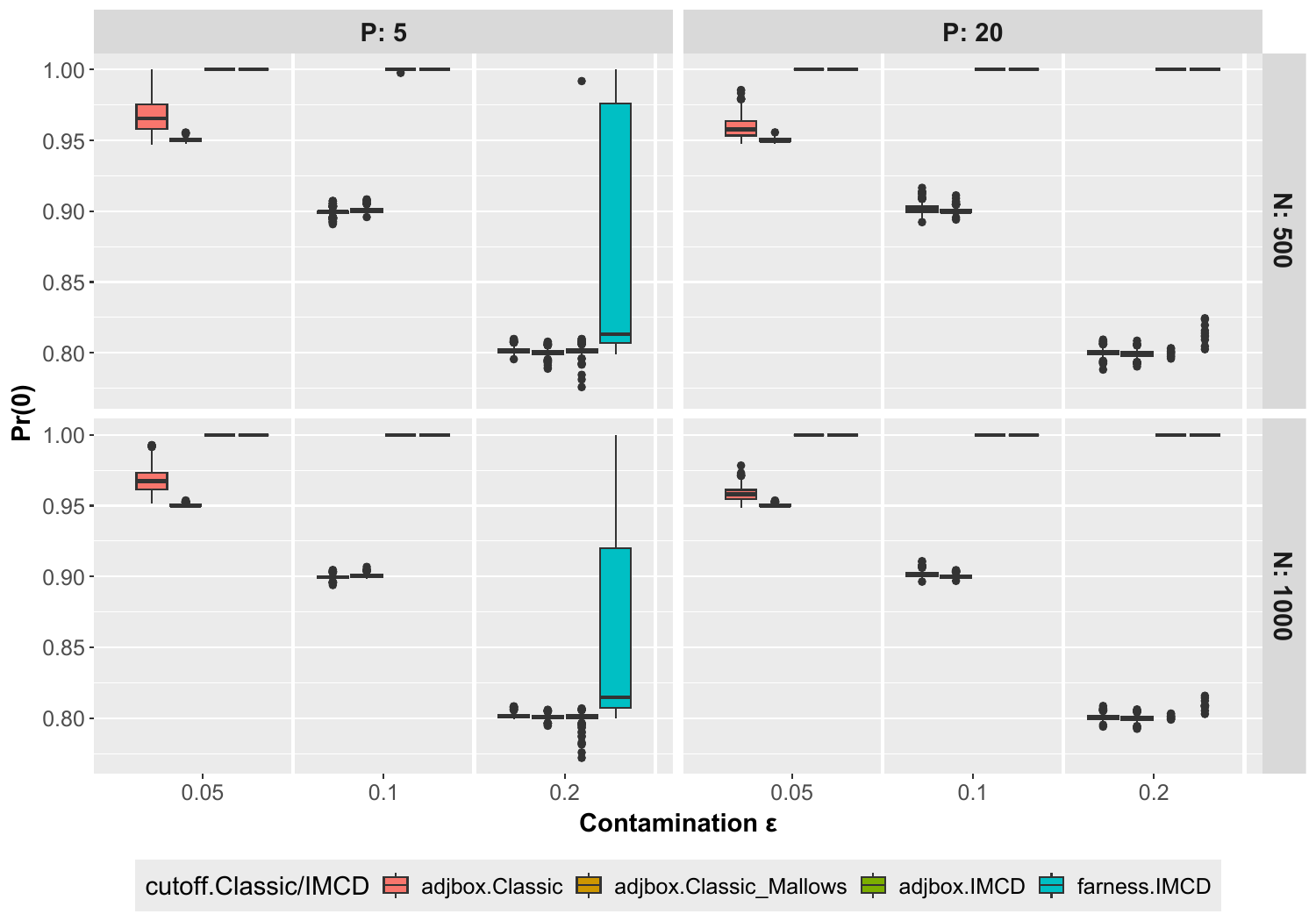}
        \caption{Scenario 5.} 
    \end{subfigure}
    \begin{subfigure}[b]{0.49\textwidth}
        \centering
        \includegraphics[width=\textwidth]{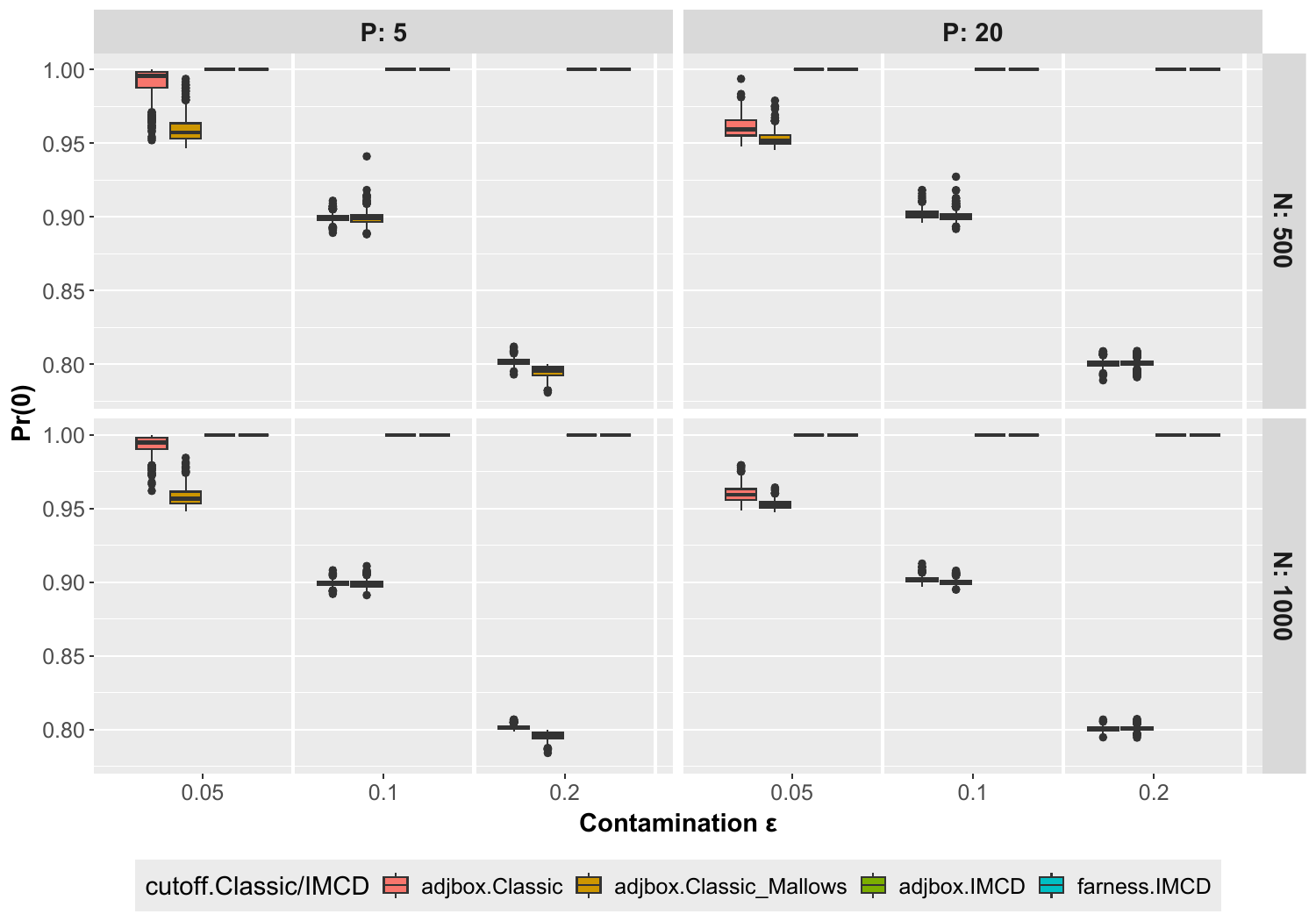}
        \caption{Scenario 3.} 
    \end{subfigure}
    \hfill
    \begin{subfigure}[b]{0.49\textwidth}
        \centering
        \includegraphics[width=\textwidth]{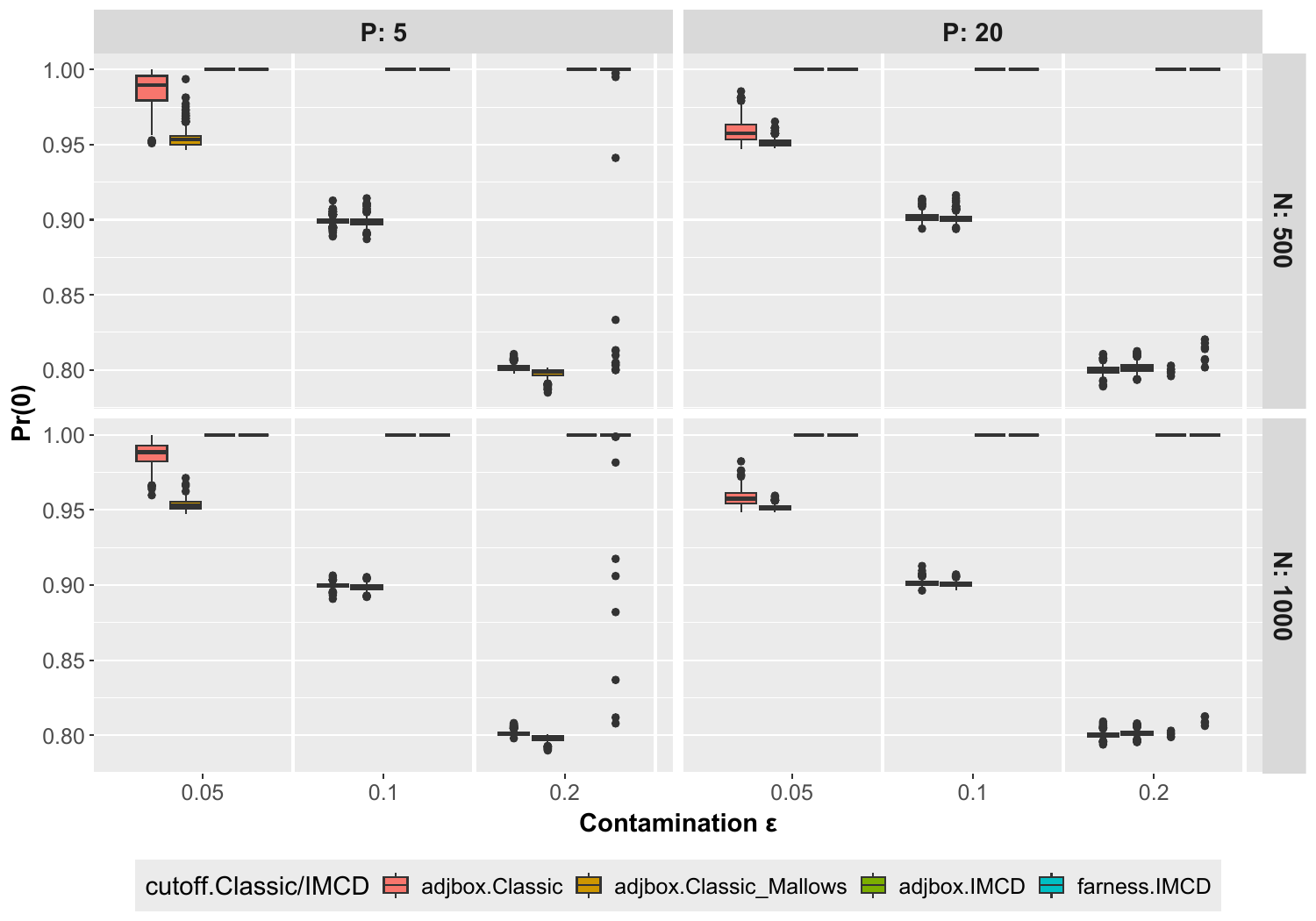}
        \caption{Scenario 6.} 
    \end{subfigure}
    \caption{Boxplots of the precision of class 0 (regular) \eqref{eq:pr_re_0} obtained for the six different scenarios, levels of contamination ($\epsilon$), number of variables ($P$), and sample size ($N$). For each case, we have four outlier detection methods: classic Interval-Mahalanobis distance with adjusted boxplot cutoff (adjbox.Classic), Mallows distance with adjusted boxplot cutoff (adjbox.Classic\_Mallows), robust Interval-Mahalanobis distance with adjusted boxplot (adjbox.IMCD), and farness (farness.IMCD) reweighting/cutoff.}
    \label{fig:precision0}
\end{figure}

\begin{figure}[ht]
    \centering
    \begin{subfigure}[b]{0.49\textwidth}
        \centering
        \includegraphics[width=\textwidth]{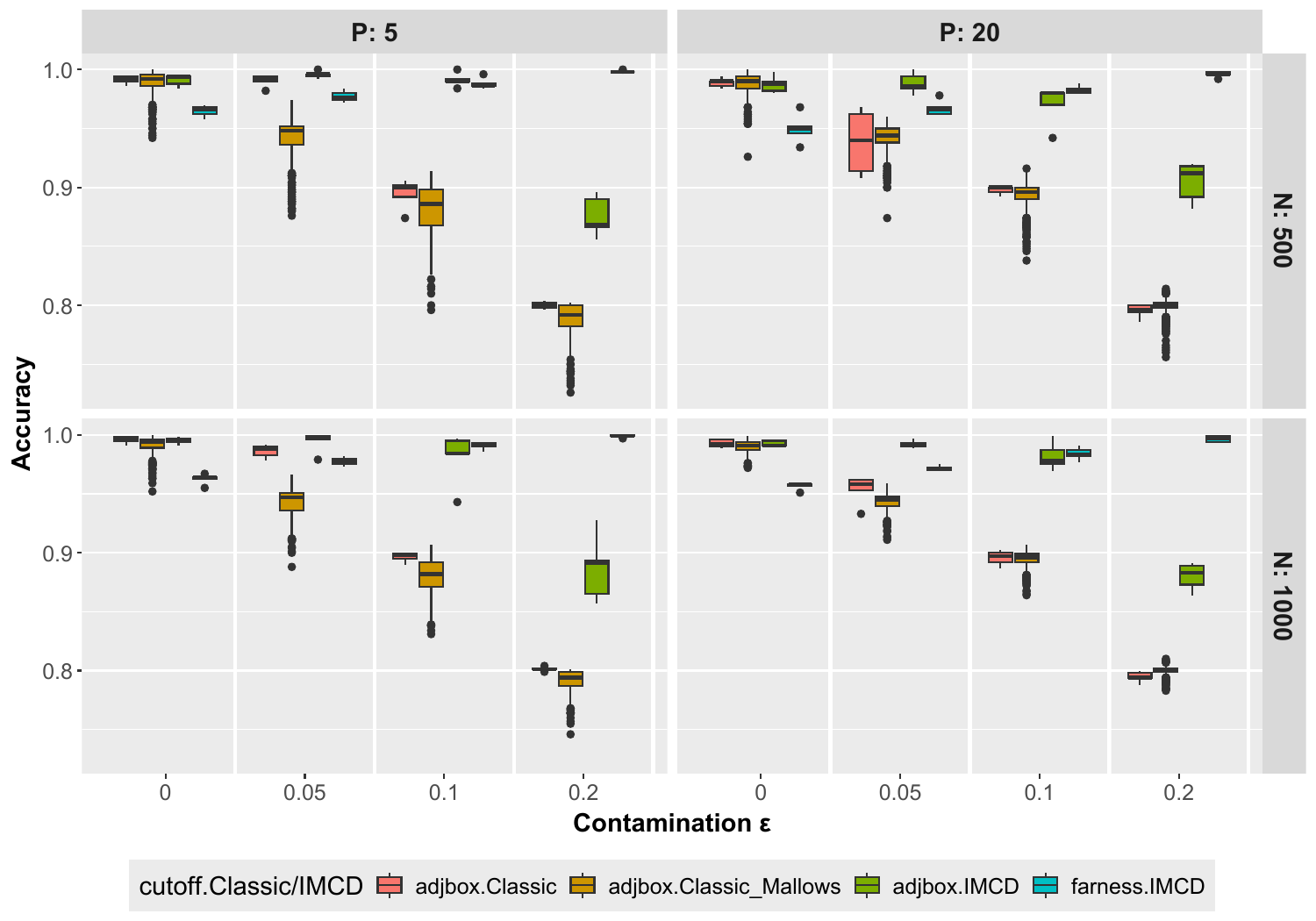}
        \caption{Scenario 1.} 
    \end{subfigure}
    \hfill
    \begin{subfigure}[b]{0.49\textwidth}
        \centering
        \includegraphics[width=\textwidth]{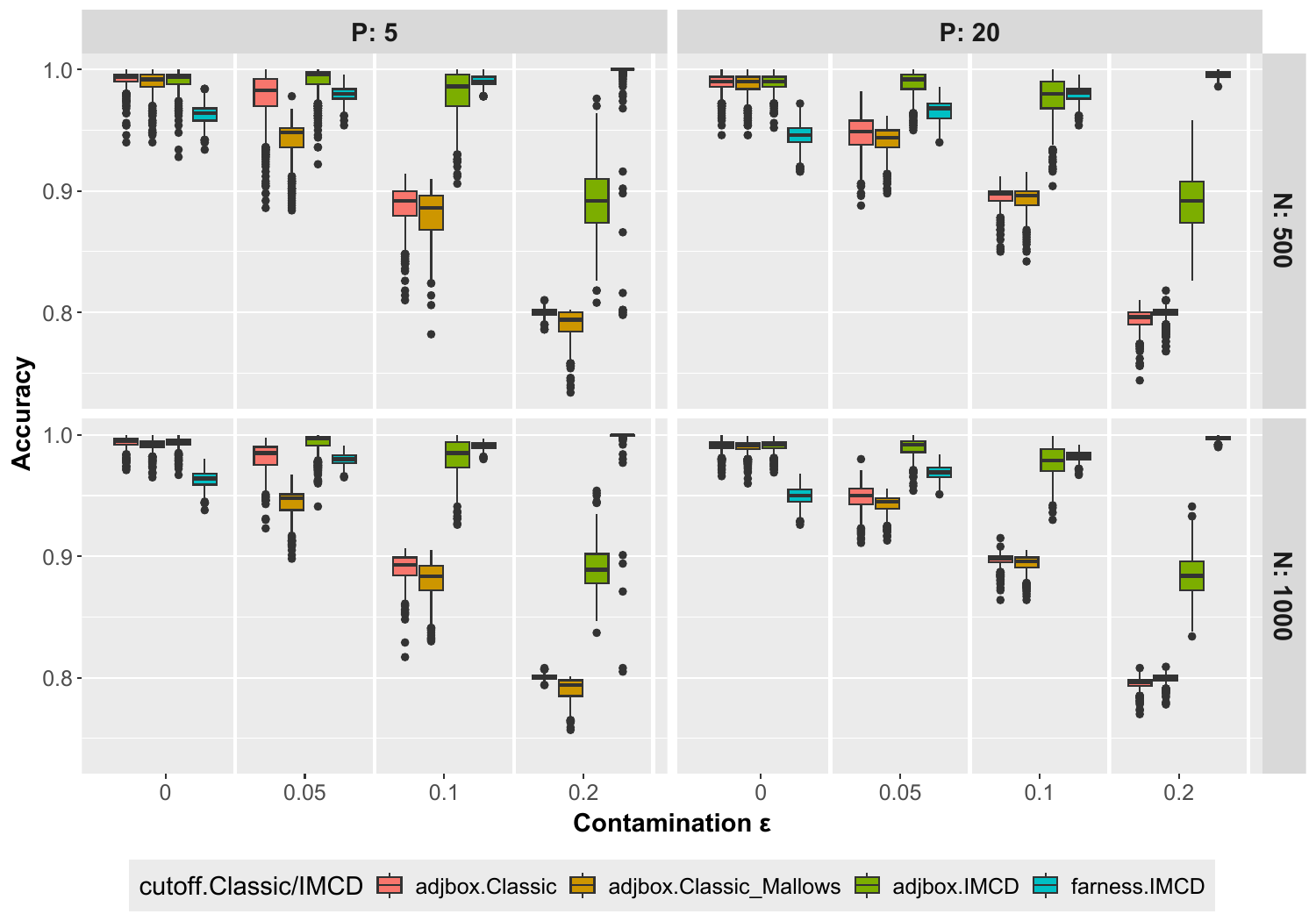}
        \caption{Scenario 4.} 
    \end{subfigure}
    \begin{subfigure}[b]{0.49\textwidth}
        \centering
        \includegraphics[width=\textwidth]{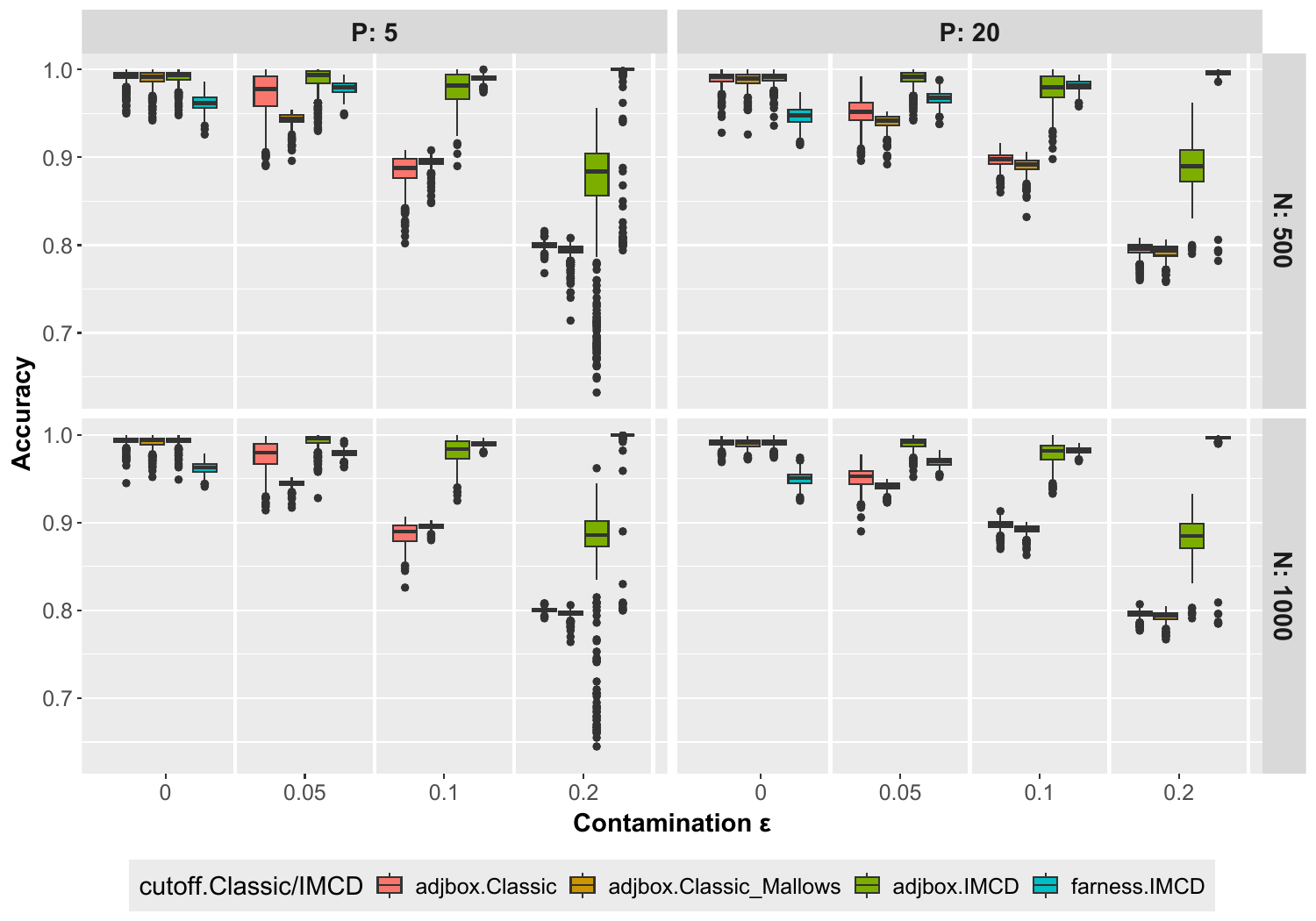}
        \caption{Scenario 2.} 
    \end{subfigure}
    \hfill
    \begin{subfigure}[b]{0.49\textwidth}
        \centering
        \includegraphics[width=\textwidth]{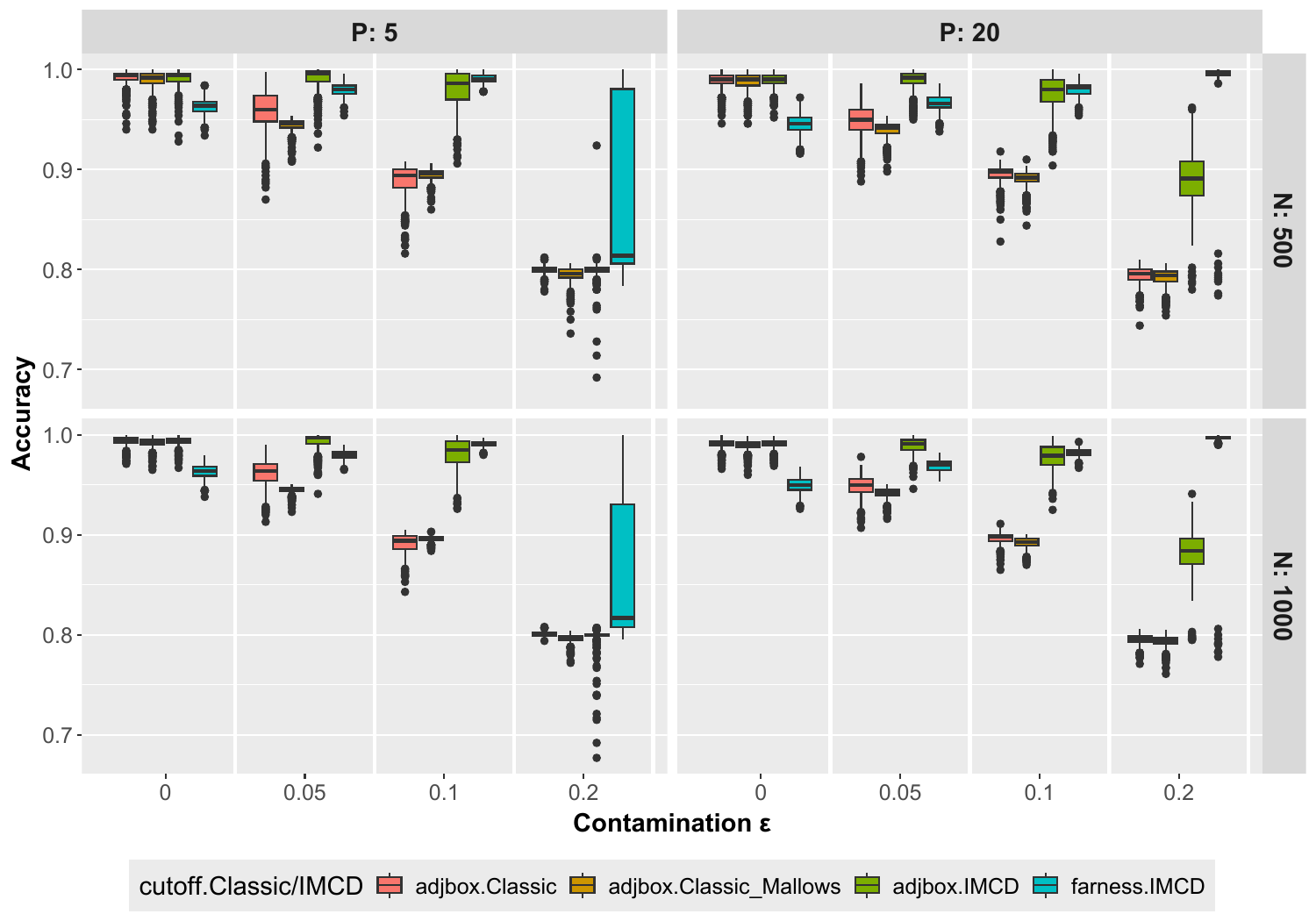}
        \caption{Scenario 5.} 
    \end{subfigure}
    \begin{subfigure}[b]{0.49\textwidth}
        \centering
        \includegraphics[width=\textwidth]{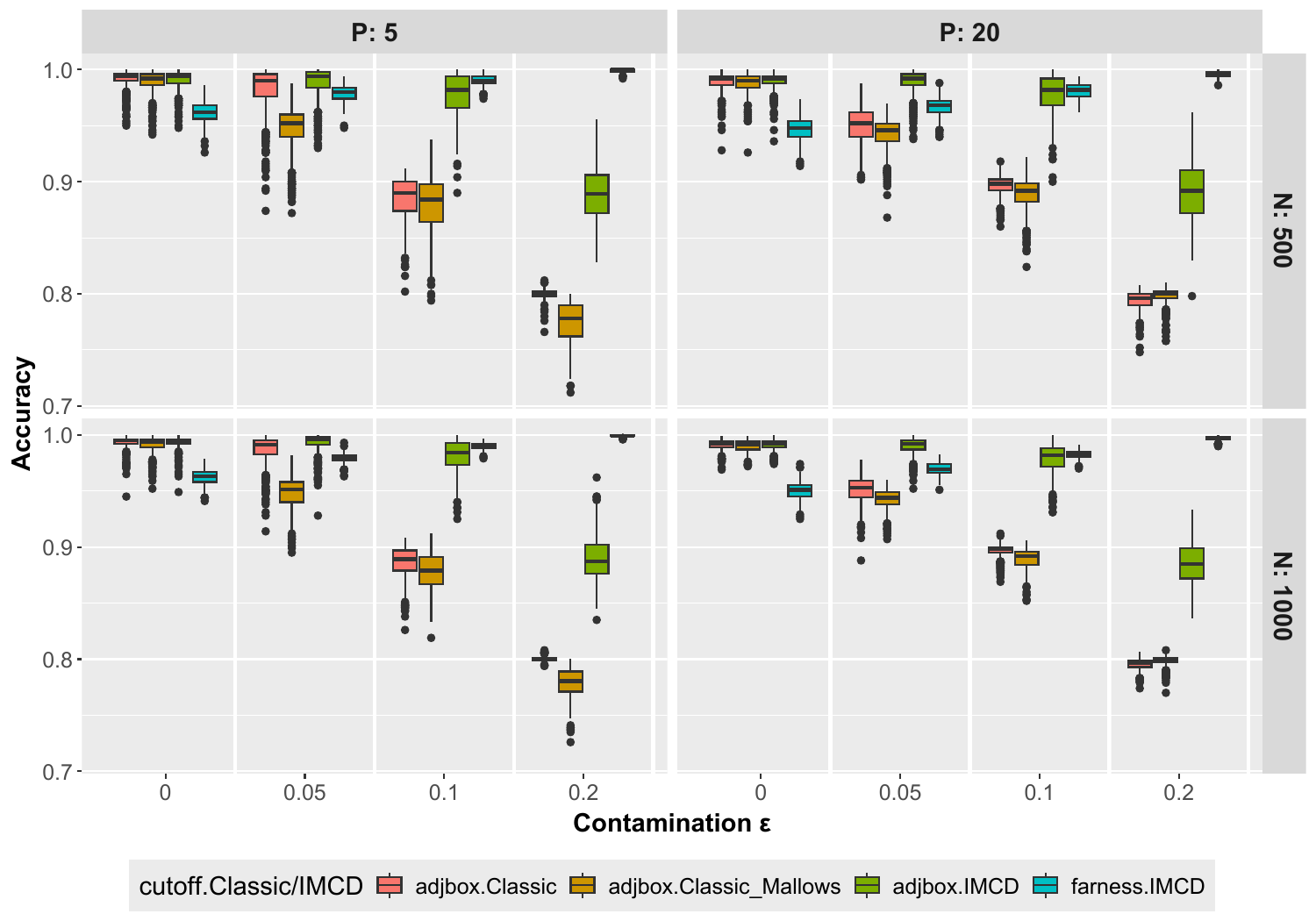}
        \caption{Scenario 3.} 
    \end{subfigure}
    \hfill
    \begin{subfigure}[b]{0.49\textwidth}
        \centering
        \includegraphics[width=\textwidth]{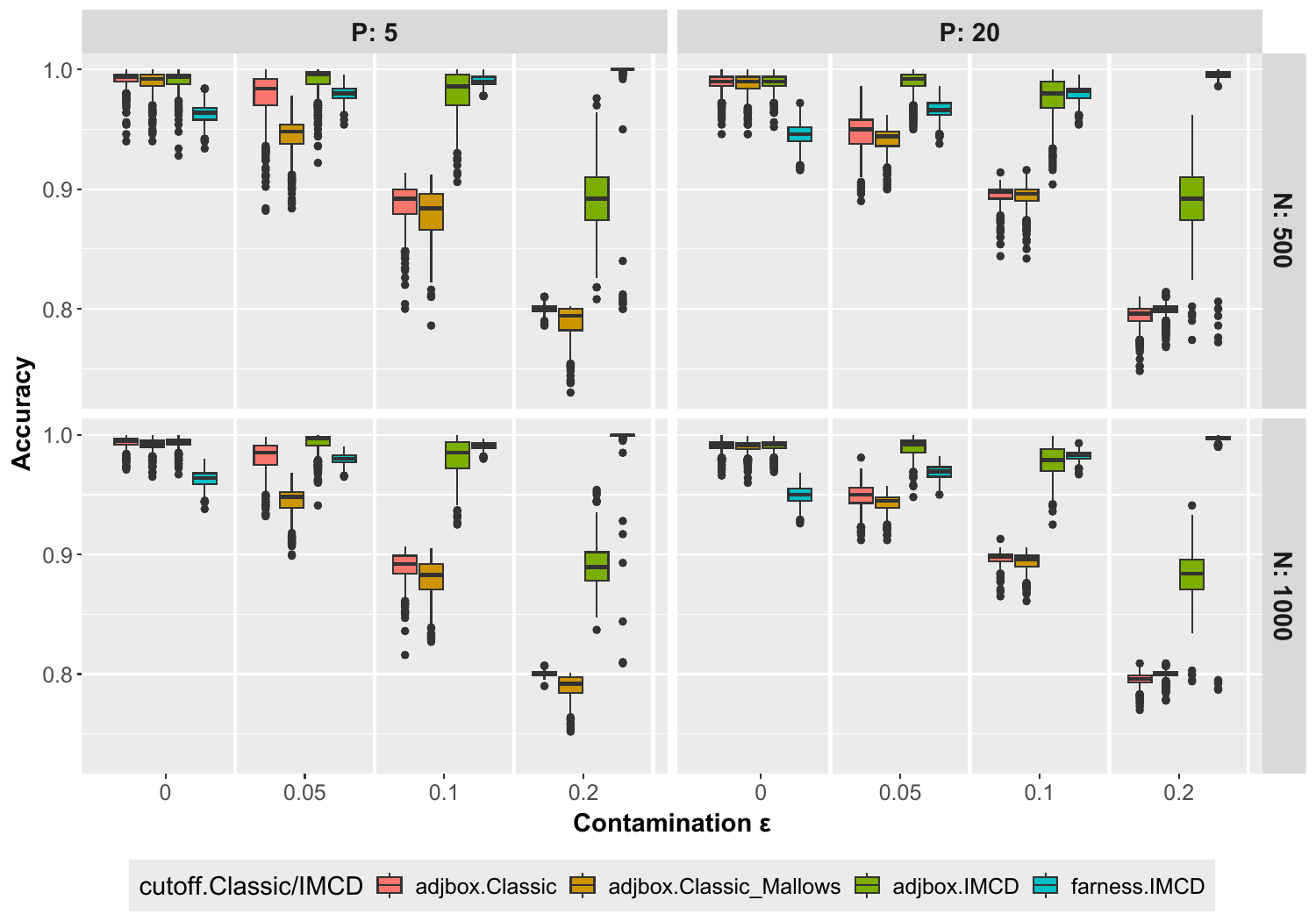}
        \caption{Scenario 6.} 
    \end{subfigure}
    \caption{Boxplots of the accuracy \eqref{eq:accuracy} obtained for the six different scenarios, levels of contamination ($\epsilon$), number of variables ($P$), and sample size ($N$). For each case, we have four outlier detection methods: classic Interval-Mahalanobis distance with adjusted boxplot cutoff (adjbox.Classic), Mallows distance with adjusted boxplot cutoff (adjbox.Classic\_Mallows), robust Interval-Mahalanobis distance with adjusted boxplot (adjbox.IMCD), and farness (farness.IMCD) reweighting/cutoff.}
    \label{fig:accuracy}
\end{figure}

\begin{figure}[ht]
    \centering
    \begin{subfigure}[b]{0.49\textwidth}
        \centering
        \includegraphics[width=\textwidth]{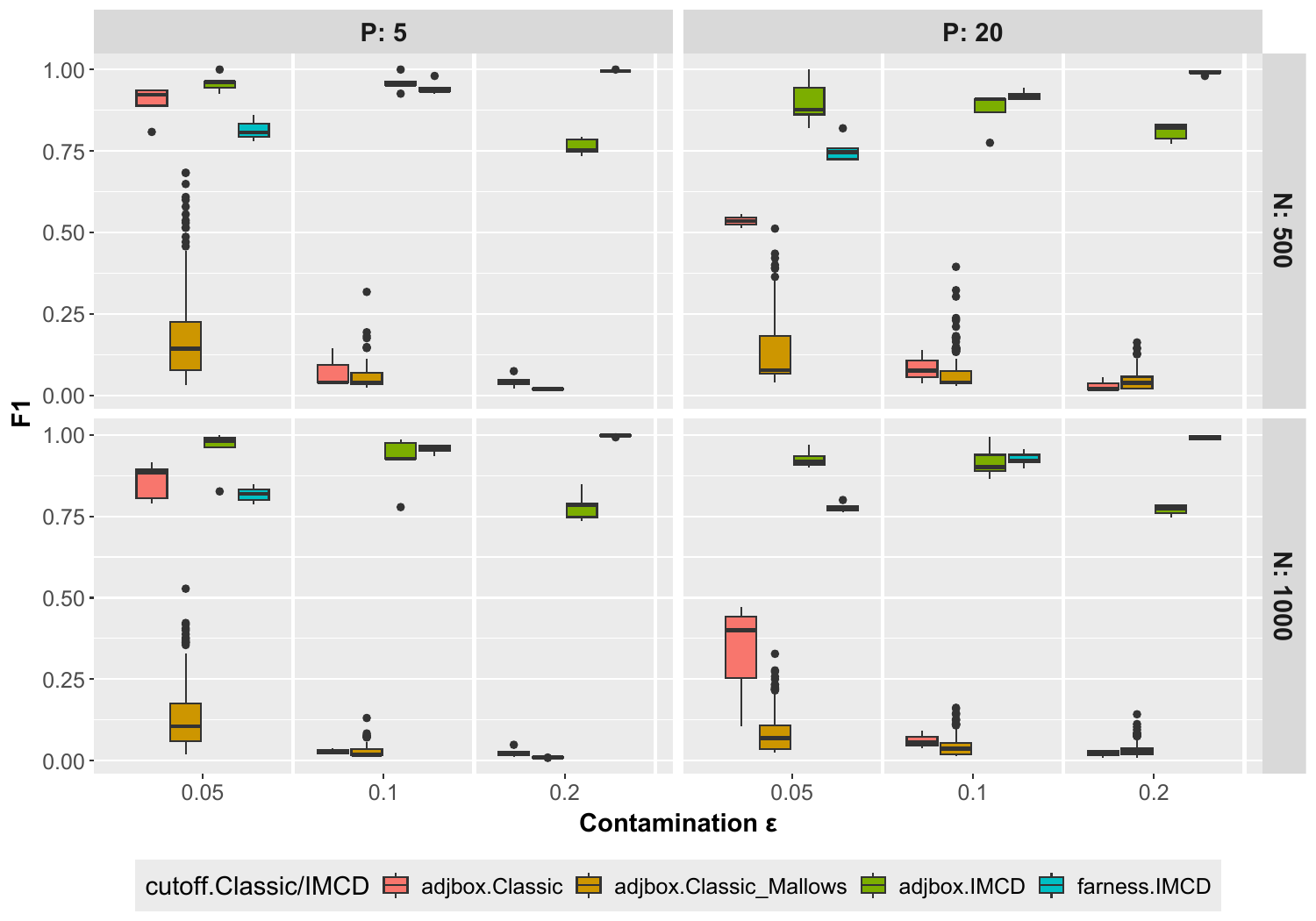}
        \caption{Scenario 1.} 
    \end{subfigure}
    \hfill
    \begin{subfigure}[b]{0.49\textwidth}
        \centering
        \includegraphics[width=\textwidth]{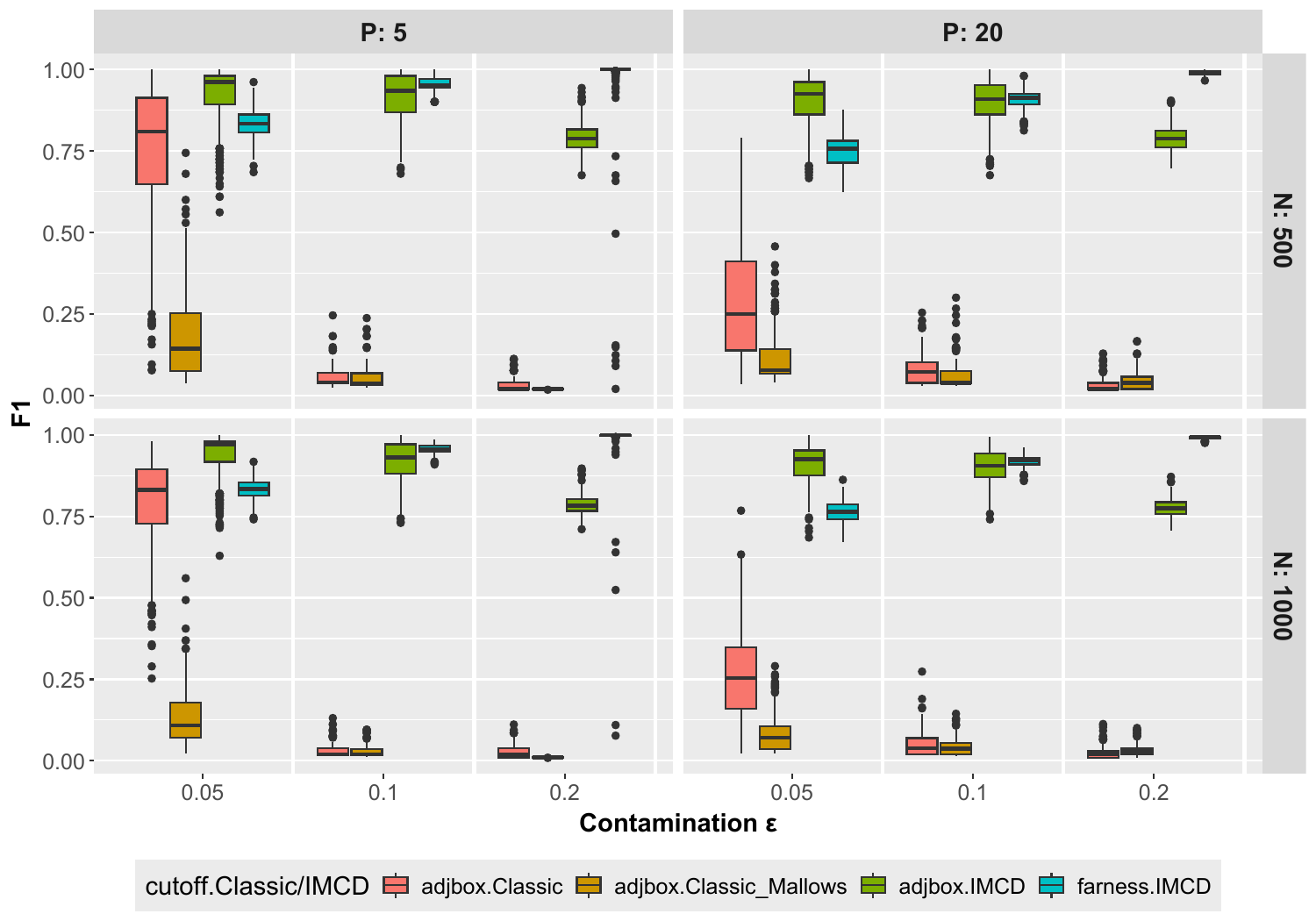}
        \caption{Scenario 4.} 
    \end{subfigure}
    \begin{subfigure}[b]{0.49\textwidth}
        \centering
        \includegraphics[width=\textwidth]{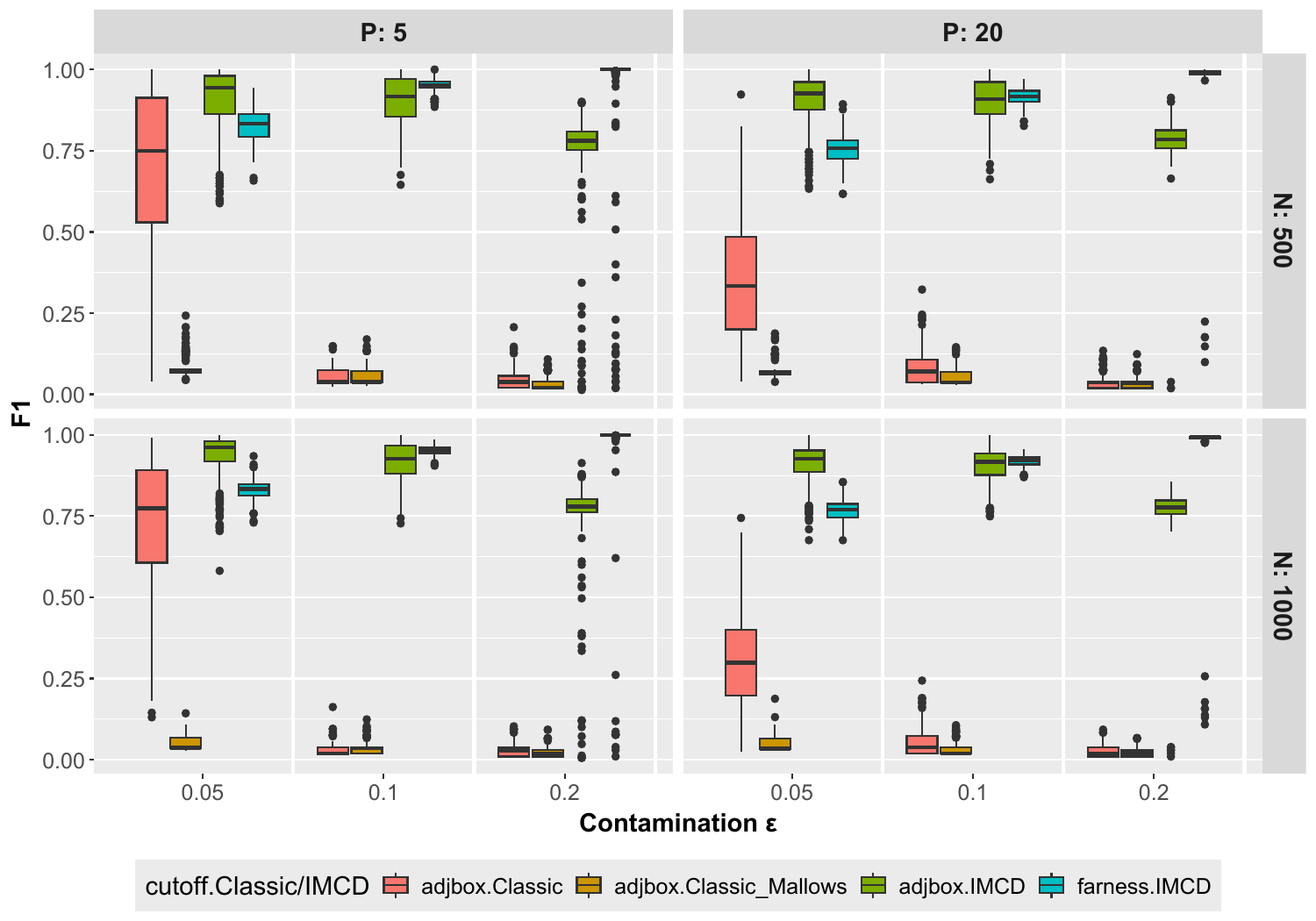}
        \caption{Scenario 2.} 
    \end{subfigure}
    \hfill
    \begin{subfigure}[b]{0.49\textwidth}
        \centering
        \includegraphics[width=\textwidth]{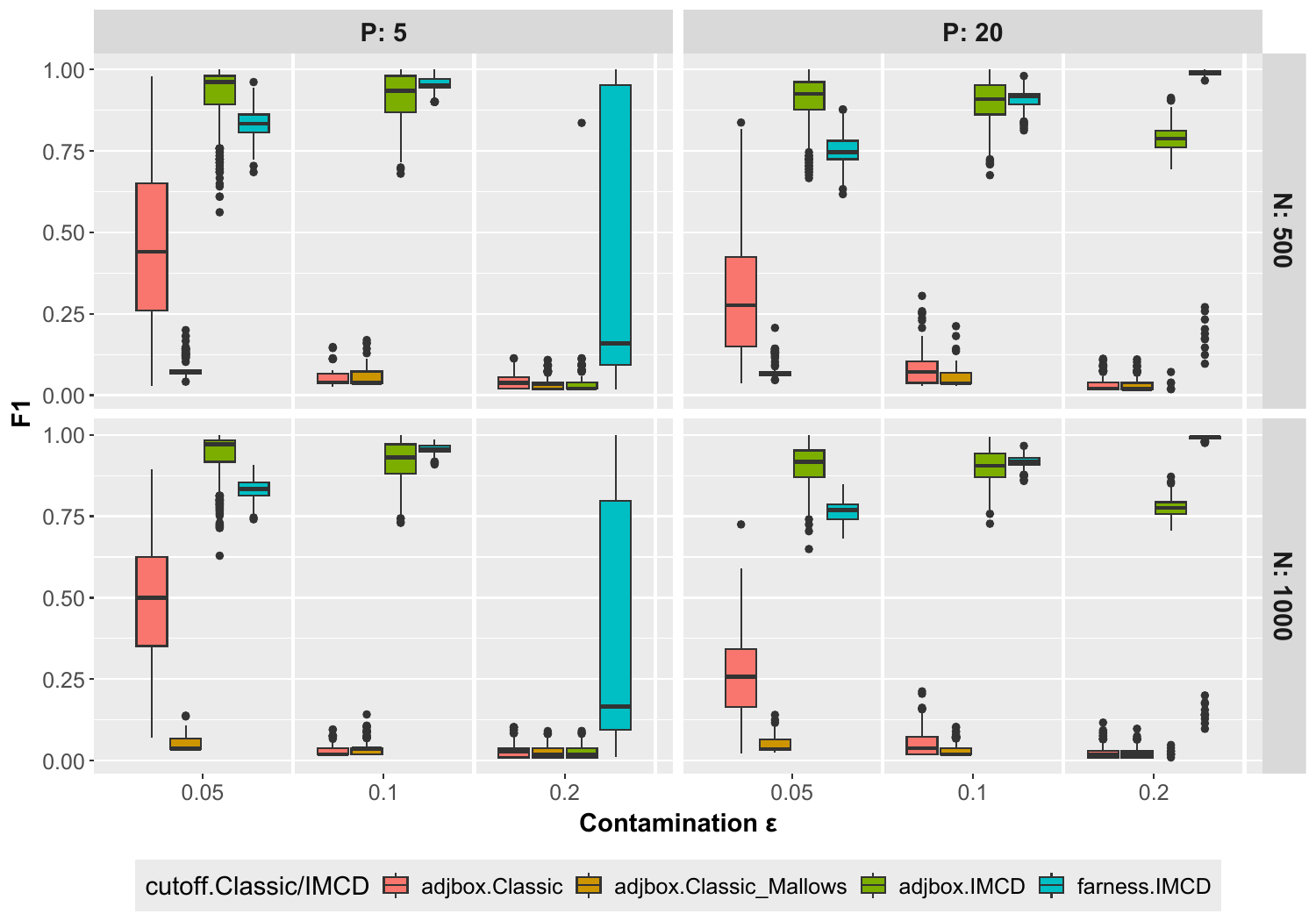}
        \caption{Scenario 5.} 
    \end{subfigure}
    \begin{subfigure}[b]{0.49\textwidth}
        \centering
        \includegraphics[width=\textwidth]{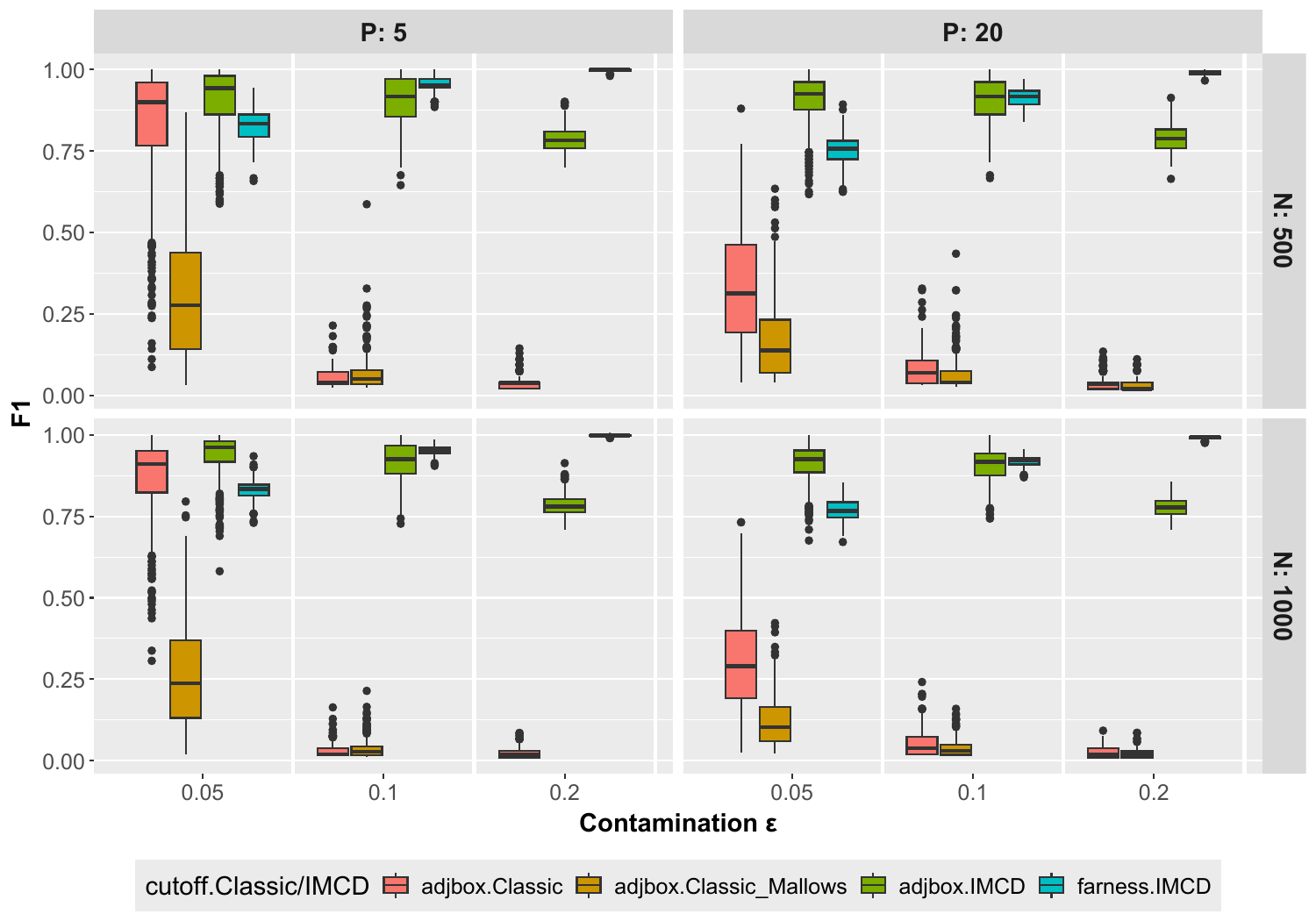}
        \caption{Scenario 3.} 
    \end{subfigure}
    \hfill
    \begin{subfigure}[b]{0.49\textwidth}
        \centering
        \includegraphics[width=\textwidth]{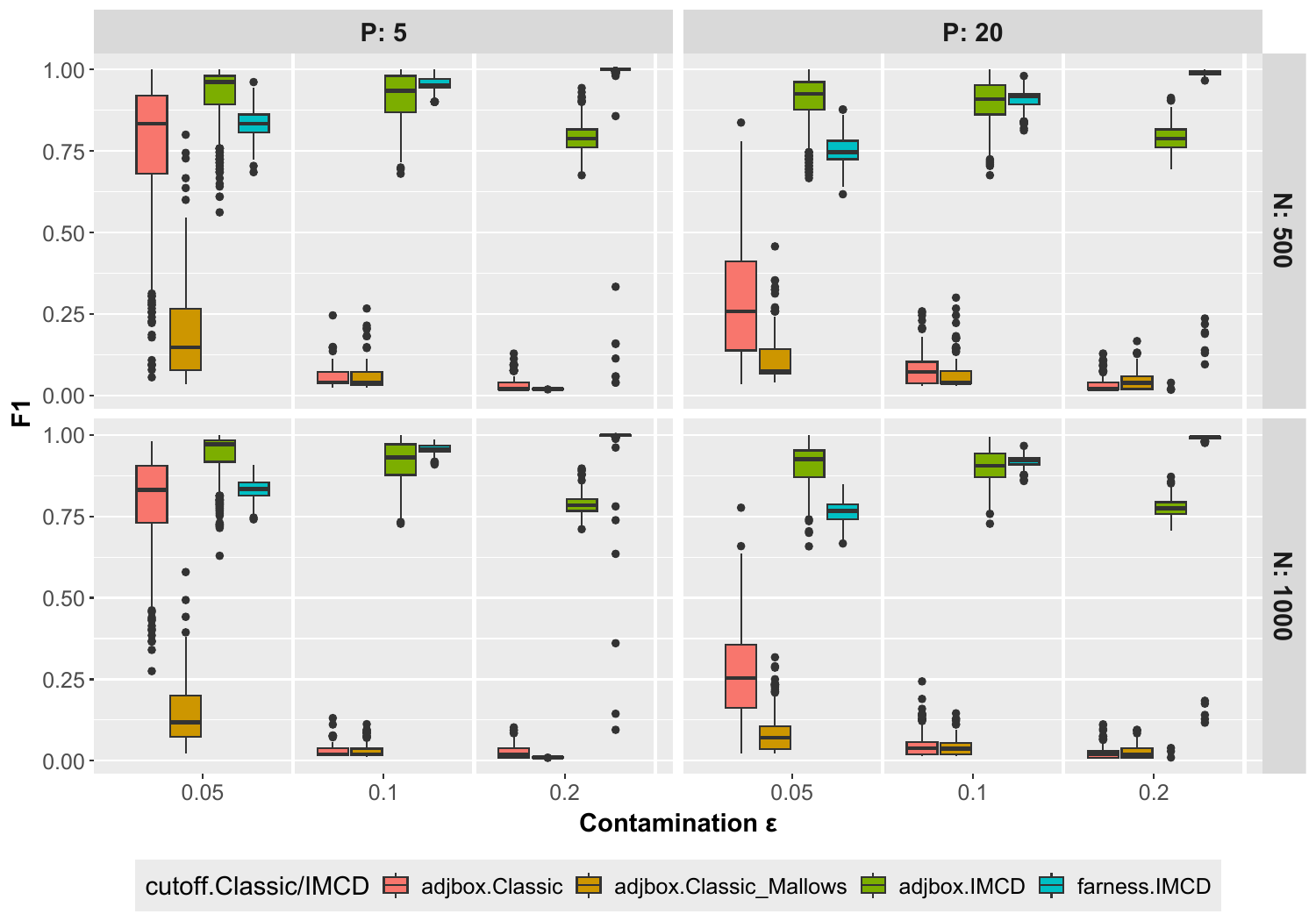}
        \caption{Scenario 6.} 
    \end{subfigure}
    \caption{Boxplots of the F$_1$-score \eqref{eq:f1score} obtained for the six different scenarios, the different levels of contamination ($\epsilon$), number of variables ($P$), and sample size ($N$). For each case, we have four outlier detection methods: classic Interval-Mahalanobis distance with adjusted boxplot cutoff (adjbox.Classic), Mallows distance with adjusted boxplot cutoff (adjbox.Classic\_Mallows), robust Interval-Mahalanobis distance with adjusted boxplot (adjbox.IMCD), and farness (farness.IMCD) reweighting/cutoff.}
    \label{fig:f1score}
\end{figure}

\begin{figure}[ht]
    \centering
    \begin{subfigure}[b]{0.49\textwidth}
        \centering
        \includegraphics[width=\textwidth]{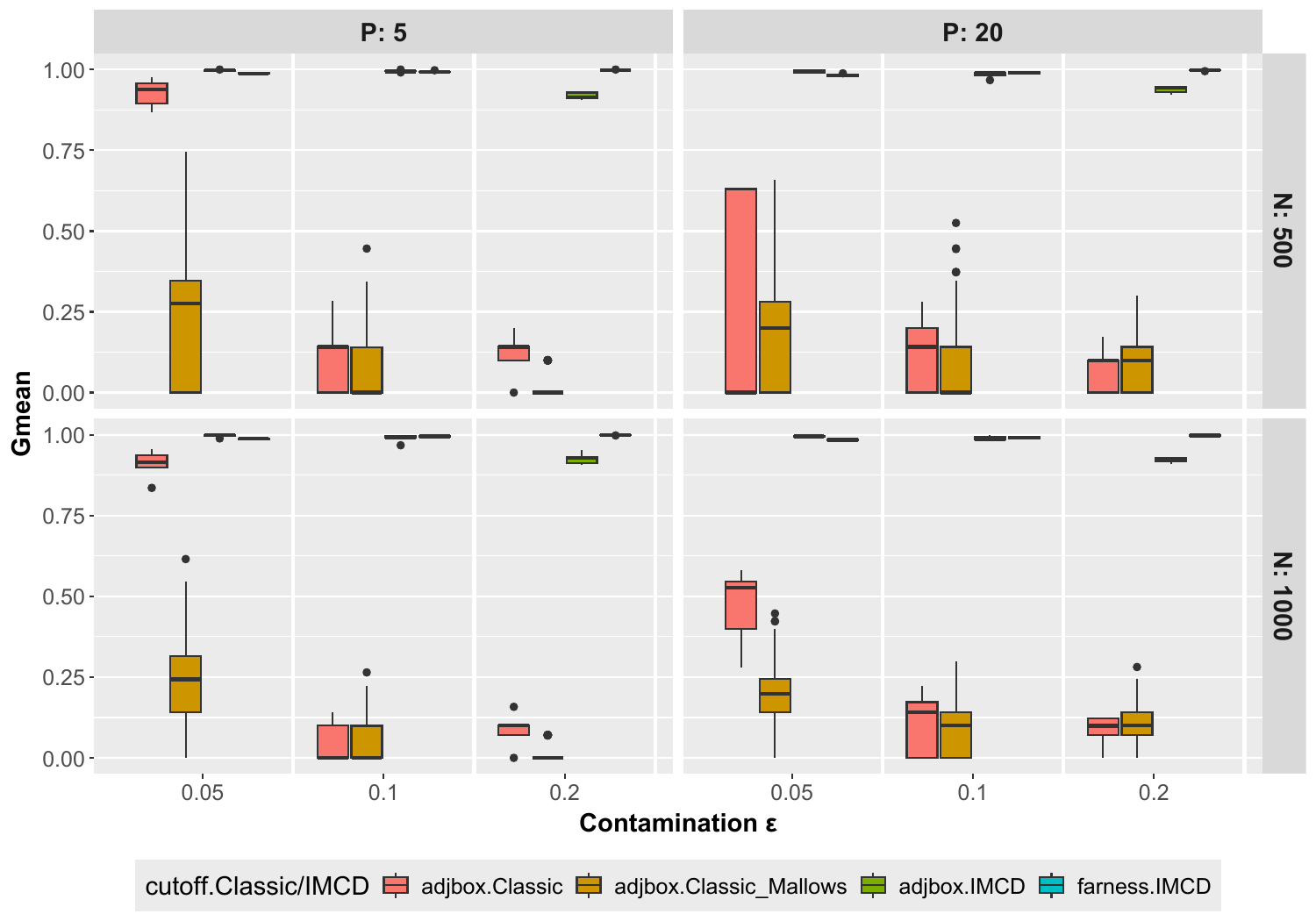}
        \caption{Scenario 1.} 
    \end{subfigure}
    \hfill
    \begin{subfigure}[b]{0.49\textwidth}
        \centering
        \includegraphics[width=\textwidth]{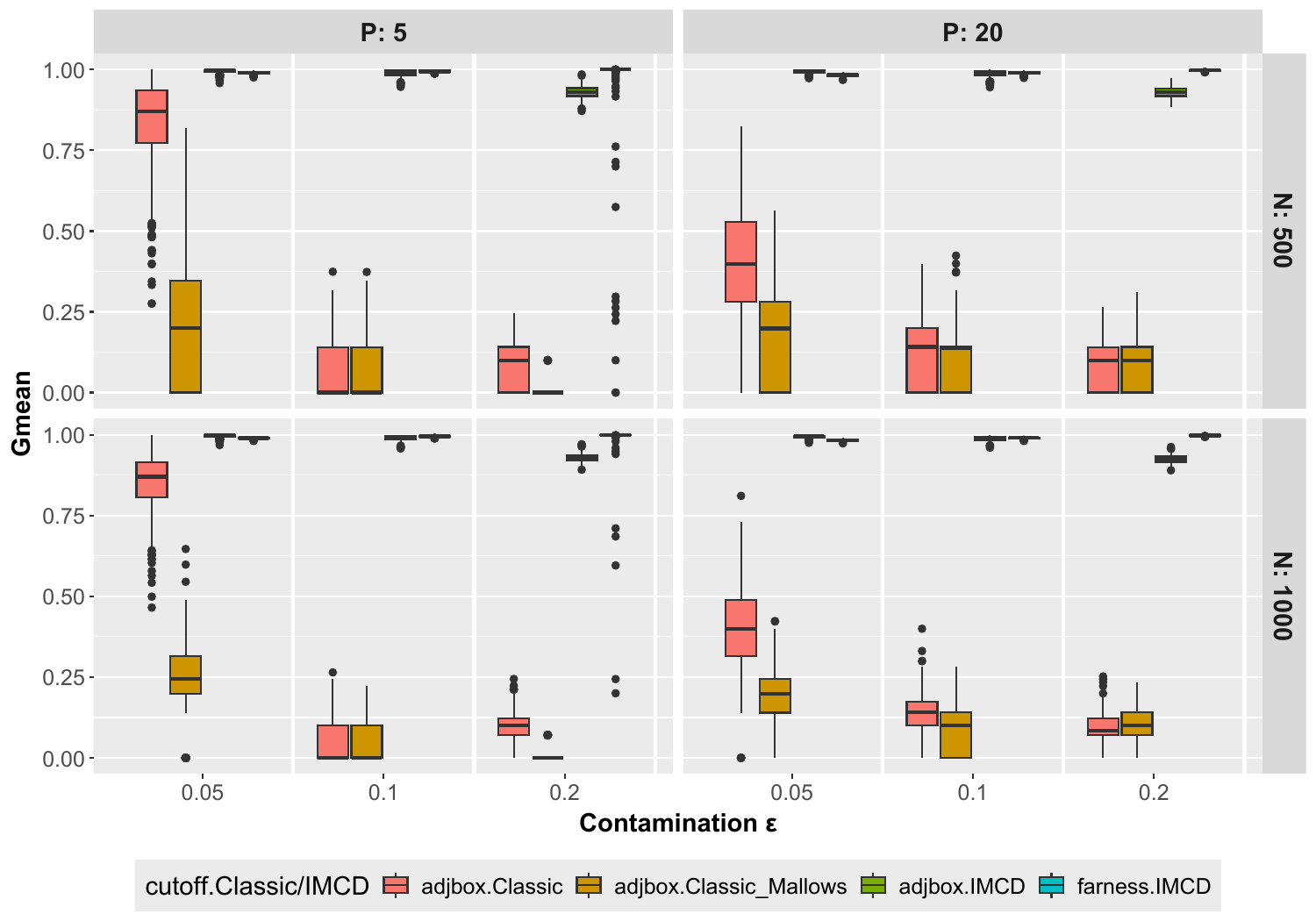}
        \caption{Scenario 4.} 
    \end{subfigure}
    \begin{subfigure}[b]{0.49\textwidth}
        \centering
        \includegraphics[width=\textwidth]{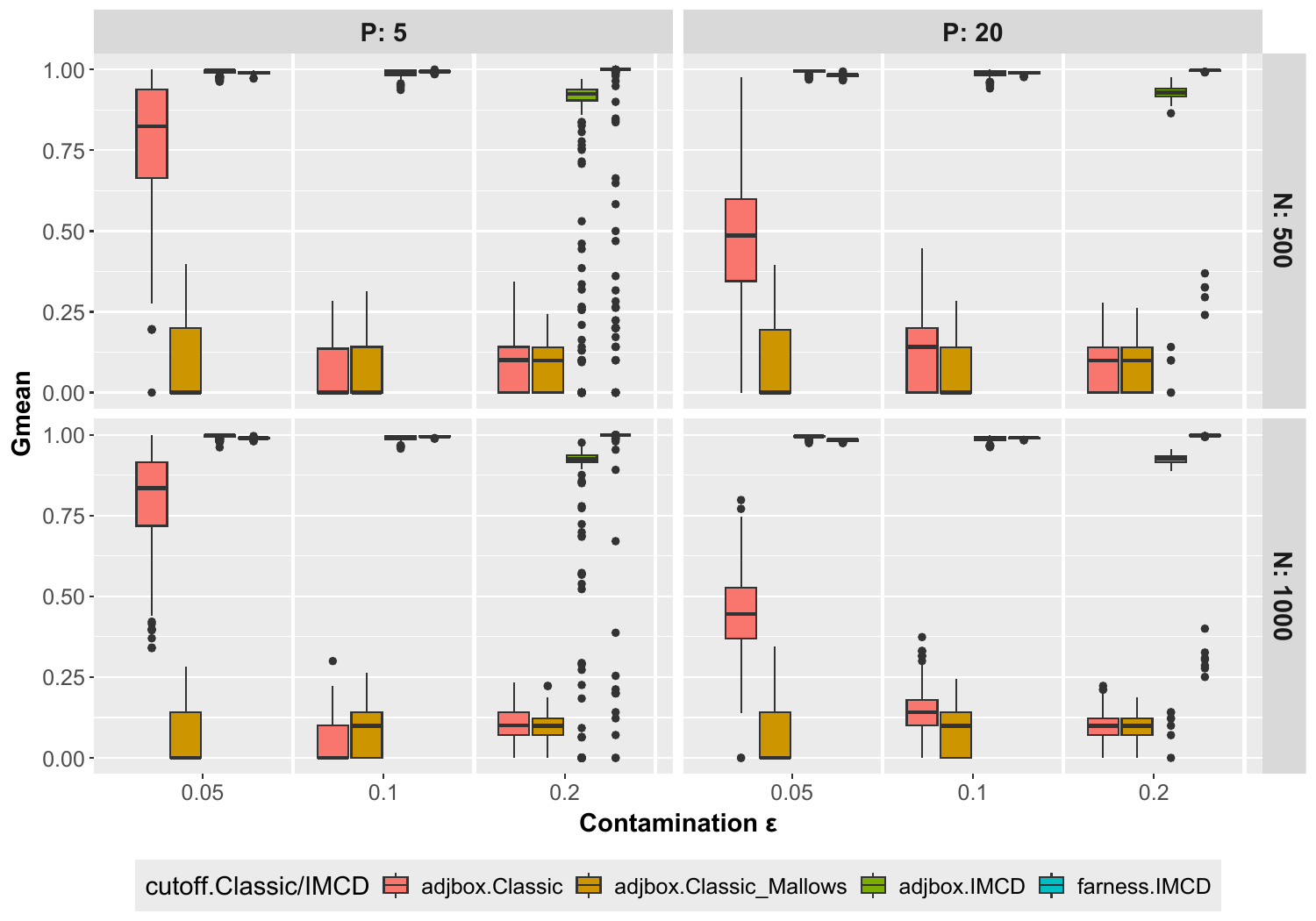}
        \caption{Scenario 2.} 
    \end{subfigure}
    \hfill
    \begin{subfigure}[b]{0.49\textwidth}
        \centering
        \includegraphics[width=\textwidth]{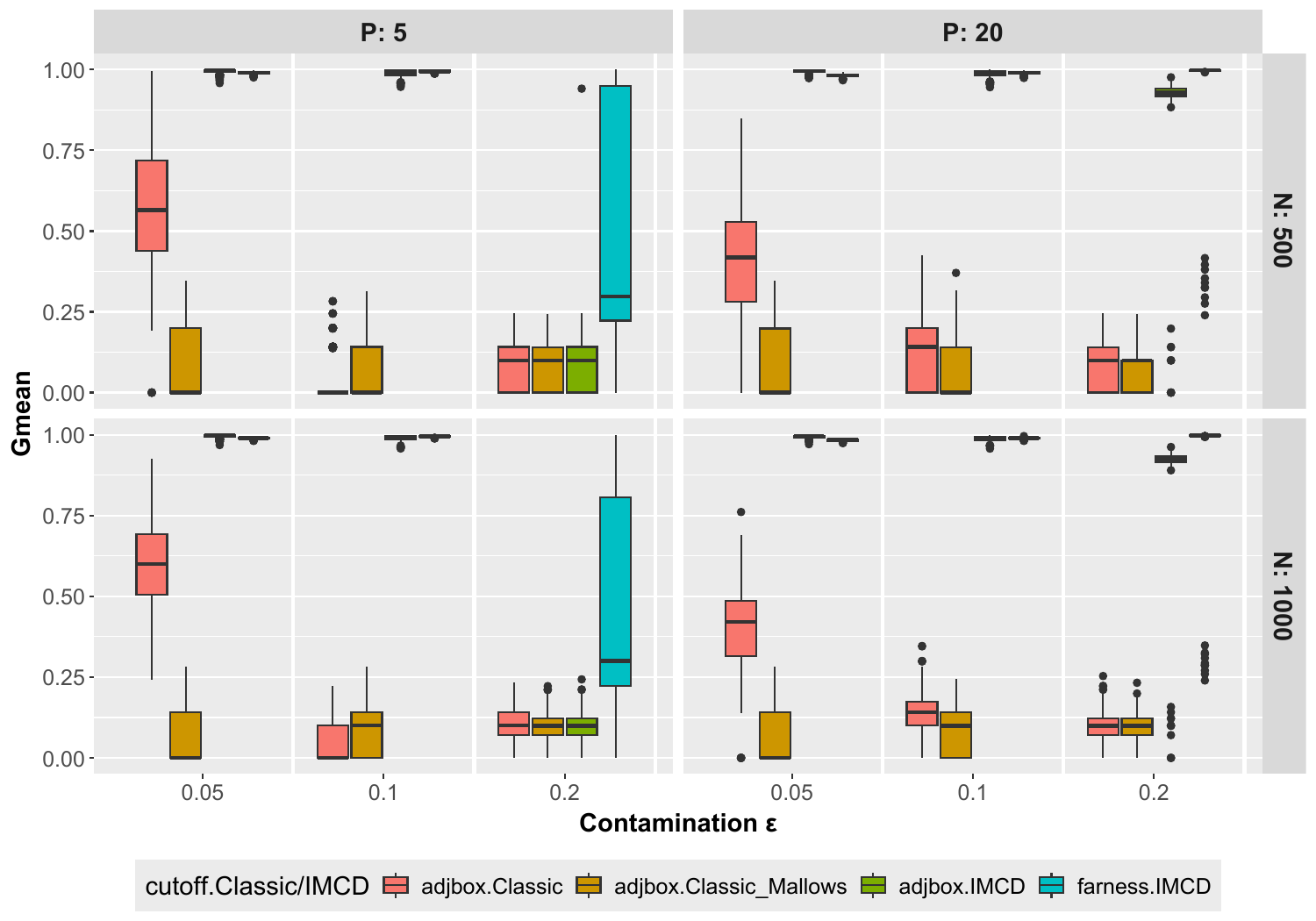}
        \caption{Scenario 5.} 
    \end{subfigure}
    \begin{subfigure}[b]{0.49\textwidth}
        \centering
        \includegraphics[width=\textwidth]{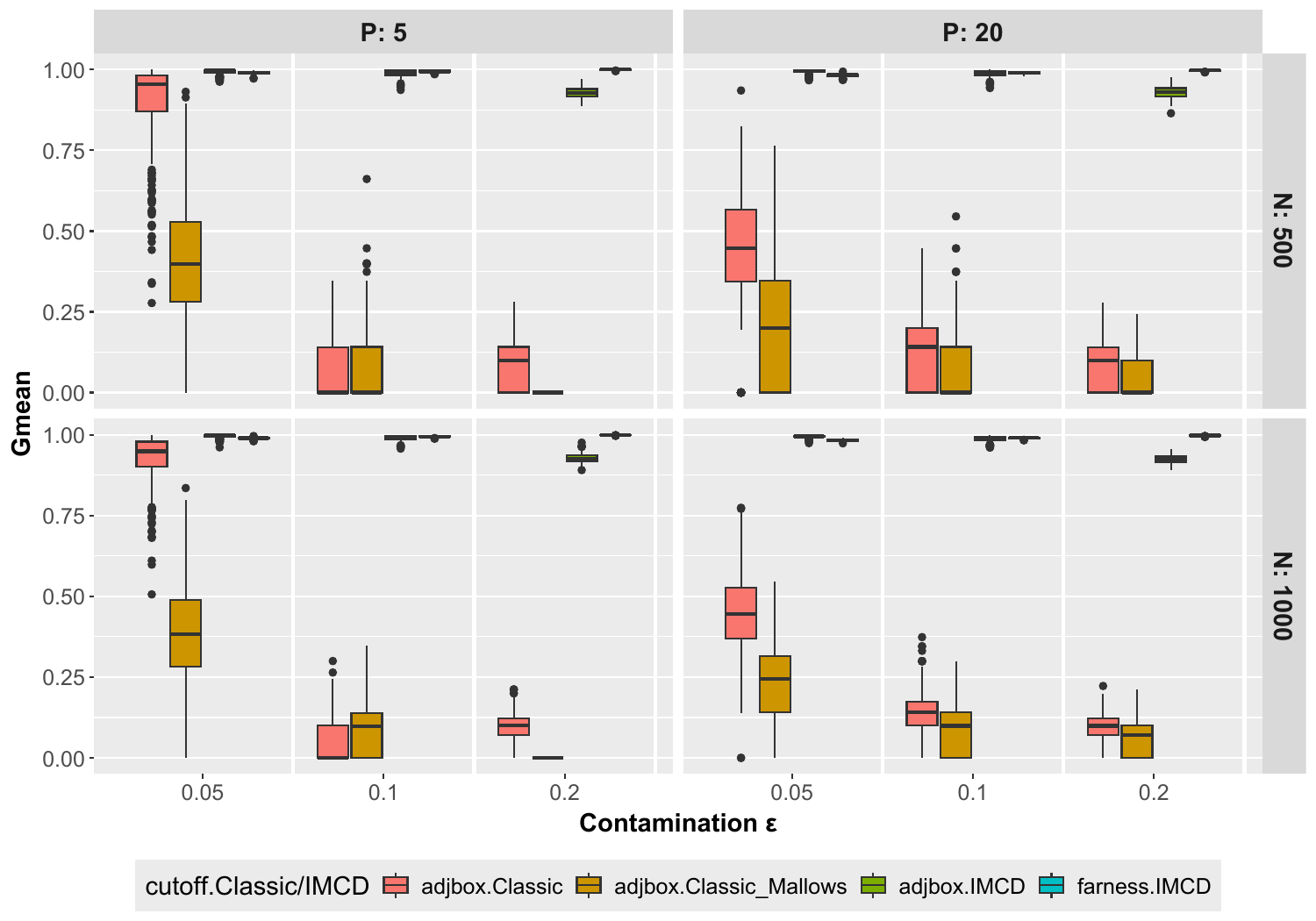}
        \caption{Scenario 3.} 
    \end{subfigure}
    \hfill
    \begin{subfigure}[b]{0.49\textwidth}
        \centering
        \includegraphics[width=\textwidth]{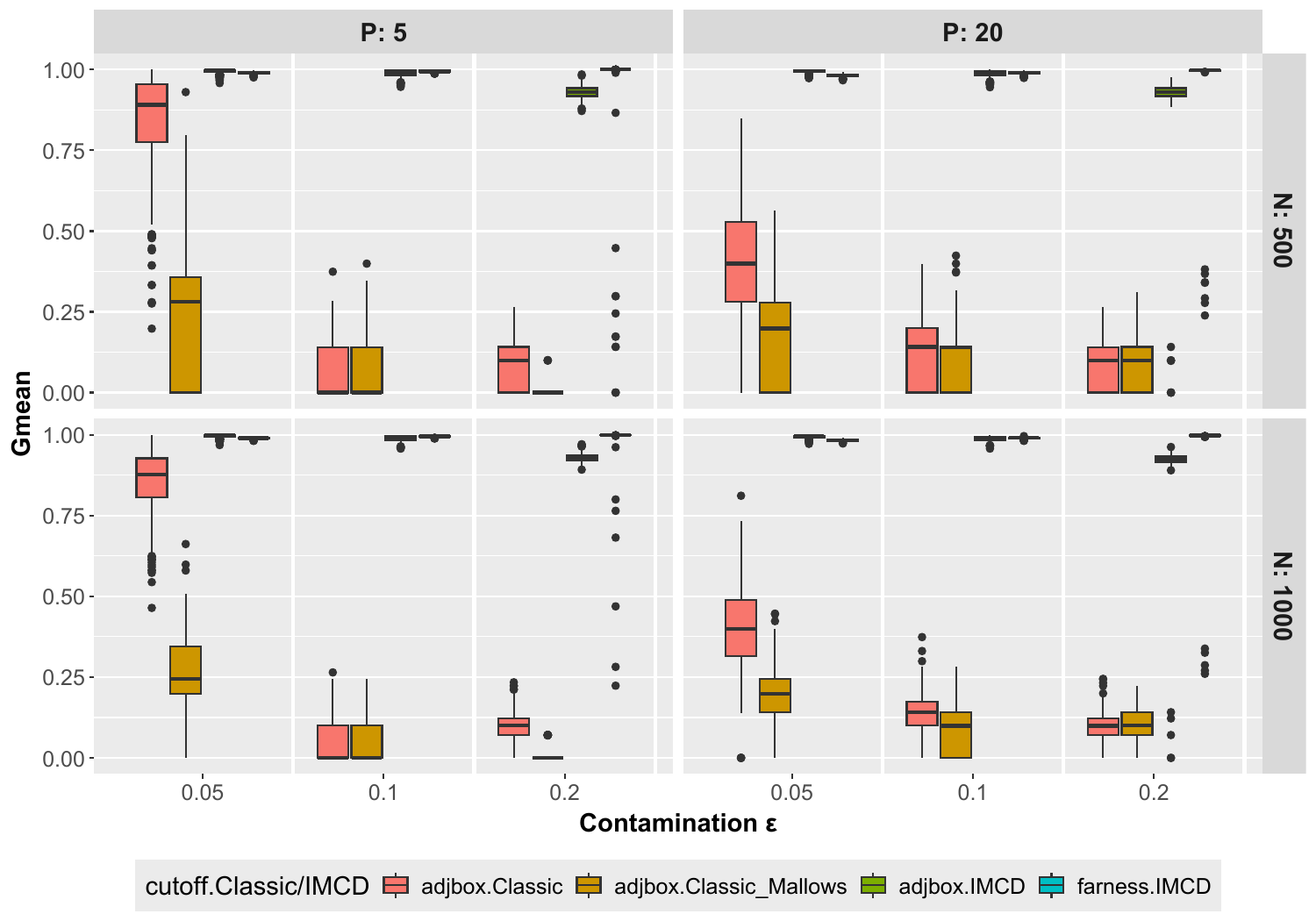}
        \caption{Scenario 6.} 
    \end{subfigure}
    \caption{Boxplots of the G-mean \eqref{eq:gmean} obtained for the six different scenarios, levels of contamination ($\epsilon$), number of variables ($P$), and sample size ($N$). For each case, we have four outlier detection methods: classic Interval-Mahalanobis distance with adjusted boxplot cutoff (adjbox.Classic), Mallows distance with adjusted boxplot cutoff (adjbox.Classic\_Mallows), robust Interval-Mahalanobis distance with adjusted boxplot (adjbox.IMCD), and farness (farness.IMCD) reweighting/cutoff.}
    \label{fig:gmean}
\end{figure}

\end{document}